\documentclass[twoside,11pt]{article}

\usepackage{arxiv,theapa, rawfonts}

\usepackage{latexsym}
\usepackage{amssymb,amsmath}
\usepackage{amsthm}
\usepackage{url}
\usepackage{pifont} 

\usepackage{enumerate}

\usepackage{tikz}
\usetikzlibrary{automata}
\usetikzlibrary{positioning,arrows,fit,calc}
\usetikzlibrary{decorations.pathmorphing}
\usetikzlibrary{shapes}
\usetikzlibrary{backgrounds}
\usepackage{multirow}
\usepackage{booktabs}
\usepackage{xcolor,colortbl}
\usepackage{setspace}

\usepackage{hhline}
\newcommand{\pp}{p}
\newcommand{\perm}[3]{\fu^{{#1}_{#2}^{#3}}}

\newcommand{\sub}{\mathsf{sub}}

\newcommand{\tp}{\boldsymbol{\tau}}

\newcommand{\yes}{$\mathsf{yes}$}
\newcommand{\no}{$\mathsf{no}$}
\newcommand{\RPR}{\mathsf{RPR}}

\newcommand{\SA}{\mathfrak S_\Abox}
\renewcommand{\rq}{\boldsymbol Q}

\newcommand{\Z}{\mathbb{Z}}

\newcommand{\A}{\ensuremath{\mathfrak A}}
\newcommand{\I}{\ensuremath{\mathcal{I}}}

\newcommand{\TO}{\mathcal{O}}
\newcommand{\Abox}{\mathcal{A}}

\newcommand{\OWLQL}{\textsl{OWL\,2\,QL}}
\newcommand{\q}{{\boldsymbol{q}}}

\newcommand{\fu}{\delta}

\newcommand{\Xallop}{^{\smash{\Box\raisebox{1pt}{$\scriptscriptstyle\bigcirc$}}}}
\newcommand{\Xnext}{^{\smash{\raisebox{1pt}{$\scriptscriptstyle\bigcirc$}}}}

\newcommand{\Xbox}{^{\smash{\Box}}}

\newcommand{\bool}{\textit{bool}}
\newcommand{\horn}{\textit{horn}}

\newcommand{\krom}{\textit{krom}}
\newcommand{\core}{\textit{core}}

\newcommand{\lang}{\mathcal{L}}

\renewcommand{\L}{{\boldsymbol{L}}}
\newcommand{\Type}{{\boldsymbol{T}}}
\newcommand{\sig}{\mathsf{sig}}

\newcommand{\lit}{\mathsf{lit}}

\newcommand{\frag}{{\boldsymbol{c}}}
\newcommand{\op}{{\boldsymbol{o}}}
\newcommand{\nxt}{{\ensuremath\raisebox{0.25ex}{\text{\scriptsize$\bigcirc$}}}}
\newcommand{\Rnext}{\nxt_{\!\scriptscriptstyle F}}
\newcommand{\Lnext}{\nxt_{\!\scriptscriptstyle P}}

\newcommand{\Rdiamond}{\Diamond_{\!\scriptscriptstyle F}}
\newcommand{\Ldiamond}{\Diamond_{\!\scriptscriptstyle P}}

\newcommand{\Rbox}{\rule{0pt}{1.4ex}\Box_{\!\scriptscriptstyle F}}
\newcommand{\Lbox}{\rule{0pt}{1.4ex}\Box_{\!\scriptscriptstyle P}}

\newcommand{\avec}[1]{\boldsymbol{#1}}

\newcommand{\LTL}{\textsl{LTL}}
\newcommand{\DL}{\textsl{DL-Lite}}
\newcommand{\tem}{\mathsf{tem}}

\newcommand{\NCo}{{{\ensuremath{\textsc{NC}^1}}}}
\newcommand{\ACz}{{\ensuremath{\textsc{AC}^0}}}
\newcommand{\ACC}{{\ensuremath{\textsc{ACC}^0}}}

\newcommand{\MOD}{\ensuremath{\mathsf{MOD}}}
\newcommand{\coNP}{{\ensuremath{\textsc{coNP}}}}
\newcommand{\ExpSpace}{{\ensuremath{\textsc{ExpSpace}}}}
\newcommand{\PSpace}{{\ensuremath{\textsc{PSpace}}}}
\newcommand{\NL}{\textsc{NL}}
\newcommand{\NP}{\textsc{NP}}

\newcommand{\FO}{\mathsf{FO}}
\newcommand{\MFO}{\textup{MFO}}

\newcommand{\qa}{q_{\textit{acc}}}
\newcommand{\B}{\mathsf{b}}
\newcommand{\M}{\boldsymbol{M}}
\newcommand{\conf}{\mathfrak c}

\newcommand{\aut}{\mathfrak B}

\newcommand{\simm}{\sim}

\newcommand{\Sg}{\mathfrak S}
\newcommand{\unit}{identity}
\newcommand{\fum}{\tilde{\fu}}
\newcommand{\ord}{o}
\newcommand{\id}{\mathsf{id}}
\newcommand{\G}{\mathfrak G}
\newcommand{\Qr}{Q^r}
\newcommand{\simclass}[1]{#1/_{\mathop{\simm}}}
\newcommand{\Amin}{\A_{\L}}

\newcommand{\Tmabc}{\Gamma}
\newcommand{\Aiabc}{\Sigma}
\newcommand{\Tmtran}{\gamma}
\newcommand{\Aabc}{\Sigma_+}
\newcommand{\ppn}{N}
\newcommand{\tst}{\textit{tr}}

\newcommand{\nm}[1]{\textit{#1}}

\tikzset{
  basic box/.style = {
    shape = rectangle,
    align = center,
    draw  = #1,
    rounded corners},
  header node/.style = {
    font          = \strut\Large\ttfamily,
    text depth    = +0pt,
    fill          = white,
    draw},
  header/.style = {%
    inner ysep = +1.5em,
    append after command = {
      \pgfextra{\let\TikZlastnode\tikzlastnode}
      node [header node] (header-\TikZlastnode) at (\TikZlastnode.north) {#1}
    }
  },
  hv/.style = {to path = {-|(\tikztotarget)\tikztonodes}},
  vh/.style = {to path = {|-(\tikztotarget)\tikztonodes}},
  fat blue line/.style = {ultra thick, blue}
}

\newtheorem{theorem}{Theorem}

\newtheorem{example}[theorem]{Example}
\newtheorem{proposition}[theorem]{Proposition}
\newtheorem{lemma}[theorem]{Lemma}

\newtheorem{remark}[theorem]{Remark}

\begin{document}

\date{ }

\title{Deciding FO-rewritability of Regular Languages and Ontology-Mediated Queries in Linear Temporal Logic}

\author{Agi Kurucz\\
		King's College London, U.K.\\
		\texttt{agi.kurucz@kcl.ac.uk}\\
		\And Vladislav Ryzhikov\\
		Birkbeck, University of London, U.K.\\
		\texttt{vlad@dcs.bbk.ac.uk}\\
		\And Yury Savateev\\
		University of Southampton, U.K.\\
		\texttt{y.savateev@soton.ac.uk}\\
		\And Michael Zakharyaschev\\
		Birkbeck, University of London, U.K.\\
		\texttt{michael@dcs.bbk.ac.uk}
	}


	\renewcommand{\headeright}{}
	\renewcommand{\undertitle}{Technical Report}

\maketitle


\begin{abstract}
Our concern is the problem of determining the data complexity of answering an ontology-mediated query (OMQ) formulated in linear temporal logic \LTL{} over $(\mathbb Z,<)$ and deciding whether it is rewritable to an $\FO(<)$-query, possibly with some extra predicates.
First, we observe that, in line with the circuit complexity and FO-definability of regular languages, OMQ answering in \ACz{}, $\textsc{ACC}^0$ and $\NCo$ coincides  with $\FO(<,\equiv)$-rewritability using unary predicates \mbox{$x \equiv 0\, (\text{mod}\ n)$}, $\FO(<,\mathsf{MOD})$-rewritability,  and  $\FO(\RPR)$-rewritability using relational primitive recursion, respectively.
We prove that, similarly to known \PSpace-completeness of recognising $\FO(<)$-definability of regular languages, deciding \mbox{$\FO(<,\equiv)$}- and \mbox{$\FO(<,\mathsf{MOD})$}-definability is also \PSpace-complete (unless $\ACC = \NCo$).
We then use this result to show that deciding $\FO(<)$-, \mbox{$\FO(<,\equiv)$}- and $\FO(<,\mathsf{MOD})$-rewritability of \LTL{} OMQs is \ExpSpace-complete, and that these problems become \PSpace-complete for OMQs with a linear Horn ontology and an atomic query, and also a positive query in the cases of $\FO(<)$- and \mbox{$\FO(<,\equiv)$}-rewritability.
Further, we consider $\FO(<)$-rewritability of OMQs with a binary-clause ontology and identify OMQ classes, for which deciding it is \PSpace-, $\Pi_2^p$- and \coNP-complete.
\end{abstract}



\section{Introduction}\label{intro}

\paragraph*{\bf\em Motivation.}
The problem we consider in this paper originates in the area of \emph{ontology-based data access} (\emph{OBDA}) to temporal data. The aim of the OBDA paradigm~\cite{PLCD*08,DBLP:conf/ijcai/XiaoCKLPRZ18} and systems such as Mastro\footnote{\url{https://www.obdasystems.com}} or Ontop\footnote{\url{https://ontopic.biz}} is to facilitate management and integration of possibly incomplete and heterogeneous data by providing the user with a view of the data through the lens of a description logic (DL) ontology. As a result, the user can think of the data as a virtual knowledge graph~\cite{DBLP:journals/dint/XiaoDCC19}, $\Abox$, whose labels---unary and binary predicates supplied by an  ontology, $\TO$---are the only thing to know when formulating queries, $\varkappa$. Ontology-mediated queries (OMQs) $\q = (\TO,\varkappa)$ are supposed to be answered over $\Abox$ under the open-world semantics (taking account of all models of $\TO$ and $\Abox$), which can be prohibitively complex. So the key to practical OBDA is ensuring first-order rewritability of $\q$ (aka \emph{boundedness} in the datalog literature~\cite{Abitebouletal95}), which reduces open-world reasoning to evaluating an FO-formula over $\Abox$. The W3C standard ontology language \OWLQL{} for OBDA is based on the \DL{} family of DL~\cite{CDLLR07,ACKZ09}, which uniformly guarantees FO-rewritability of all \OWLQL{} OMQs with a conjunctive query. Other ontology languages with this feature include various dialects of tgds~\cite<e.g.,>{DBLP:journals/ai/BagetLMS11,DBLP:journals/ai/CaliGP12,DBLP:conf/datalog/CiviliR12}. 
However, this \emph{uniform} approach to ensuring FO-rewritability inevitably imposes severe syntactical restrictions on ontology languages, making them  rather inexpressive.

Theory and practice of OBDA have revived the interest in the \emph{non-uniform} approach, where the problem is to decide whether a given OMQ, formulated in some expressive language, is FO-rewritable. This problem 
was thoroughly investigated  in the 1980--90s for datalog queries~\cite<e.g.,>{DBLP:conf/pods/Vardi88,DBLP:journals/algorithmica/UllmanG88,DBLP:conf/stoc/CosmadakisGKV88,DBLP:journals/jacm/AfratiP93,DBLP:conf/lics/Marcinkowski96}. The data complexity and rewritability of OMQs in various DLs and disjunctive datalog have become an active research area in the past decade~\cite{DBLP:journals/tods/BienvenuCLW14,DBLP:journals/ai/KaminskiNG16,DBLP:conf/ijcai/LutzS17,DBLP:journals/lmcs/FeierKL19,DBLP:conf/kr/GerasimovaKKPZ20} lying at the crossroads of logic, database theory, knowledge representation in AI, circuit and descriptive complexity, and constraint satisfaction problems.

There have been numerous attempts to extend ontology and query languages with constructors that are capable of representing events over \emph{temporal data\/}; consult~\citeA{LutzWZ08,DBLP:conf/time/ArtaleKKRWZ17} for surveys and \citeA{DBLP:conf/ijcai/Gutierrez-Basulto17,DBLP:conf/ruleml/BorgwardtFK19,DBLP:conf/ijcai/WalegaGKK20,DBLP:conf/kr/WalegaGKK20,DBLP:journals/corr/abs-2111-06806} 
for more recent developments. However, so far the focus has only been on the uniform complexity of reasoning with arbitrary ontologies and queries in a given language rather than on determining the data complexity and FO-rewritability of individual temporal OMQs.
%
On the other hand, standard temporal logics are interpreted over linearly-ordered structures, and so the non-uniform analysis of OMQs in DLs and datalog mentioned above is not applicable to them.

In this paper, we take a first step towards understanding the problem of non-uniform FO-rewritability of OMQs over temporal data by focusing on the temporal dimension and considering OMQs given in linear temporal logic \LTL{} interpreted over $(\mathbb Z,<)$. In fact, already this basic `one-dimensional' temporal OBDA formalism provides enough expressive power in those real-world situations where the interaction among individuals in the object domain is not important and can be disregarded in data modelling. (This interaction is usually captured by binary relations (roles) in DLs, giving the models a `two-dimensional' character.) 
We illustrate this claim and the language of \LTL{} OMQs by an example. 

\begin{example}\em \label{ex:first}
A typical scenario for the use of OBDA technologies is where a non-IT-expert user, say a turbine engineer, analyses the behaviour of a complex system, turbines in our example, based on various  sensor measurements stored in a relational database. To be more specific, imagine that  turbines, $t$, are equipped with sensors, $s$, to measure such parameters as the rotor speed, the temperature of the blades, vibration, active power, etc. The relational database in a remote diagnostic centre might store a binary predicate $\nm{location}(s,t)$ saying that sensor $s$ is located in turbine $t$ and a ternary predicate $\nm{measurement}(s,v,n)$ giving the numerical value $v$ of the reading of $s$ at time instant $n$. The timestamps of sensor readings are  synchronised with a central clock, and so can be regarded as integers.

When defining events of interest like `active power trip' or `purging is over'\!, engineers usually operate with facts such as `the active power of turbine $t$ measured by $s$ is above 1.5MW at moment~$n$'\!, which can be obtained as database views of the form $\nm{ActivePower}^{t,s}_{\geq 1.5}(n)$. We regard these unary predicates as atomic concepts that can be true or false at different moments of time. Omitting $t$ and $s$ to unclutter notation, we can then assume that our virtual database~$\Abox$ consists of facts like 
\begin{equation}\label{turbinesdata}
\nm{Run}(6), \ \nm{ActivePower}_{\geq 1.5}(7), \ 
\nm{Malfunction}(7), \  \nm{Disabled}(10),
\end{equation}
based on which we analyse the behaviour of the turbines. As some sensors might occasionally fail to send their measurements, we cannot assume the data to be complete. Thus, in our example data above, the sensor detecting if the turbine is running (by measuring the electric current) failed to send a signal at time instant $7$. However, the power sensor attached to the turbine recorded $\ge 1.5$MW at 7, which should imply that the turbine was running at 7. This piece of domain knowledge can be encoded by the ontology axiom
\begin{equation}\label{ax1}
\Rbox\Lbox (\nm{ActivePower}_{\geq 1.5} \to \nm{Run})
\end{equation}
with the \LTL{}-operators $\Rbox$ (always in the future) and $\Lbox$ (always in the past).
Other \LTL{} axioms in our example ontology $\TO$ (designed by a domain expert) could look like 
\begin{align}\label{ax2}
& \Rbox\Lbox ( \nm{Pause} \land \nm{Run} \to \bot),\\\label{ax3}
& \Rbox\Lbox (\nm{Malfunction} \to \Rnext \nm{Pause}),\\\label{ax40}
& \Rbox\Lbox (\nm{Malfunction} \to \Rdiamond \nm{Diagnostics}),\\
\label{ax4}
& \Rbox\Lbox (\nm{Disabled} \to \neg \Rdiamond \nm{Diagnostics}).
\end{align}
The first of them says that a turbine cannot be paused and running at the same time; the second and third say that immediately  after ($\Rnext$) a malfunction, the turbine is paused and will eventually ($\Rdiamond$) be  diagnosed; the 
fourth axiom asserts that a disabled turbine will never undergo diagnostics in the future.

Now, if we are interested in continuous runs 
lasting at least two time units that end up in a non-run state, we (engineers) could write and execute the following simple query $\varkappa(x)$ 
with the previous-time operator $\Lnext$, assuming that 
$\varkappa(x)$ 
is mediated by the ontology $\TO$:
\begin{align*}
& \varkappa(x) \  = \ \neg \nm{Run} \land \Lnext \nm{Run} \land \Lnext \Lnext \nm{Run}.
\end{align*}
Intuitively, 
we are
looking for those timestamps $x$ in the active domain of the database at which 
this temporal formula   
is a logical consequence of $\TO$ and the data.
It is not hard to see that the only certain answer to the OMQ $(\TO,\varkappa(x))$ over $\Abox$ given by 
\eqref{turbinesdata} is the time instant $8$ because we can derive $\neg\nm{Run}(x)$ if $\nm{Pause}(x)$ or $\nm{Malfunction}(x-1)$ is in $\Abox$, or $\TO$ and $\Abox$ are inconsistent; and we know for certain that $\nm{Run}(x)$ iff $\Abox$ contains $\nm{Run}(x)$ or $\nm{ActivePower}_{\geq 1.5}(x)$, or again $\TO$ and $\Abox$ are inconsistent. These conditions can be expressed by the $\FO(<)$-query 
$\rq(x) = \varphi(x) \lor \nm{Incons}$,  
to be evaluated over $\Abox$, 
where 
\begin{align*}
&\varphi(x) = (\nm{Pause}(x) \lor \nm{Malfunction}(x-1)) \land {} (\nm{Run}(x-1) \lor \nm{ActivePower}_{\geq 1.5}(x-1)) \land {} \\
& \hspace*{8cm}(\nm{Run}(x-2) \lor \nm{ActivePower}_{\geq 1.5}(x-2)).
\end{align*}
and 
$\nm{Incons}$ is a disjunction 
of a few sentences such as 
\begin{align*}
& \exists x\, (\nm{Malfunction}(x) \land \nm{ActivePower}_{\geq 1.5}(x+1)),\\
& \exists x,y \, ((y \geq x) \land \nm{Disabled}(x) \land \nm{Malfunction}(y) ),\ \ \dots
\end{align*}
 that describe all of the cases when $\TO$ is inconsistent with $\Abox$ (which are left to the reader). The aim of a temporal OBDA system is to construct such an $\FO(<)$-rewriting $\rq(x)$ of the OMQ 
$(\TO,\varkappa(x))$ automatically, and evaluate it over the original relational data using a conventional database management system.
%
%
The OMQ $(\TO,\varkappa'(x))$ with 
$$\varkappa'(x)= \varkappa(x)  \ \land \ (\nm{Diagnostics} \ \lor \ \Rnext \nm{Diagnostics} \ \lor \ \Rnext \Rnext \nm{Diagnostics})$$ 
also returns $8$ over $\Abox$ because~\eqref{ax40} and~\eqref{ax4} imply that diagnostics took place some time in the interval $[8,10]$. We obtain an $\FO(<)$-rewriting of 
$(\TO,\varkappa'(x))$
 by adding to  $\rq(x)$ the conjunct 
\begin{equation*}
\exists y \, [ (x\leq y \leq x+2) \land (\nm{Diagnostics}(y) \lor (\nm{Disabled}(y) \land \exists z\, ((y-3 \leq z < y) \land \nm{Malfunction}(z)))) ]. 
\end{equation*}
%
%
%
%
%
%
\end{example}

\paragraph*{\bf\em Problems and related work.}
The initial problem we are interested in can be formulated in complexity-theoretic terms: given an \LTL{} OMQ $\q$, determine the data complexity of answering $\q$ over any data instance $\Abox$ in a given signature $\Xi$. For simplicity's sake, let us assume that $\q$ is Boolean (with a \yes/\no{} certain answer). It is also convenient to think of each $\Abox$ as a word whose symbol at position $\ell$ is the set of all atoms in $\Abox$ with timestamp $\ell$. Then the data instances $\Abox$ over which the answer to $\q$ is \yes{} form a language, $\L(\q)$, over the alphabet $2^\Xi$. In fact, using the automata-theoretic view of \LTL{}~\cite{VardiW86}, one can show (see Proposition~\ref{Prop:rewr-def} below) that the language $\L(\q)$ is regular, and so can be decided in \NCo~\cite{DBLP:journals/jacm/BarringtonT88,DBLP:journals/jcss/Barrington89}.

This observation naturally leads to the task of recognising the complexity of the word problem for a given regular language.
The circuit and descriptive complexity of regular languages was investigated by~\citeA{DBLP:journals/jcss/Barrington89,DBLP:journals/jcss/BarringtonCST92,Straubing94} who  established an $\ACz$/$\ACC$/$\NCo$ trichotomy, gave algebraic characterisations of languages in these classes (implying that the trichotomy is decidable) and also in terms of extensions of FO. Namely, the regular languages $\L$ in \ACz{} are definable by $\FO(<,\equiv)$-sentences with unary predicates $x \equiv 0\, (\text{mod}\ n)$; those in $\ACC$ are definable by $\FO(<,\MOD)$-sentences with quantifiers $\exists^n x\, \psi(x)$ checking whether the number of positions satisfying $\psi$ is divisible by $n$; and all regular languages $\L$ are definable in $\FO(\RPR)$ with relational primitive recursion~\cite{DBLP:journals/iandc/ComptonL90}.
$\FO(<)$-definable regular languages, which are decidable in \ACz{}, were proven to be the same as star-free languages~\cite{McNaughton&Papert71}, and their algebraic characterisation as languages with aperiodic syntactic monoids was obtained by~\citeA{DBLP:journals/iandc/Schutzenberger65a}.
The problem of deciding whether the language of a given DFA $\A$ is $\FO(<)$-definable is known to be \PSpace-complete~\cite{DBLP:journals/iandc/Stern85,DBLP:journals/TCS/ChoHyunh91,DBLP:journals/actaC/Bernatsky97}\footnote{This is also a special case of general results on finite monoids~\cite{Beaudryetal92,fleischeretal18}.}\!. However, the precise complexity of deciding whether a given regular language is in $\ACz$ and \mbox{$\FO(<,\equiv)$}-definable, or in $\ACC$ and $\FO(<,\MOD)$-definable, or $\NCo$-complete and is not $\FO(<,\MOD)$-definable (unless $\ACC = \NCo$) has remained open. It will be the first major problem we address in this article.

The characterisation of regular languages in terms of FO-definability allows us to reformulate the initial problem in terms of FO-rewritability that reduces OMQ answering (under the open world assumption) to model checking various types of FO-formulas: given an \LTL{} OMQ $\q$, how complex is it to decide whether $\q$ is $\FO(<)$-, $\FO(<,\MOD)$- or $\FO(<,\MOD)$-rewritable (that is, $\L(\q)$ is $\FO(<)$-, $\FO(<,\equiv)$- or $\FO(<,\MOD)$-definable)? Note that, by Kamp's Theorem~\cite{phd-kamp,DBLP:journals/corr/Rabinovich14}, $\FO(<)$-rewritability reduces answering \LTL{} OMQs to model checking \LTL-formulas.
$\FO(\RPR)$-rewritability of all \LTL{} OMQs was established by~\citeA{DBLP:journals/ai/ArtaleKKRWZ21} who also provided uniform rewritability results for various classes of \LTL{} OMQs (to be defined below); see Table~\ref{LTL-table}.

\paragraph*{\bf\em Our contribution.}
The first main result of this paper consists of the following parts. 
Let $\lang$ be one of the languages $\FO(<)$, $\FO(<,\equiv)$ or $\FO(<,\MOD)$.
First, using the algebraic characterisation results of~\citeA{DBLP:journals/jcss/Barrington89,DBLP:journals/jcss/BarringtonCST92,Straubing94}, we
obtain criteria for the $\lang$-definability of the language $\L(\A)$ of any given DFA $\A$ in terms of a limited part of the transition monoid of $\A$ (Theorem~\ref{DFAcrit}). Then, using our criteria and generalising the construction of~\citeA{DBLP:journals/TCS/ChoHyunh91},  we show that deciding $\lang$-definability of $\L(\A)$ for any minimal DFA $\A$ is \PSpace-hard (Theorem~\ref{DFAhard}).
Finally, we apply our criteria to give a \PSpace-algorithm deciding $\lang$-definability of $\L(\A)$ for not only any DFA but also any 2NFA $\A$ (Theorem~\ref{thm:2NFA}). 

To investigate $\lang$-rewritability of \LTL{} OMQs $\q = (\TO, \varkappa)$, we follow the classification  of~\citeA{DBLP:journals/ai/ArtaleKKRWZ21}, according to which the axioms of every \LTL{} ontology $\TO$ are given in the clausal form
\begin{equation}\label{normal-dl}
\Lbox\Rbox \big( C_1 \land \dots \land C_k ~\to~ C_{k+1} \lor \dots \lor C_{k+m} \big),
\end{equation}
where the $C_i$ are atoms, possibly prefixed by the temporal operators $\Rnext$, $\Lnext$, $\Rbox$, $\Lbox$.
Given any $\op \in \{\Box,\nxt,\Box\nxt\}$ and $\frag \in \{\bool, \horn, \krom,\core\}$, we denote by $\LTL_{\frag}^{\op}$ the fragment of \LTL{} with clauses~\eqref{normal-dl}, in which the $C_i$ can only use the (future and past) operators indicated in $\op$, and $m\leq 1$ if $\frag = \horn$; $k + m\leq 2$ if $\frag = \krom$; $k + m\leq 2$ and $m \leq 1$ if $\frag = \core$; and arbitrary $k$, $m$ if $\frag = \bool$. If $\op$ is omitted, the $C_i$ are atomic.
An $\LTL_\horn^{\op}$-ontology $\TO$ is linear if, in each of its axioms \eqref{normal-dl}, at most one $C_{i}$, for $1 \le i \le k$, can occur on the right-hand side of an axiom in $\TO$ (is an IDB predicate in datalog parlance).
We distinguish between arbitrary $\LTL^{\op}_\frag$ OMQs $\q = (\TO, \varkappa)$, where $\TO$ is any $\LTL^{\op}_\frag$ ontology and $\varkappa$ any \LTL-formula with $\nxt$-, $\Box$- and $\Diamond$-operators; positive OMQs (OMPQs), where $\varkappa$ is $\to, \neg$-free; existential OMPQs (OMPEQs) with $\Box$-free $\varkappa$; and atomic OMQs (OMAQs) with atomic $\varkappa$.

\begin{table}[t]\centering\small
\begin{tabular}{l|c|c|c}\toprule
 class of OMQs 
 %
 & $\FO(<)$  & { $\FO(<,\equiv)$} & { $\FO(<,\mathsf{MOD})$}\\
\hline
$\LTL_\horn\Xnext$ OMAQs & \multicolumn{3}{c}{\multirow{2}{*}{$\ExpSpace$ [Th.\,\ref{hornExpSpacehard}]}} \\
 $\LTL_{\bool}\Xallop$ OMQs& \multicolumn{3}{c}{} \\\hline
$\LTL_\krom$ OMPEQs & \multicolumn{3}{c}{$\ExpSpace$ [Th.\,\ref{thm:kromlowerexp}]} \\
\hline\hline
linear $\LTL_\horn\Xnext$ OMAQ & \multicolumn{3}{c}{$\PSpace$ [Th.\,\ref{th:linear-omaq-pspace}]} \\\hline
linear $\LTL_\horn\Xnext$ OMPQs & $\PSpace$ [Th.\,\ref{th:lin-ompq-fo-pspace}]   & $\PSpace$ [Th.\,\ref{th:lin-ompq-foe-pspace}]& ? \\
\hline\hline
$\LTL_\krom^{\smash{\scriptscriptstyle\bigcirc}}$ OMAQs & $\coNP$ [Th.\,\ref{thm:coNP}]& \multirow{3}{*}{all in \ACz~\cite{DBLP:journals/ai/ArtaleKKRWZ21} }& \multirow{3}{*}{--} \\
$\LTL_\core\Xnext$ OMPEQs & $\Pi^p_2$ [Th.\,\ref{thm:corepi-upper}]   & & \\
$\LTL_\core\Xnext$ OMPQs & \PSpace{} [Th.\,\ref{thm:ompqsforcore}]   & & \\
\bottomrule
\end{tabular}
\caption{Complexity of deciding FO-rewritability of \LTL{} OMQs.}
\label{tab:main-res}
\end{table}

The second main result of this article is the tight complexity bounds on deciding $\lang$-rewritability (and so data complexity) of \LTL{} OMQs from the classes defined above, which are summarised in Table~\ref{tab:main-res}. The \ExpSpace{} upper bound in the first stripe is shown using our $\lang$-definability criteria and exponential-size NFAs for \LTL{} akin to those of~\citeA{DBLP:books/el/07/Vardi07}; in the proof of the matching lower bound, an exponential-size automaton is encoded in a polynomial-size ontology. If the ontology in an $\LTL_\horn\Xnext$ OMAQ is linear, we show that its language (\yes-data instances) can be captured by a 2NFA with polynomially-many states, which allows us to reduce the complexity of deciding $\lang$-rewritability to \PSpace. However, for linear $\LTL_\horn\Xnext$ OMPQs (with more expressive queries $\varkappa$), the existence of polynomial-state 2NFAs remains open; instead, we show how the structure of the canonical models for $\LTL_\horn\Xnext$-ontologies can be utilised to yield a \PSpace{} algorithm. In the third stripe of the table, we deal with binary-clause ontologies. The \coNP-completeness of deciding FO-rewritability of $\LTL_\krom^{\smash{\scriptscriptstyle\bigcirc}}$ OMAQs is established using unary NFAs and results of~\citeA{DBLP:conf/stoc/StockmeyerM73}. The $\Pi^p_2$-completeness for $\LTL_\core\Xnext$ OMPEQs (without $\lor$ in ontologies but with $\land$, $\lor$, $\Diamond$ in queries) and the \PSpace-completeness for $\LTL_\core\Xnext$ OMPQs (admitting $\Box$ in queries, too) can be explained by the fact that the combined complexity of answering such OMPEQs and OMPQs is \NP-hard
rather than tractable as in the previous case.

It might be of interest to compare the results in Table~\ref{tab:main-res} with the complexity of deciding FO-rewritability (boundedness) of datalog queries and OMQs with a DL ontology and a conjunctive (CQ) or atomic query, which is:
\begin{itemize}
\item undecidable for linear datalog queries with binary predicates and for ternary linear datalog queries with a single recursive rule~\cite{DBLP:journals/jlp/HillebrandKMV95,DBLP:journals/siamcomp/Marcinkowski99};

\item 2\textsc{NExpTime}-complete for monadic disjunctive datalog queries and OMQs with an $\mathcal{ALC}$ ontology and a CQ~\cite{DBLP:conf/kr/BourhisL16,DBLP:journals/lmcs/FeierKL19};

\item 2\textsc{ExpTime}-complete for monadic datalog queries~\cite{DBLP:conf/stoc/CosmadakisGKV88,DBLP:conf/lics/BenediktCCB15}, even with a single recursive rule~\cite{DBLP:conf/pods/KikotKPZ21};

\item \textsc{NExpTime}-complete for OMQs with an ontology in any DL between $\mathcal{ALC}$  and $\mathcal{SHIU}$ and an atomic query~\cite{DBLP:journals/tods/BienvenuCLW14};

\item \textsc{ExpTime}-complete for OMQs with an $\mathcal{EL}$ ontology~\cite{DBLP:conf/ijcai/LutzS17,DBLP:journals/corr/abs-1904-12533};

\item \PSpace-complete for linear mon\-adic programs~\cite{DBLP:conf/stoc/CosmadakisGKV88,DBLP:journals/ijfcs/Meyden00};

\item \NP-complete for linear monadic single rule programs~\cite{DBLP:conf/pods/Vardi88}.
\end{itemize}

\paragraph{\bf\em Structure.} The article is organised in the following way. In the next  section, we introduce and illustrate by multiple examples \LTL{} OMQs and their semantics. We also briefly remind the reader of the basic algebraic and automata-theoretic notions that will be used later on in this article and show that FO-rewritability of \LTL{} OMQs is equivalent to FO-definability of certain regular languages. In Section~\ref{sec:groups}, we obtain algebraic characterisations of FO-definability, which are used in Sections~\ref{sec:reglang} and~\ref{sec:2nfa} to show that deciding each type of FO-definability of regular languages is \PSpace-complete. In Sections~\ref{sec:LTL-genearal}-\ref{sec:others}, we prove the complexity bounds from Table~\ref{tab:main-res} and then conclude in Section~\ref{sec:conclusion}. Some of the technical results and constructions are given in the appendices to the article.



\section{Preliminaries}\label{prelims}

\paragraph{\bf\em Temporal ontology-mediated queries.}
In our setting, the alphabet of linear temporal logic \LTL{} comprises a set of \emph{atomic concepts} (or simply \emph{atoms}) $A_i$, $i < \omega$. \emph{Basic temporal concepts}, $C$, are defined by the grammar
\begin{equation*} 
C \ \ ::=\ \ A_i  \ \ \mid\ \ \Rbox C \ \ \mid \ \ \Lbox C \ \ \mid\ \ \Rnext C \ \ \mid \ \ \Lnext C
\end{equation*}
with the \emph{temporal operators} $\Rbox$/$\Lbox$ (always in the future/past) and $\Rnext$/$\Lnext$ (at the next/pre\-vious moment). A \emph{temporal ontology}, $\TO$, is a finite set of \emph{axioms} of the form\footnote{From now on, to improve readability we make the prefix $\Lbox\Rbox$ in axioms implicit (which is taken into account in their semantics).}
\begin{equation}\label{axiom1}
C_1 \land \dots \land C_k ~\to~ C_{k+1} \lor \dots \lor C_{k+m},
\end{equation}
where $k,m \ge 0$, the $C_i$ are basic temporal concepts, the empty $\land$ is $\top$, and the empty $\lor$ is $\bot$.
Following the \DL{} convention~\cite{ACKZ09,DBLP:conf/ijcai/ArtaleKKRWZ15}, we classify ontologies by the shape of their axioms and the temporal operators that can occur in them. Suppose $\frag \in \{\horn, \krom,\core,\bool\}$ and $\op \in \{\Box, \nxt,\Box\nxt\}$. The axioms of an $\LTL_\frag^{\op}$-\emph{ontology} may only contain occurrences of the (future and past) temporal operators in $\op$ and satisfy the following restrictions on $k$ and $m$ in~\eqref{axiom1} indicated by $\frag$:
\textit{horn} requires $m\leq 1$,
\textit{krom} requires $k + m\leq 2$,
\textit{core} both $k + m\leq 2$ and $m \leq 1$, while
\textit{bool} imposes no restrictions. To illustrate, axioms~\eqref{ax1} and~\eqref{ax2} from Example~\ref{ex:first} are allowed in all of these fragments, \eqref{ax3} is in $\LTL_\core\Xnext$,~\eqref{ax4} can be expressed in $\LTL_\core\Xbox$ and~\eqref{ax40} can be expressed in $\LTL_\krom\Xbox$ as explained in Remark~\ref{rem1} below.

A basic concept is called an \emph{IDB} (intensional database) \emph{concept} in an ontology $\TO$ if its atom occurs on the right-hand side of some axiom in $\TO$. The set of IDB atomic concepts in $\TO$ is denoted by $\textit{idb}(\TO)$. An $\LTL_\horn^{\op}$-ontology is called \emph{linear} if each of its axioms $C_1 \land \dots \land C_k \to D$, where $D$ is either a basic temporal concept $C$ or $\bot$, contains \emph{at most one} IDB concept $C_i$, for $1 \le i \le k$.

A \emph{data instance}---or an \emph{ABox} in description logic parlance---is a finite set $\Abox$ of atoms $A_i(\ell)$, for some \emph{timestamps} $\ell \in \Z$, together with a finite interval $\tem(\Abox) = [m,n] \subseteq \mathbb Z$, the \emph{active domain} of $\Abox$, such that $m \le \ell \le n$, for all $A_i(\ell) \in \Abox$. If $\Abox = \emptyset$, then $\tem(\Abox)$ may also be $\emptyset$. Otherwise, we assume without loss of generality that $m = 0$.  If $\tem(\Abox)$ is not specified explicitly, it is assumed to be either empty or $[0,n]$, where $n$ is the maximal timestamp in $\Abox$.  By a \emph{signature}, $\Xi$, we mean any finite set of atomic concepts. An ABox $\Abox$ is a $\Xi$-\emph{ABox} if $A_i(\ell) \in \Abox$ implies $A_i \in \Xi$.

We query ABoxes by means of \emph{temporal concepts}, $\varkappa$, which are \LTL-formulas built from the atoms $A_i$, Booleans $\land$, $\lor$, $\neg$, temporal operators $\Rnext$, $\Rbox$, $\Rdiamond$ (eventually) and their past-time counterparts $\Lnext$, $\Lbox$, $\Ldiamond$ (previously).
If $\varkappa$ does not contain $\neg$, we call it \emph{positive}; if $\varkappa$ does not contain $\Lbox$ and $\Rbox$ either, we call it \emph{positive existential}.

A \emph{temporal interpretation} is a structure of the form $\I = (\Z, A_0^\I, A_1^\I,\dots)$ with   $A_i^\I\subseteq \Z$, for every $i < \omega$. The \emph{extension} $\varkappa^\I$ of a temporal concept $\varkappa$ in $\I$ is defined inductively as usual in \LTL{} under the `strict semantics'~\cite{gkwz,DBLP:books/cu/Demri2016}:
\begin{align*}
&
(\Rnext \varkappa)^\I  =  \bigl\{\, n\in\Z \mid n + 1\in \varkappa^\I \,\bigr\}, \\
&
(\Rbox \varkappa)^\I  =  \bigl\{\, n\in\Z \mid k\in \varkappa^\I \text{  for all } k>n \,\bigr\}, \\
&
(\Rdiamond \varkappa)^\I  =  \bigl\{\, n\in\Z \mid \text{there is } k>n \text{ with } k\in\varkappa^\I\,\bigr\},
%
\end{align*}
and symmetrically for the past-time operators.
We regard $\I,n\models \varkappa$ as synonymous to $n \in \varkappa^\I$.
An axiom~\eqref{normal-dl} is \emph{true} in an interpretation $\I$ if $C_1^{\I} \cap \dots \cap C_k^{\I} ~\subseteq~ C_{k+1}^{\I} \cup \dots \cup C_{k+m}^{\I}$. An interpretation $\I$ is a \emph{model} of $\TO$ if all axioms of $\TO$ are true in $\I$; it is a \emph{model} of $\Abox$ if $A_i(\ell) \in \Abox$ implies $\ell \in A_i^\I$. 

An $\LTL^{\op}_\frag$ \emph{ontology-mediated query} (OMQ) is a pair of the form $\q = (\TO, \varkappa)$, where $\TO$ is an $\LTL^{\op}_\frag$ ontology and $\varkappa$ a temporal concept. If $\varkappa$ is positive, we call $\q$ a \emph{positive OMQ} (OMPQ, for short), if $\varkappa$ is positive existential, we call $\q$ a \emph{positive existential OMQ} (OMPEQ), and if $\varkappa$ is an atomic concept, we call $\q$  \emph{atomic} (OMAQ). The set of atomic concepts occurring in $\q$ (in $\TO$) is denoted by $\sig(\q)$ (respectively, $\sig(\TO)$).

We can treat $\q = (\TO, \varkappa)$ as a \emph{Boolean} OMQ, which returns \yes/\no, or as a \emph{specific} OMQ, which returns timestamps from the ABox in question assigned to the free variable, say $x$, in the standard FO-translation of $\varkappa$. In the latter case, we write $\q(x) = (\TO, \varkappa(x))$.
More precisely, the \emph{certain answer} to a Boolean OMQ $\q = (\TO, \varkappa)$  over an ABox $\mathcal A$ is \yes{} if, for every model $\I$ of $\TO$ and $\mathcal{A}$, there is $k \in \Z$ such that $k \in \varkappa^\I$, in which case we write $(\TO, \Abox) \models \exists x\, \varkappa(x)$. If $(\TO, \Abox) \not\models \exists x \varkappa(x)$, the certain answer to $\q$ over $\Abox$ is \no.
We write $(\TO, \Abox) \models \varkappa(k)$, for $k \in \Z$, if $k \in \varkappa^\I$ in all models $\I$ of $\TO$ and $\mathcal{A}$. A \emph{certain answer} to a specific OMQ $\q(x) = (\TO, \varkappa(x))$ over $\Abox$ is any $k \in \tem(\Abox)$ with $(\TO, \Abox) \models \varkappa(k)$.
By the \emph{answering} (or \emph{evaluation}) \emph{problem} for $\q$ or $\q(x)$ we understand the decision problem `$(\TO, \Abox) \models^?\! \exists x \varkappa(x)$' or `$(\TO, \Abox) \models^?\! \varkappa(k)$' with input $\Abox$ or, respectively, $\Abox$ and $k \in \tem(\Abox)$,
We say that $\q$/$\q(x)$ is in a complexity class $\mathcal{C}$ if the answering problem for $\q$/$\q(x)$ is in $\mathcal{C}$.

\begin{example}\label{example1}\em
$(i)$ Suppose $\TO_1 = \{ A \to \Rbox B,\ \Rbox B \to C \}$ and $\q_1 = (\TO_1, C \land D)$. The certain answer to $\q_1$ over $\Abox_1 = \{D(0), B(1), A(1)\}$ is \yes{}, and \no{} over $\Abox_2 = \{D(0), A(1)\}$. The only  answer to $\q_1(x) = \big(\TO_1, (C \land D)(x)\big)$ over $\Abox_1$ is $0$.

$(ii)$ Let $\TO_2 = \{\,\Lnext A \to B, \ \Lnext B \to A,\, A \land B \to \bot\,\}$.
The certain answer to $\q_2 = (\TO_2, C)$ over $\Abox_1 = \{A(0)\}$ is \no{}, and \yes{} over $\Abox_2 = \{A(0), A(1) \}$. There are no certain answers to $\q_2(x) = (\TO_1, C(x))$ over $\Abox_1$, while over $\Abox_2$ the answers are $0$ and $1$.

$(iii)$ Consider next
$\TO_3 = \{\Lnext B_k \land A_0 \to B_k, \, \Lnext B_{1-k} \land A_1 \to B_k \mid k = 0, 1\}.
$
For any word $\avec{e} = e_1\dots e_{n} \in \{0,1\}^n$, let
$$
\Abox_{\avec{e}}  =   \{B_0(0)\} \cup  \{A_{e_i}(i) \mid 0 < i \leq n\} \cup \{E(n)\}.
$$
The certain answer to $\q_3 = (\TO_3,B_0 \land E)$ over $\Abox_{\avec{e}}$ is \yes{} iff the number of 1s in $\avec{e}$ is even.

$(iv)$ Let $\TO_4 = \{A \to \Rnext B\}$ and $\q_4 = (\TO_4, B)$. Then, the answer to $\q_4$ over $\Abox = \{ A(0) \}$ is \yes; however, there are no certain answers to $\q_4(x) = (\TO_4, B(x))$ over $\Abox$.

$(v)$ Finally, suppose $\TO_5 = \{A \to B \lor \Rnext B\}$. The certain answer to $\q_5 = (\TO_5,B)$ over $\Abox = \{ A(0), C(1) \}$ is \yes; however, there are no certain answers to $\q_5(x)$ over $\Abox$. \hfill $\dashv$
\end{example}

Thus, as shown by Example~\ref{example1} $(iv)$ and $(v)$, a Boolean OMAQ $\q = (\TO, B)$ can have an answer \yes{} over an ABox $\Abox$ even though the set of certain answers to the specific OMAQ $\q(x) = (\TO, B(x))$ over $\Abox$ is empty. (Clearly, the existence of certain answers to $\q(x)$ over $\Abox$ implies that the answer to $\q$ over $\Abox$ is \yes{}.) In $(iv)$, the reason for the absence of certain answers to $\q_4(x)$ is that any $k \in \Z$ with $(\TO, \Abox) \models B(k)$ is not in $\tem(\Abox)$. In $(v)$, the reason is that there is no $k \in \Z$ with $(\TO, \Abox) \models B(k)$ even though every model $\I$ of $\TO$ and $\Abox$ contains some $k \in \tem(\Abox) \subseteq \Z$ with $\I, k \models B$.

Two OMQs are called $\Xi$-\emph{equivalent}, for a signature $\Xi$, if they return the same certain answers over any $\Xi$-ABox. Without loss of generality, we assume that, when answering an \LTL{} OMQ $\q$ or $\q(x)$ over $\Xi$-ABoxes, we always have $\Xi \subseteq \sig(\q)$. Indeed, if this is not the case, we can extend the ontology of $\q$ with $|\Xi|$-many dummy axioms of the form $A \to A$ and obtain a $\Xi$-equivalent OMQ.

\begin{remark}\label{rem1}\em
If arbitrary \LTL-formulas (possibly with the until or since operators) in the scope of $\Lbox \Rbox$ are used as axioms of an ontology $\TO$, then one can construct an $\LTL_{\bool}\Xallop$ ontology $\TO'$ that is a \emph{model-conservative extension} of $\TO$~\cite<e.g.,>{FisherDP01,AKRZ:LPAR13}. For example, let $\TO'$ be the result of replacing axiom~\eqref{ax40} in $\TO$ from Example~\ref{ex:first} by two axioms $\textit{Malfunction} \land \Rbox X \to \bot$ and \mbox{$\top \to X \lor \textit{Diagnostics}$}, for a fresh $X$. Then the OMQ $\q = (\TO, \varkappa)$ is $\sig(\q)$-equivalent to $\q' = (\TO', \varkappa)$. Axiom~\eqref{ax4} can be replaced with $\nm{Diagnostics} \to \Lbox Y$ and $\nm{Disabled} \land Y \to \bot$ with fresh $Y$.

Similarly, every $\LTL_{\horn}\Xallop$ OMQ $\q = (\TO, \varkappa)$ has the same certain answers over any $\sig(\q)$-ABox as an $\LTL_{\horn}\Xallop$ OMQ $\q' = (\TO', \varkappa)$, in which $\TO'$  contains axioms of the form $\avec{C}\to \bot$ or $\avec{C}\to B$ only, for some $\avec{C}= C_1 \land \dots \land C_n$ and an atomic concept $B$. For example, the axiom $A \to \Rnext \Rbox B$ can be replaced by $\Lnext A \to X$, $\Lnext X \to X$, and $\Lnext X \to B$ with fresh $X$. Note also that if $\TO$ is a linear $\LTL_{\horn}\Xnext$ ontology, then $\TO'$ is also a linear $\LTL_{\horn}\Xnext$ ontology.
\end{remark}

We now introduce the central notion of this article, which reduces answering OMQs to evaluating FO-formulas over structures representing ABoxes.

Let $\lang$ be a class of FO-formulas that can be interpreted over finite linear orders. A Boolean OMQ $\q$ is $\lang$-\emph{rewritable over $\Xi$-ABoxes} if there is an $\lang$-sentence $\rq$ such that, for any $\Xi$-ABox~$\Abox$, the certain answer to $\q$ over $\Abox$ is \yes{} iff $\SA \models \rq$. Here, $\SA$ is a structure\footnote{We allow structures with the empty domain, in which $\exists x\, (x = x)$ is false~\cite<e.g.,>{DBLP:books/daglib/0071316}.} with domain $\tem(\Abox)$ ordered by $<$, in which $\SA \models A_i(\ell)$ iff $A_i(\ell)\in \Abox$.
A specific OMQ $\q(x)$ is $\lang$-\emph{rewritable over $\Xi$-ABoxes} if there is an $\lang$-formula $\rq(x)$ with one free variable $x$ such that, for any $\Xi$-ABox~$\Abox$, $k$ is a certain answer to $\q(x)$ over $\Abox$ iff $\SA \models \rq(k)$. The sentence $\rq$ and formula $\rq(x)$ are called $\lang$-\emph{rewritings} of the OMQs $\q$ and $\q(x)$, respectively.

We require four languages $\lang$ for rewriting \LTL{} OMQs, which are listed below in order of increasing expressive power:
\begin{description}
\item[$\FO(<)$:] (monadic) first-order formulas with the built-in predicate $<$ for order;

\item[$\FO(<,\equiv)$:] $\FO(<)$-formulas with unary predicates $x \equiv 0\, (\text{mod}\ N)$, for all $N >1$;

\item[$\FO(<,\MOD)$:] $\FO(<)$-formulas with  quantifiers $\exists^N\! x$, for all $N >1$, that are defined by taking $\SA \models \exists^N\! x\, \psi(x)$ iff the cardinality of $\{ n \in \tem(\Abox) \mid \SA \models \psi(n) \}$ is divisible by $N$ (note that $x \equiv 0\, (\text{mod}\ N)$ is definable as $\exists^N\! y\, (y < x) $);

\item[$\FO(\RPR)$:] $\FO(<)$ with relational primitive recursion~\cite{DBLP:journals/iandc/ComptonL90}.
\end{description}
As well-known, $\FO(<,\equiv)$ is strictly more expressive than $\FO(<)$ and strictly less expressive than $\FO(<,\MOD)$, which is illustrated by the examples below.
\begin{example}\label{example2}\em
$(i)$ An $\FO(<)$-rewriting of $\q_1(x)$ from Example~\ref{example1} is
\begin{equation*}
\rq_1(x) = D(x) \land [C(x) \lor \exists y\, (A(y) \land{} \forall z \, ((x < z \leq y) \to B(z)))],
\end{equation*}
$\exists x\, \rq_1(x)$ is an $\FO(<)$-rewriting of $\q_1$.

$(ii)$ An $\FO(<, \equiv)$-rewriting of $\q_2(x)$ is
\begin{multline*}
\rq_2(x) =~ C(x) \lor \exists x,y \, [(A(x) \land A(y) \land{} \mathsf{odd}(x,y)) \lor{}\\ (B(x) \land B(y) \land \mathsf{odd}(x,y)) \lor{}
(A(x) \land B(y) \land \neg\mathsf{odd}(x,y))],
\end{multline*}
where
$\mathsf{odd}(x,y) = \big(x \equiv 0\, (\text{mod}\ 2) \leftrightarrow y \not\equiv 0\, (\text{mod}\ 2)\big)$ implies that $|x - y|$ is odd; $\exists x\, \rq_2(x)$ is an $\mbox{$\FO(<, \equiv)$}$-rewriting of $\q_2$. Recall that $\mathsf{odd}$ is not $\FO(<)$-expressible \cite{Libkin}.

$(iii)$ The OMQ $\q_3$ is not rewritable to an FO-formula with any numeric predicates as PARITY is not in $\ACz$~\cite{DBLP:journals/mst/FurstSS84}; the following sentence is an \mbox{$\FO(<,\MOD)$}-rewriting of $\q_3$:
\begin{multline*}
\rq_3 =  \exists x, y\, \big[ E(x) \land (y \leq x) \land \forall z\, \big((y < z \leq x) \to{}
A_0(z) \lor A_1(z) \big) \land{} \\
\big((B_0(y) \land \exists^2 z\, ((y < z \leq x) \land A_1(z))) \lor{}
(B_1(y) \land \neg \exists^2 z\, ((y < z \leq x) \land A_1(z))) \big) \big].
\end{multline*}

$(iv)$ An $\FO(<)$-rewriting of $\q_4(x)$ is $B(x) \lor A(x-1)$; an $\FO(<)$-rewriting of the Boolean query $\q_4$ is $\rq_4 = \exists x \, (A(x) \lor B(x))$.

$(v)$ $\rq_4$ is also an $\FO(<)$-rewriting of $\q_5$; $B(x)$ is an $\FO(<)$-rewriting of $\q_5(x)$. \hfill $\dashv$
\end{example}


As shown by~\citeA{DBLP:journals/ai/ArtaleKKRWZ21}, all Boolean and specific \LTL{} OMQs are $\FO(\RPR)$-rewri\-table and specific OMPQs can be  classified syntactically by their rewritability type as shown in Table~\ref{LTL-table}. This means, e.g., that all $\LTL_\core\Xallop$ OMPQs  are $\FO(<,\equiv)$-rewritable, with some of them being not $\FO(<)$-rewritable. It is to be noted that $\FO(<,\MOD)$-rewritable OMQs such as $\q_3$ in Example~\ref{example1} are not captured by these  syntactic classes.

\begin{table}[h]
\centering%
\tabcolsep=5pt%
\begin{tabular}{ccccc}\toprule
     \rule[-3pt]{0pt}{12pt}
     &\multicolumn{2}{c}{OMAQs}
     & \multicolumn{2}{c}{OMPQs}
     \\
     \rule[-3pt]{0pt}{13pt}$\frag$
     & $\LTL_\frag\Xbox$ & $\LTL_\frag\Xnext$ {\footnotesize and} $\LTL_\frag\Xallop$ & $\LTL_\frag\Xbox$ &  $\LTL_\frag\Xnext$ {\footnotesize and} $\LTL_\frag\Xallop$
     \\\midrule%
     \bool & \multirow{4}{*}{$\FO(<)$ } &   $\FO(\RPR)$  & \multirow{2}{*}{$\FO(\RPR)$}  & \multirow{3}{*}{FO(RPR)}  \\\cmidrule[0pt](lr){3-3}
     \krom &  & $\FO(<,\equiv)$ &&     \\\cmidrule(lr){4-4}
     \horn & & $\FO(\RPR)$ & \multirow{2}{*}{$\FO(<)$ }  & \\\cmidrule(lr){5-5}
     \core & & $\FO(<,\equiv)$ & & $\FO(<,\equiv)$ \\\bottomrule\\
   \end{tabular}
\caption{Rewritability of specific \LTL{} OMQs.}
\label{LTL-table}
\end{table}

Our aim here is to understand how complex it is to decide the optimal type of FO-rewritability for a given \LTL{} OMQ $\q$ over $\Xi$-ABoxes.
As this will rely on an intimate connection between $\lang$-rewritability of OMQs and $\lang$-definability of certain regular languages, we briefly remind the reader of the basic algebraic and automata-theoretic notions that are used in the remainder of the article.

\paragraph{\bf\em Monoids and Groups.}

A \emph{semigroup} is a structure $\Sg=(S,\cdot)$, where $\cdot$ is an associative binary operation.
For $s,s'\in S$ and $n>0$, we set $s^n = \underbrace{s\cdot s \cdot \ldots \cdot s}_n$ and often write $ss'$ for $s\cdot s'$.
An element $s$ of $\Sg$ is \emph{idempotent} if $s^2=s$.
An element $e$ is an \emph{\unit{}} in $\Sg$ if $e\cdot x=x\cdot e=x$ for all $x\in S$
(such an $e$ is unique, if exists). The \unit{} element is clearly idempotent.
A \emph{monoid} is a semigroup with an \unit{} element.
For any element $s$ in a monoid, we set $s^0=e$.
A monoid $\Sg=(S,\cdot)$ is a \emph{group} if, for any $x\in S$, there is $x^-\in S$---the \emph{inverse of} $x$---such that
$x\cdot x^-=x^-\cdot x=e$ (every element of a group has a unique inverse). A group is \emph{trivial} if it has one element, and \emph{nontrivial} otherwise.

Given two groups $\G=(G,\cdot)$ and $\G'=(G',\cdot')$,
a map $h\colon G\to G'$ is a \emph{group homomorphism from $\G$ to $\G'$} if  $h(g_1\cdot g_2)=h(g_1)\cdot' h(g_2)$ for all
$g_1,g_2\in G$. (It is easy to see that any group homomorphism maps the \unit{} of $\G$
to the \unit{} of $\G'$ and preserves the inverses. The set
$\{h(g)\mid g\in G\}$
is closed under $\cdot'$, and so is a group, the \emph{image of $\G$ under $h$\/}.)
$\G$ is a \emph{subgroup of} $\G'$ if $G\subseteq G'$ and
the identity map $\id_G$ is a group homomorphism.
Given $X\subseteq G$, the \emph{subgroup of $\G$ generated by $X$} is the smallest subgroup of $\G$ containing $X$.
%
%
The \emph{order} $\ord_\G(g)$ of an element $g$ in $\G$ is
the smallest positive number $n$ with $g^n=e$, which always exists. Clearly, $\ord_\G(g)=\ord_\G(g^-)$ and, if $g^k=e$
then $\ord_\G(g)$ divides $k$. Also,
\begin{equation}\label{idemp}
\mbox{if $g$ is a nonidentity element in a group $\G$, then $g^k\ne g^{k+1}$ for any $k$.}
\end{equation}
A semigroup $\Sg'=(S',\cdot')$ is a \emph{subsemigroup} of a semigroup $\Sg=(S,\cdot)$
if $S'\subseteq S$ and $\cdot'$ is the restriction of $\cdot$ to $S'$.
Given a monoid $\M=(M,\cdot)$ and a set $S\subseteq M$, we say that $S$ \emph{contains the group} $\G=(G,\cdot')$, if
$G\subseteq S$ and $\G$ is a subsemigroup of $\M$.
Note that we do {\bf not} require the \unit{} of $\M$ to be in $\G$, even if it is in $S$.
If $S=M$, we also say that $\M$ \emph{contains the group} $\G$, or $\G$ \emph{is in} $\M$. We call a monoid $\M$ \emph{aperiodic} if it does not contain any nontrivial groups.

Let $\Sg=(S,\cdot)$ be a finite semigroup and $s\in S$. By the pigeonhole principle,
there exist $i,j\geq 1$ such that $i+j\leq |S|+1$ and $s^{i}=s^{i+j}$. Take the minimal such numbers, that is,
let $i_s,j_s\geq 1$ be such that $i_s+j_s\leq |S|+1$  and $s^{i_s}=s^{i_s+j_s}$ but $s^{i_s},s^{i_s+1},\dots,s^{i_s+j_s-1}$ are all different.
Then clearly $\G_s=(G_s,\cdot)$, where $G_s=\{s^{i_s},s^{i_s+1},\dots,s^{i_s+j_s-1}\}$, is a subsemigroup of $\Sg$.
It is easy to see that there is $m\geq 1$ with $i_s\leq m\cdot j_s<i_s+j_s\leq |S|+1$, and so $s^{m\cdot j_s}$ is idempotent. Thus, for every element $s$ in a semigroup $\Sg$, we have the following:
\begin{align}
\label{gini}
& \mbox{there is $n\geq 1$ such that $s^n$ is idempotent;}\\
\label{ginii}
& \mbox{$\G_s$ is a group in $\Sg$ (isomorphic to the cyclic group $\mathbb Z_{j_s}$);}\\
\label{giniii}
& \mbox{$\G_s$ is nontrivial iff $s^n\ne s^{n+1}$ for any $n$.}
\end{align}
%
%
%
 %
%
%
%
%
%
%
Let $\fu \colon Q\to Q$ be a function on a finite set $Q \ne \emptyset$.
For any $p\in Q$, the subset $\{\fu^k(p)\mid k<\omega\}$ with the obvious multiplication is a semigroup,
and so we have:
%
\begin{align}
\label{fpi}
& \mbox{for every $p\in Q$, there is $n_p\geq 1$ such that $\fu^{n_p}\bigl(\fu^{n_p}(p)\bigr)=\fu^{n_p}(p)$;}\\
\label{fpii}
& \mbox{there exist $q\in Q$ and $n\geq 1$ such that $q=\fu^n(q)$;}\\
\nonumber
& \mbox{for every $q\in Q$, if $q=\fu^k(q)$ for some $k\geq 1$,} \\
\label{fpiii}
& \hspace*{3.5cm}\mbox{then there is $n$, $1\leq n\leq |Q|$, with $q=\fu^n(q)$.}
\end{align}
For a definition of \emph{solvable} and \emph{unsolvable} groups the reader is referred to~\citeA{rotman1999introduction}. In this article, we only use the fact that any homomorphic image of a solvable group is solvable and the Kaplan--Levy criterion \citeyear{kaplan_levy_2010} (generalising Thompson's \citeyear[Corollary~3]{thompson1968}) according to which a finite group $\G$ is unsolvable iff it contains elements $a,b,c$ such that $\ord_\G(a)=2$, $\ord_\G(b)$ is an odd prime,
$\ord_\G(c)>1$ and coprime to both $2$ and $\ord_\G(b)$, and $abc$ is the \unit{} of $\G$.

A one-to-one and onto function on a finite set $S$ is called a \emph{permutation on} $S$.
The \emph{order of a permutation} $\fu$ is its order in the group of all permutations on $S$
(whose operation is composition, and its identity element is the identity permutation $\id_S$).
We use the standard cycle notation for permutations.

Suppose that $\G$ is a monoid of $Q\to Q$ functions, for some finite set $Q \ne \emptyset$.
Let $S=\{q\in Q\mid e_\G(q)=q\}$, where $e_\G$ the \unit{} element in $\G$. For every function $\fu$ in $\G$, let $\fu\!\!\restriction_S$ denote
the restriction of $\fu$ to $S$.
Then
\begin{align}
\label{groupini}
& \mbox{$\G$ is a group iff  $\fu\!\!\restriction_S$ is a permutation on $S$, for every $\fu$ in $\G$;}\\
\nonumber
& \mbox{if $\G$ is a group and $\fu$ is a nonindentity element in it, then $\fu\!\!\restriction_S\ne\id_S$ and}\\
\label{groupinii}
& \hspace*{3cm}\mbox{the order of the permutation $\fu\!\!\restriction_S$ divides $\ord_\G(\fu)$.}
\end{align}

\paragraph{\bf\em Automata, languages, and OMQs.}

A \emph{two-way nondeterministic finite automaton} is a quintuple $\A = (Q, \Sigma, \delta, Q_0, F)$ that consists of an alphabet $\Sigma$, a finite set $Q$ of states with a subset $Q_0 \ne \emptyset$ of initial states and a subset $F$ of accepting states, and a transition function $\delta \colon Q \times \Sigma \to 2^{Q \times \{-1,0,1\}}$ indicating the next state and whether the head should move left ($-1$), right ($1$), or stay put. If $Q_0 = \{q_0\}$ and $|\delta(q, a)| = 1$, for all $q \in Q$ and $a \in \Sigma$, then $\A$ is \emph{deterministic}, in which case we write $\A = (Q, \Sigma, \delta, q_0, F)$.
If $\delta(q, a) \subseteq Q \times \{1\}$, for all $q \in Q$ and $a \in \Sigma$, then $\A$ is a \emph{one-way} automaton, and we write $\delta \colon Q \times \Sigma \to 2^Q$. As usual, DFA and NFA refer to one-way deterministic and non-deterministic finite automata, respectively, while 2DFA and 2NFA to the corresponding two-way automata. Given a 2NFA $\A$, we write $q \to_{a,d} q'$ if $(q', d) \in \delta(q,a)$; given an NFA $\A$, we write $q \to_{a} q'$ if $q' \in \delta(q,a)$.
A \emph{run} of a 2NFA $\A$ is a word in $(Q \times \mathbb N)^*$. A run $(q_0, i_0), \dots, (q_m, i_m)$ is a \emph{run of $\A$ on a word} $w = a_0 \dots a_n \in \Sigma^*$ if $q_0 \in Q_0$, $i_0 = 0$ and there exist $d_0, \dots, d_{m-1} \in \{-1,0,1\}$ such that $q_j \to_{a_{i_j}, d_j} q_{j+1}$ and $i_{j+1} = i_j+d_j$ for all $j$, $0 \leq j < m$. The  run is \emph{accepting} if $q_m \in F$, $i_m = n+1$. $\A$ \emph{accepts} $w \in \Sigma^*$ if there is an accepting run of $\A$ on $w$; the language $\L(\A)$ of $\A$ is the set of all words accepted by $\A$.

Given an NFA $\A$, states $q,q' \in Q$, and $w = a_0 \dots a_n \in \Sigma^*$, we write $q \to_w q'$ if either $w = \varepsilon$ and $q' = q$ or there is a run of $\A$ on $w$ that starts with $(q_0, 0)$ and ends with $(q', n+1)$. We say that a state $q \in Q$ is \emph{reachable} if $q' \to_w q$, for some $q' \in Q_0$ and $w \in \Sigma^*$.


Given a DFA $\A = (Q, \Sigma, \fu, q_0, F)$ and a
word $w \in \Sigma^\ast$, we define a function $\fu_w \colon Q \to Q$ by taking  $\fu_w(q) = q'$ iff $q \to_w q'$. We also define an equivalence relation $\simm$ on the set $\Qr\subseteq Q$ of reachable states by taking $q\simm q'$ iff,
for every $w \in \Sigma^\ast$, we have $\fu_w(q)\in F$ just in case $\fu_w(q')\in F$. We denote the $\simm$-class of $q$ by $\simclass{q}$, and let
$\simclass{X}=\{\simclass{q}\mid q\in X\}$ for $X\subseteq \Qr$. Define $\fum_w\colon \simclass{\Qr\!}\to \simclass{\Qr\!}$ by taking
$\fum_w(\simclass{q})=\simclass{\fu_w(q)}$.
Then $\bigl(\simclass{\Qr\!},\Sigma,\fum,\simclass{q_0},\simclass{(F\cap \Qr)}\bigr)$ is the \emph{minimal DFA} whose language coincides with the language of $\A$.
Given a regular language $\L$, we denote by $\Amin$ the minimal DFA whose language is $\L$.

The \emph{transition monoid of} a DFA $\A$ is $M(\A) = (\{ \fu_w  \mid w \in \Sigma^\ast\},\cdot)$ with
$\fu_v\cdot\fu_w = \fu_{vw}$, for any $v,w$.
The \emph{syntactic monoid $M(\L)$ of $\L$} is the transition monoid $M(\Amin)$ of $\Amin$.
%
%
The \emph{syntactic morphism of} $\L$ is the map $\eta_\L$ from $\Sigma^*$ to the domain of $M(\L)$ defined by $\eta_\L(w) = \fum_w$. We call $\eta_\L$ \emph{quasi-aperiodic} if $\eta_\L(\Sigma^t)$ is aperiodic for every $t<\omega$.

Let $\lang \in \{ \FO(<), \FO(<,\equiv), \FO(<,\MOD)\}$. A language $\L$ over $\Sigma$ is \emph{$\lang$-definable} if there is an $\lang$-sentence $\varphi$ in the signature $\Sigma$, whose symbols are treated as unary predicates, such that, for any $w \in \Sigma^*$, we have $w=a_0\ldots a_n \in \L$ iff $\mathfrak S_w \models \varphi$, where $\mathfrak S_w$ is an FO-structure with domain $\{0,\dots,n\}$ ordered by $<$, in which $\mathfrak S_w \models a(i)$ iff $a=a_i$, for $0\leq i\leq n$.

%





Table~\ref{tab:algebra} summarises the known results that connect definability of a regular language $\L$ with properties of the syntactic monoid $M(\L)$ and syntactic morphism $\eta_\L$~\cite{DBLP:journals/jcss/BarringtonCST92} and with its circuit complexity under a reasonable binary encoding of $\L$'s alphabet~\cite<e.g.,>[Lemma~2.1]{DBLP:journals/actaC/Bernatsky97} and the assumption that $\ACC \ne \NCo$. We also remind the reader that a regular language is $\FO(<)$-definable iff it is star-free~\cite{Straubing94}, and that $\ACz \subsetneqq \ACC \subseteq \NCo$~\cite{Straubing94,DBLP:books/daglib/0028687}.

\begin{table}[th]
\centering
\begin{tabular}{c|c|c}\toprule
definability of $\L$ & algebraic characterisation of $\L$ & circuit complexity\\
\hline
$\FO(<)$& $M(\L)$ is aperiodic & \multirow{2}{*}{in \ACz}\\
\hhline{|-|-|~|}
$\FO(<,\equiv)$& $\eta_\L$ is quasi-aperiodic &  \\
\hline
$\FO(<,\MOD)$ &all groups in $M(\L)$ are solvable &in \ACC\\
\hline
$\FO(\RPR)$ &arbitrary $M(\L)$& in \NCo\\
\hline
\hline
not in $\FO(<,\MOD)$  & $M(\L)$ has an unsolvable group & \NCo-hard\\
\bottomrule
\end{tabular}
\caption{Definability, algebraic characterisations and circuit complexity of regular language $\L$, where $M(\L)$ is the syntactic monoid and $\eta_\L$ the syntactic morphism of $\L$.}
\label{tab:algebra}
\end{table}

We are now in a position to establish the connection between the rewritability of temporal OMQs and definability of regular languages mentioned above.
For any OMQ $\q$ and $\Xi\subseteq\sig(\q)$, we regard $\Sigma_{\Xi} = 2^{\Xi}$ as an \emph{alphabet}.
Any $\Xi$-ABox $\Abox$ can be given as a $\Sigma_{\Xi}$-word $w_\Abox = a_0 \dots a_n$ with $a_i = \{A \mid A(i) \in \Abox\}$.
Conversely, any $\Sigma_{\Xi}$-word $w = a_0 \dots a_n$ gives the ABox $\Abox_{w}$ with $\tem(\Abox_{w}) = [0,n]$ and $A(i) \in \Abox_{w}$ iff $A \in a_i$. The word $\emptyset$ corresponds to $\Abox_\emptyset = \emptyset$ with $\tem(\Abox_\emptyset) = [0,0]$.
%
%
The \emph{language} $\L_\Xi(\q)$ is defined to be the set of $\Sigma_{\Xi}$-words  $w_\Abox$ with a \yes-answer to $\q$ over $\Abox$.
For a specific OMQ $\q(x)$, we take $\Gamma_\Xi=\Sigma_\Xi\cup\Sigma_\Xi'$ with a disjoint copy $\Sigma_{\Xi}'$ of $\Sigma_\Xi$ and represent a pair $(\Abox,i)$ with a $\Xi$-ABox $\Abox$ and $i\in\tem(\Abox)$ as a $\Gamma_\Xi$-word $w_{\Abox,i} = a_0 \dots a_i'\dots a_n$, where $a_i'= \{A' \mid A(i) \in \Abox\}\in\Sigma_\Xi'$ and $a_j = \{A \mid A(j) \in \Abox\}\in\Sigma_\Xi$, for $j\neq i$.
The \emph{language} $\L_\Xi(\q(x))$ is the set of $\Gamma_\Xi$-words $w_{\Abox,i}$ such that $i$ is  a certain answer to $\q(x)$ over $\Abox$. The following is proved similarly to Vardi and Wolper's~\citeyear[Theorem~2.1]{VardiW86}.

\begin{theorem}\label{Prop:rewr-def}
Let $\q = (\TO, \varkappa)$ be a Boolean and $\q(x) = (\TO, \varkappa(x))$ a specific OMQ. Then

$(i)$ both $\L_\Xi(\q)$ and $\L_\Xi(\q(x))$ are regular languages\textup{;}

$(ii)$ for any $\lang \in \{ \FO(<), \FO(<,\equiv), \FO(<,\MOD)\}$ and $\Xi\subseteq\sig(\q)$, the OMQ $\q$ is $\lang$-rewritable over $\Xi$-ABoxes iff $\L_\Xi(\q)$ is $\lang$-definable\textup{;} similarly, $\q(x)$ is $\lang$-rewritable over $\Xi$-ABoxes iff $\L_\Xi(\q(x))$ is $\lang$-definable.
\end{theorem}
\begin{proof}
$(i)$ Let $\sub_{\q}$ (or $\sub_{\TO}$) be the set of temporal concepts occurring in $\q$ (respectively, $\TO$) and their negations. A \emph{type for} $\q$ (respectively, $\TO$) is any maximal subset $\tp \subseteq \sub_{\q}$ (respectively, $\tp \subseteq \sub_{\TO}$) \emph{consistent} with $\TO$ in the sense that all formulas in $\tp$ are true at some point of a model of $\TO$. Let $\Type$ be the set  of all types for $\q$.
%
%
Define an NFA $\A$ over $\Sigma_\Xi$ whose language $\L(\A)$ is $\Sigma_\Xi^* \setminus \L_\Xi(\q)$. Its states are $Q_{\neg \varkappa}=\{\tp \in\Type \mid \neg\varkappa\in\tp\}$. The transition relation $\to_a$, for $a \in \Sigma_\Xi$,  is defined by taking $\tp_1 \to_a \tp_2$ if the following conditions hold:
\begin{description}
\item[\rm (a)] $a \subseteq \tp_2$,

\item[\rm (b)] $\Rnext C \in \tp_1$ iff $C \in \tp_2$, for every $\Rnext C \in \sub_{\q}$, 

\item[\rm (c)] $\Rbox C \in \tp_1$ iff $C \in \tp_2$ and $\Rbox C \in \tp_2$, for every $\Rbox C \in \sub_{\q}$, 

\item[\rm (d)] $\Rdiamond C \in \tp_1$ iff $C \in \tp_2$ or $\Rdiamond C \in \tp_2$, for every $\Rdiamond C \in \sub_{\q}$, 
\end{description}
and symmetrically for the past-time operators.
The initial (accepting) states are those $\tp \in Q_{\neg \varkappa}$ for which $\tp \cup \{\Lbox \neg \varkappa\}$ (and, respectively, $\tp \cup \{\Rbox \neg \varkappa\}$) is consistent with $\TO$.
Then $w \in \L(\A)$ iff $(\TO, \Abox_w) \not\models \exists x \,\varkappa(x)$, for any $w \in \Sigma_\Xi^*$. Indeed, if $w \in \L(\A)$, we take an accepting run $\tp_0,\dots,\tp_n$ of $\A$ on $w$, a model $\I^-$ of $\TO$ with $\I^-,k \models \tp_0 \cup \{\Lbox \neg \varkappa\}$, a model $\I^+$ of $\TO$ with $\I^+,l \models \tp_n \cup \{\Rbox \neg \varkappa\}$, for some $k,l \in \mathbb Z$, and construct a new interpretation $\I$ that has the types $\tp_0,\dots,\tp_n$ in the interval $[0,n]$, before (after) which it has the same types as in $\I^-$ in $(-\infty,k)$ (respectively, $\I^+$ on $(l,\infty)$). One can readily check that $\I$ is a model of $\TO$ and $\Abox_w$ such that $\varkappa^\I = \emptyset$, and so $(\TO, \Abox_w) \not\models \exists x \,\varkappa(x)$. The opposite direction is obvious.

To show that $\L_\Xi(\q(x))$ is regular, we observe first that the language $\L$ over $\Gamma_\Xi$ comprising words of the form $w_{\Abox, i}$, for all non-empty $\Xi$-Aboxes $\Abox$ and $i \in \tem(\Abox)$, is regular. Thus, it suffices to define an NFA $\A$ over $\Gamma_\Xi$ such that $\L_\Xi(\q(x)) = \L \setminus \L(\A)$. The set of states in $\A$ is $\Type \cup \Type'$ with a disjoint copy $\Type'$ of $\Type$. The set of initial states is $\Type$ and the set of accepting states is $\Type'$. The transition relation $\to_a$, for $a \in \Sigma_\Xi$,  is defined by taking $\tp_1 \to_a \tp_2$ if either $\tp_1, \tp_2 \in \Type$ or $\tp_1, \tp_2 \in \Type'$ and conditions (a)--(d) are satisfied; for $a' \in \Sigma_\Xi'$, we set $\tp_1 \to_{a'} \tp_2$ if $\tp_1 \in \Type$, $\tp_2 \in \Type'$, $\neg \varkappa \in \tp_2$, $a' \subseteq \tp_2$, and (b)--(d) hold. It is easy to see that, for any $\Xi$-ABox $\Abox$ and $i \in \tem(\Abox)$, there exists a model $\I$ of $\TO$ and $\Abox$ with $i \not \in \varkappa^\I$ iff $w_{\Abox, i} \in \L(\A)$.

The proof of $(ii)$ is easy and can be found in Appendix~\ref{app:thm5}.
\end{proof}

Note that the number of states in the NFAs in the proof above is $2^{O(|\q|)}$ and that they can be constructed in exponential time in the size $|\q|$ of $\q$ as $\LTL$-satisfiability is in \PSpace.

By Theorem~\ref{Prop:rewr-def}, we can reformulate  the evaluation problem for $\q$ and $\q(x)$ over $\Xi$-ABoxes as the \emph{word problem} for the regular languages $\L_\Xi(\q)$ and $\L_\Xi(\q(x))$.
%
%
Then Table~\ref{tab:algebra} yields the following correspondences between the data complexity of answering and FO-rewritability of Boolean and specific \LTL{} OMQs $\q$:
\begin{description}
\item $\q$ is $\FO(<,\equiv)$-rewritable iff it can be answered in $\ACz$;

\item $\q$ is $\FO(<,\MOD)$-rewritable iff it can be answered in $\ACC$;

\item $\q$ is not $\FO(<,\MOD)$-rewritable iff answering $\q$ in $\NCo$-complete (unless $\ACC = \NCo$);

\item $\q$ is $\FO(<,\RPR)$-rewritable iff it can be answered in $\NCo$\!.
\end{description}


\section{Characterising FO-rewritability of regular languages}\label{sec:groups}

In this section, we show that the algebraic characterisations of FO-definability of $\L(\A)$ in   Table~\ref{tab:algebra}
can be captured by localisable properties
of the transition monoid of $\A$.
Note that Theorem~\ref{DFAcrit}~$(i)$ was already observed by \citeA{DBLP:journals/iandc/Stern85} and used in
proving that $\FO(<)$-definability of $\L(\A)$ is \PSpace-complete \cite{DBLP:journals/iandc/Stern85,DBLP:journals/TCS/ChoHyunh91,DBLP:journals/actaC/Bernatsky97}; criteria $(ii)$ and $(iii)$ of $\FO(<,\equiv)$- and $\FO(<,\MOD)$-definability are novel.

%
%

\begin{theorem}\label{DFAcrit}
For any DFA $\A=(Q,\Sigma,\delta,q_0,F)$, the following criteria hold\textup{:}
\begin{description}
\item[$(i)$] $\L(\A)$ is not $\FO(<)$-definable iff  $\A$ contains a nontrivial cycle, that is, there exist a word $u\in\Sigma^\ast$, a
state $q\in \Qr$, and a number $k\leq|Q|$ such that $q\not\simm \fu_u(q)$ and $q= \fu_{u^k}(q)$\textup{;}

\item[$(ii)$] $\L(\A)$ is not $\FO(<,\equiv)$-definable iff there are words $u,v\in\Sigma^\ast$, a
state $q\in \Qr$, and a number $k\leq |Q|$ such that $q\not\simm\fu_u(q)$, $q=\fu_{u^k}(q)$, $|v|=|u|$, and
$\fu_{u^i}(q)=\fu_{u^iv}(q)$, for every $i<k$\textup{;}


\item[$(iii)$] $\L(\A)$ is not $\FO(<,\MOD)$-definable iff there exist words $u,v\in\Sigma^\ast$, a
state $q\in \Qr$ and numbers $k,l\leq |Q|$ such that $k$ is an odd prime, $l>1$ and coprime to both $2$ and $k$,
$q\not\simm\fu_u(q)$, $q\not\simm\fu_v(q)$, $q\not\simm\fu_{uv}(q)$ and, for all $x\in\{u,v\}^\ast$, we have
$\fu_{x}(q)\simm\fu_{xu^{2}}(q)\simm\fu_{xv^{k}}(q)\simm\fu_{x(uv)^{l}}(q)$.
\end{description}
\end{theorem}


\begin{proof}
We use the algebraic criteria of Table~\ref{tab:algebra} for $\L=\L(\A)$.
Thus,  $M(\L)$ is the transition monoid of the minimal DFA $\A_{\L(\A)}$, whose transition function is denoted by $\fum$.

$(i)~(\Rightarrow)$ Suppose $\G$ is a nontrivial group in $M(\A_{\L(\A)})$.
Let $u\in\Sigma^\ast$ be such that $\fum_u$ is a nonidentity element in $\G$.
We claim that there is $p\in \Qr$ such that $\fum_{u^n}(\simclass{p})\ne\fum_{u^{n+1}}(\simclass{p})$ for any $n>0$.
Indeed, otherwise for every $p\in \Qr$ there is $n_p>0$ with $\fum_{u^{n_p}}(\simclass{p})=\fum_{u^{n_p+1}}(\simclass{p})$.
Let $n=\max\{n_p\mid p\in \Qr\}$. Then $\fum_{u^n}=\fum_{u^{n+1}}$, contrary to~\eqref{idemp}.
By \eqref{fpi}, there is $m\geq 1$ with $\fum_{u^{2m}}(\simclass{p})=\fum_{u^{m}}(\simclass{p})$.
Let $\simclass{s}=\fum_{u^m}(\simclass{p})$. Then $\simclass{s}=\fum_{u^m}(\simclass{s})$, and so the restriction of $\fu_{u^m}$ to the subset $\simclass{s}$
of $\Qr$  is an $\simclass{s}\to \simclass{s}$ function.
By \eqref{fpii},
there are $q\in \simclass{s}$ and $n\geq 1$
with
$(\fu_{u^{m}})^n(q)=q$. Thus, $\fu_{u^{mn}}(q)=q$, and so by \eqref{fpiii}, there is
$k\leq |Q|$ with $\fu_{u^{k}}(q)=q$.
As $\simclass{s}\ne\fum_{u}(\simclass{s})$, we also have $q\not\simm\fu_u(q)$, as required.

$(i)~(\Leftarrow)$
Suppose the condition holds for $\A$.
Then there are $u\in\Sigma^\ast$,
$q\in \simclass{\Qr\!}$, and $k<\omega$ with $q\ne\fum_u(q)$ and $q=\fum_{u^k}(q)$.
So $\fum_{u^n}\ne\fum_{u^{n+1}}$ for any $n>0$. Indeed, otherwise we would have some $n>0$ with $\fum_{u^n}(q)=\fum_{u^{n+1}}(q)$. Let $i,j$ be such that $n=i\cdot k+j$ and $j<k$. Then
\[
q=\fum_{u^k}(q)=\fum_{u^{(i+1)k}}(q)=\fum_{u^n u^{k-j}}(q)=\fum_{u^{n+1}u^{k-j}}(q)=\fum_{u^{(i+1)k}u}(q)=\fum_u(q).
\]
So, by \eqref{ginii} and \eqref{giniii}, $\G_{\fum_u}$ is a nontrivial group in $M(\A_{\L(\A)})$.


$(ii)~(\Rightarrow)$
Let $\G$ be a nontrivial group in $\eta_\L(\Sigma^t)$, for some $t<\omega$, and
let $u\in\Sigma^t$ be such that $\fum_u$ is a nonidentity element in $\G$.
As shown in the proof of $(i)~(\Rightarrow)$,
there exist $s\in \Qr$ and $m\geq 1$ such that $\simclass{s}\ne\fum_{u}(\simclass{s})$ and
$\simclass{s}=\fum_{u^m}(\simclass{s})$.
Now let $v\in\Sigma^t$ be such that $\fum_v$ is the \unit{} element in $\G$, and consider $\fu_v$.
By \eqref{gini}, there is $\ell\geq 1$ such that $\fu_{v^\ell}$ is idempotent.
Then $\fu_{v^{2\ell-1}v^{2\ell}}=\fu_{v^{2\ell-1}}$. Thus, if we let $\bar{u}=uv^{2\ell-1}$ and $\bar{v}=v^{2\ell}$, then
$|\bar{u}|=|\bar{v}|$ and
$\fu_{\bar{u}^i}=\fu_{\bar{u}^i\bar{v}}$ for any $i<\omega$.
Also,  $\fum_{u^i}=\fum_{\bar{u}^i}$ for every $i\geq 1$, and so
the restriction of $\fu_{\bar{u}^m}$ to $\simclass{s}$ is an $\simclass{s}\to \simclass{s}$ function.
By \eqref{fpii},
there exist $q\in \simclass{s}$ and $n\geq 1$ such that
$(\fu_{\bar{u}^{m}})^n(q)=q$.
Thus, $\fu_{\bar{u}^{mn}}(q)=q$, and so by \eqref{fpiii}, there is some $k\leq |Q|$ with $\fu_{\bar{u}^{k}}(q)=q$.
As $\simclass{s}\ne\fum_u(\simclass{s})=\fum_{\bar{u}}(\simclass{s})$, we also have $q\not\simm\fu_{\bar{u}}(q)$, as required.

$(ii)~(\Leftarrow)$
If the condition holds for $\A$, then  there exist
$u,v\in\Sigma^\ast$, $q\in \simclass{\Qr\!}$, and $k<\omega$ such that $q\ne\fum_u(q)$, $q=\fum_{u^k}(q)$, $|v|=|u|$,
and $\fum_{u^i}(q)=\fum_{u^iv}(q)$, for every $i<k$.
As $M(\A_{\L(\A)})$ is finite, it has finitely many subsets.
So there exist $i,j\geq 1$ such that $\eta_\L(\Sigma^{i|u|})=\eta_\L(\Sigma^{(i+j)|u|})$.
Let $z$ be a multiple of $j$ with $i\leq z<i+j$.
Then $\eta_\L(\Sigma^{z|u|})=\eta_\L(\Sigma^{(z|u|)^2})$, and so $\eta_\L(\Sigma^{z|u|})$ is closed under the composition of functions (that is, the semigroup operation of $M(\A_{\L(\A)})$). Let $w=uv^{z-1}$ and consider the group
$\G_{\fum_w}$. Then $G_{\fum_w}\subseteq\eta_\L(\Sigma^{z|u|})$.
We claim that $\G_{\fum_w}$ is nontrivial. Indeed, we have $\fum_w(q)=\fum_{uv^{z-1}}(q)=\fum_u(q)\ne q$.
On the other hand, $\fum_{w^{k}}(q)=\fum_{u^{k}}(q)=q$.
By the proof of $(i)~(\Leftarrow)$, $\G_{\fum_w}$ is nontrivial.

$(iii)~(\Rightarrow)$
Suppose $\G$ is an unsolvable group in $M(\A_{\L(\A)})$. By the Kaplan--Levy criterion,
$\G$ contains three functions $a,b,c$ such that $\ord_\G(a)=2$, $\ord_\G(b)$ is an odd prime,
$\ord_\G(c)>1$ and coprime to both $2$ and $\ord_\G(b)$, and $c\circ b\circ a=e_\G$ for the \unit{} element $e_\G$ of $\G$.
Let $u,v\in\Sigma^\ast$ be such that $a=\fum_u$, $b=\fum_v$ and $c=(\fum_{uv})^-$,
and let $k=\ord_\G(\fum_{v})$ and $r=\ord_\G(c)=\ord_\G(\fum_{uv})$. Then $r>1$ and coprime to both $2$ and $k$.
Let $S=\bigl\{p\in \simclass{\Qr\!}\mid e_\G(p)=p\bigr\}$.
As $\fum_x$ is $\G$ for every $x\in\{u,v\}^\ast$, we have $e_\G\circ\fum_x=\fum_x$.
Thus,
\begin{align*}
& \fum_{xu^2}(q)=\fum_{u^2}\bigl(\fum_x(q)\bigr)=e_\G\bigl(\fum_x(q)\bigr)=(e_\G\circ\fum_x)(q)=\fum_x(q),\quad\mbox{and}\\
& \fum_{xv^k}(q)=\fum_{v^k}\bigl(\fum_x(q)\bigr)=e_\G\bigl(\fum_x(q)\bigr)=(e_\G\circ\fum_x)(q)=\fum_x(q),\quad\mbox{for every $q\in S$}.
\end{align*}
Then, by \eqref{groupini}, each of $\fum_u\!\!\restriction_S$, $\fum_v\!\!\restriction_S$
and $\fum_{uv}\!\!\restriction_S$ is a permutation on $S$.
By \eqref{groupinii}, the order of $\fum_u\!\!\restriction_S$ is $2$, the order of $\fum_v\!\!\restriction_S$ is $k$,
and  the order $l$ of $\fum_{uv}\!\!\restriction_S$ is a $>1$ divisor of $r$, and so it is coprime to both $2$ and $k$.
Also, we have $k,l\leq |S|\leq |Q|$.
Further, for every $x$, if $q$ is in $S$ then $\fum_x(q)\in S$ as well. So we have
\[
\fum_{x(uv)^l}(q)=\fum_{(uv)^l}\bigl(\fum_x(q)\bigr)=(\fum_{uv}\!\!\restriction_S)^l\bigl(\fum_x(q)\bigr)=\id_S\bigl(\fum_x(q)\bigr)=\fum_x(q),\ \ \mbox{for all $q\in S$}.
\]
It remains to show that there is $q\in S$ with $q\ne\fum_u(q)$, $q\ne\fum_u(q)$, and $q\ne\fum_{uv}(q)$.
 %
%
Recall that the length of any cycle in a permutation divides its order.
First, we show there is $q\in S$ with $q\ne\fum_u(q)$ and $q\ne\fum_u(q)$. Indeed, as
$\fum_{u}\!\!\restriction_S\ne\id_S$, there is $q\in S$ such that $\fum_u(q)=q'\ne q$.
As the order of $\fum_{u}\!\!\restriction_S$ is $2$, $\fum_u(q')=q$.
If both $\fum_v(q)=q$ and $\fum_v(q')=q'$ were the case, then $\fum_{uv}(q)=q'$ and $\fum_{uv}(q')=q$ would hold, and so
$(qq')$ would be a cycle in $\fum_{uv}\!\!\restriction_S$, contrary to $l$ being coprime to $2$.
So take some $q\in S$ with $\fum_u(q)=q'\ne q$ and  $\fum_v(q)\ne q$. 
If $\fum_v(q')\ne q$ then $\fum_{uv}(q)\ne q$, and so $q$ is a good choice.
Suppose $\fum_v(q')=q$, and let $q''=\fum_v(q)$. Then $q''\ne q'$, as $k$ is odd. Thus, $\fum_{uv}(q')\ne q'$, and
so $q'$ is a good choice.

$(iii)~(\Leftarrow)$
Suppose $u,v\in\Sigma^\ast$, $q\in \Qr$, and $k,l<\omega$ are satisfying the conditions.
For every $x\in\{u,v\}^\ast$, we define an equivalence relation $\approx_x$ on $\simclass{\Qr\!}$ by taking $p\approx_x p'$ iff
$\fum_x(p)=\fum_x(p')$. Then we clearly have that $\approx_x\subseteq \approx_{xy}$, for all $x,y\in \{u,v\}^\ast$.
As $Q$ is finite, there is $z\in \{u,v\}^\ast$ such that
$\approx_z= \approx_{zy}$ for all $y\in \{u,v\}^\ast$.
Take such a $z$. By \eqref{gini}, $\fum_z^n$ is idempotent for some $n\geq 1$. We let $w=z^n$.
Then $\fum_w$ is idempotent and we also have that
\begin{equation}\label{fp}
\approx_w\,=\, \approx_{wy}\quad\mbox{for all $y\in \{u,v\}^\ast$.}
\end{equation}
Let $G_{\{u,v\}}=\bigl\{\fum_{wxw}\mid x\in\{u,v\}^\ast\bigr\}$. Then $G_{\{u,v\}}$ is closed under composition. Let $\G_{\{u,v\}}$ be the subsemigroup of $M(\A_{\L(\A)})$  with universe $G_{\{u,v\}}$.
Then $\fum_w=\fum_{w\varepsilon w}$ is an \unit{} element in $\G_{\{u,v\}}$.
Let $S=\{p\in \simclass{\Qr\!}\mid \fum_w(p)=p\}$.
We show that
\begin{equation}\label{Sperm}
\mbox{for every $\fum$ in $\G_{\{u,v\}}$, $\fum\!\!\restriction_S$ is a permutation on $S$,}
\end{equation}
and so $\G_{\{u,v\}}$ is a group by \eqref{groupini}.
Indeed, take some $x\in\{u,v\}^\ast$. As
$\fum_w\bigl(\fum_{wxw}(p)\bigr)=\fum_{wxww}(p)=\fum_{wxw}(p)$, for any $p\in \simclass{\Qr\!}$, $\fum_{wxw}\!\!\restriction_S$ is an $S\to S$ function. Also, if $p,p'\in S$ and $\fum_{wxw}(p)=\fum_{wxw}(p')$ then $p\approx_{wxw}p'$. Thus, by \eqref{fp}, $p\approx_w p'$, that is,
$p=\fum_w(p)=\fum_w(p')=p'$, proving \eqref{Sperm}.

We show that $\G_{\{u,v\}}$ is unsolvable by finding an unsolvable homomorphic image of it.
Let $R=\bigl\{p\in \simclass{\Qr\!}\mid p=\fum_x(q)\mbox{ for some }x\in\{u,v\}^\ast\bigr\}$.
We claim that, for every $\fum$ in $\G_{\{u,v\}}$, $\fum\!\!\restriction_R$ is a permutation on $R$, and so the function
$h$ mapping every $\fum$ to $\fum\!\!\restriction_R$ is a group homomorphism from $\G_{\{u,v\}}$ to the group of all permutations on $R$. Indeed, by \eqref{Sperm}, it is enough to show that $R\subseteq S$.
Let $\overline{w}=\overline{z}_{m}\dots\overline{z}_1$, where $w=z_1\dots z_m$ for some $z_i\in\{u,v\}$,
$\overline{u}=u$ and $\overline{v}=v^{k-1}$.
Since $\fum_{x}(q)=\fum_{x(u)^{2}}(q)=\fum_{x(v)^{k}}(q)$ for all $x\in\{u,v\}^\ast$,
we obtain that
\begin{multline}\label{barw}
\fum_{yw\overline{w}}(q)=\fum_{\overline{z}_{m-1}\dots\overline{z}_1}\bigl(\fum_{yz_1\dots z_m\overline{z}_m}(q)\bigr)
=\fum_{\overline{z}_{m-1}\dots\overline{z}_1}\bigl(\fum_{yz_1\dots z_{m-1}}(q)\bigr)=\dots\\
\dots =\fum_{\overline{z}_1}\bigl(\fum_{yz_1}(q)\bigr)=\fum_{xz_1\overline{z}_1}(q)=\fum_y(q),\quad
\mbox{for all $y\in\{u,v\}^\ast$.}
\end{multline}
Now suppose $p\in R$, that is, $p=\fum_x(q)$ for some $x\in \{u,v\}^\ast$. Then, by \eqref{barw},
\[
\fum_w(p)=\fum_w\bigl(\fum_x(q)\bigr)=\fum_{xw}(q)=\fum_{xww\overline{w}}(q)=\fum_{xw\overline{w}}(q)=\fum_x(q)=p,
\]
and so $p\in S$, as required.

Now let $\G$ be the image of $\G_{\{u,v\}}$ under $h$.
We prove that $\G$ is unsolvable by finding three elements $a,b,c$ in it such that $\ord_\G(a)=2$, $\ord_\G(b)=k$,
$\ord_\G(c)$ is coprime to both $2$ and $\ord_\G(b)$, and $c\circ b\circ a=\id_R$  (the \unit{} element of $\G$).
So let $a=h(\fum_{wuw})$,  $b=h(\fum_{wvw})$, and $c=h(\fum_{wuvw})^-$.
Observe that, for every $x\in\{u,v\}^\ast$, $h(\fum_{wxw})=\fum_x\!\!\restriction_R$, and so $c\circ b\circ a=\id_R$.
Also, for any $\fum_x(q)\in R$, $a^2\bigl(\fum_x(q)\bigr)=(\fum_u\!\!\restriction_R)^2\bigl(\fum_x(q)\bigr)=\fum_{xu^2}(q)=\fum_x(q)$
by our assumption, so $a^2=\id_R$. On the other hand, $q\in R$ as $\fum_\varepsilon(q)=q$, and
$\id_R(q)=q\ne\fum_u(q)$ by assumption, so $a\ne\id_R$.  As $\ord_\G(a)$ divides $2$, $\ord_\G(a)=2$ follows.
Similarly, we can show that $\ord_\G(b)=k$ (using that $\fum_{xv^k}(q)=\fum_x(q)$ for every $x\in\{u,v\}^\ast$, and $u\ne\fum_v(q)$). Finally (using that $\fum_{x(uv)^l}(q)=\fum_x(q)$ for every $x\in\{u,v\}^\ast$, and $u\ne\fum_{uv}(q)$),
we obtain that $h(\fum_{wuvw})^l=\id_R$ and $h(\fum_{wuvw})\ne\id_R$.
Therefore, it follows that $\ord_\G(c)=\ord_\G\bigl(h(\fum_{wuvw})^-\bigr)=\ord_\G\bigl(h(\fum_{wuvw})\bigr)>1$ and divides $l$, and so coprime to both $2$ and $k$, as required.
\end{proof}

The following technical observation will be used in Sections~\ref{sec:LTL-genearal} and~\ref{sec:linear}; its proof is given in Appendix~\ref{app:lem7}.

\begin{lemma}\label{expandL}
Suppose $\lang \in \{ \FO(<), \FO(<,\equiv), \FO(<,\MOD)\}$ and $\Sigma$, $\Gamma$ and $\Delta$ are alphabets such that $\Sigma\cup\{x,y\}\subseteq\Gamma\subseteq\Delta$, for some $x,y \notin \Sigma$. Then a regular language $\L$ over $\Sigma$ is $\lang$-definable iff the regular language
$
\L'=\{w_1xwyw_2\mid w\in\L,\ w_1,w_2\in\Gamma^\ast\}
$
%
 is $\lang$-definable over $\Delta$.
\end{lemma}

\section{Deciding FO-definability of regular languages: \PSpace-hardness}\label{sec:reglang}

\citeA{Kozen77} showed that deciding non-emptiness of the intersection of the languages recognised by a set of given deterministic DFAs is \PSpace-complete. By carefully analysing Kozen's lower bound proof and using the 
criterion of Theorem~\ref{DFAcrit}~$(i)$, \citeA{DBLP:journals/TCS/ChoHyunh91} established that
deciding $\FO(<)$-definability of $\L(\A)$, for any given minimal DFA $\A$, is \PSpace-hard. We generalise their construction and use the criteria in Theorem~\ref{DFAcrit}~$(ii)$--$(iii)$ to cover $\FO(<,\equiv)$- and $\FO(<,\MOD)$-definability as well.

\begin{theorem}\label{DFAhard}
For any $\lang \in \{ \FO(<), \FO(<,\equiv), \FO(<,\MOD)\}$,  deciding $\lang$-defina\-bi\-lity of the language $\L(\A)$ of a given minimal DFA $\A$ is \PSpace-hard.
\end{theorem}
\begin{proof}
Let $\M$ be a deterministic Turing machine that decides a language using at most $\ppn=P_{\M}(n)$ tape cells on any input of size $n$, for some polynomial $P_{\M}$. Given such an $\M$ and an input $\boldsymbol{x}$,
our aim is to define three minimal DFAs whose languages are, respectively, $\FO(<)$-, $\FO(<,\equiv)$-, and $\FO(<,\MOD)$-definable iff $\M$ rejects $\boldsymbol{x}$, and whose sizes are polynomial in $\ppn$ and the size $|\M|$ of $\M$.


Suppose $\M =(Q, \Tmabc, \Tmtran, \B, q_0, \qa)$ with a set $Q$ of states, tape alphabet $\Tmabc$ with $\B$ for blank, transition function $\Tmtran$, initial state $q_0$ and accepting state $\qa$.
Without loss of generality we assume that
$\M$ erases the tape before accepting, its head is at the left-most cell in an accepting configuration, 
and if $\M$ does not accept the input, it runs forever.
Given an input word $\boldsymbol{x}=x_1\dots x_n$ over $\Tmabc$,
we represent configurations $\conf$ of the computation of $\M$ on $\boldsymbol{x}$ by
the $\ppn$-long word written on the tape (with sufficiently many blanks at the end) in which the symbol $y$ in the active cell is replaced by the pair $(q,y)$ for the current state $q$.
The accepting computation of $\M$ on $\boldsymbol{x}$ is encoded by a word
$\sharp\, \conf_1 \, \sharp \, \conf_2 \, \sharp\,  \dots \, \sharp \, \conf_{k-1} \, \sharp \, \conf_{k} \flat$ over
the alphabet $\Aiabc=\Tmabc\cup(Q\times\Tmabc)\cup\{\sharp,\flat\}$,
with $\conf_1,\conf_2,\dots,\conf_k$ being the subsequent configurations. In particular, $\conf_1$ is the initial configuration on $\boldsymbol{x}$ (so it is of the form $(q_0,x_1)x_2\dots x_n\B\dots\B$), and $\conf_k$ is the accepting configuration (so it is of the form $(\qa,\B)\B\dots\B$).
As usual for this representation of computations, we may regard $\Tmtran$ as a partial function
from  $\bigl(\Tmabc\cup(Q\times\Tmabc)\cup\{\sharp\}\bigr)^3$ to $\Tmabc\cup(Q\times\Tmabc)$ with $\Tmtran(\sigma^j_{i-1},\sigma^j_i,\sigma^j_{i+1})=\sigma^{j+1}_i$ for each $j< k$, where $\sigma^j_i$ is the $i$th symbol of $\conf^j$.

%
%

Let $p_{\M,\boldsymbol{x}}=\pp$ be the first prime such that $\pp\geq\ppn+2$ and $\pp\not\equiv\pm 1\ (\text{mod}\ 10)$. By Corollary 1.6 of~\citeA{Illin2018}, $\pp$ is polynomial in $\ppn$.  
Our first aim is to define a $p+1$-long sequence of disjoint minimal DFAs $\A_i$ over $\Aiabc$.
Each $\A_i$ has size polynomial in $\ppn$, $|\M|$, and is constructible in logarithmic space; it
checks certain properties of an accepting computation on $\boldsymbol{x}$ such that
$\M$ accepts $\boldsymbol{x}$ iff the intersection of the $\L(\A_i)$ is not empty and
consists of the single word encoding the accepting computation on $\boldsymbol{x}$.


Formally, we define each $\A_i$ as an NFA but bear in mind that it can standardly be turned to a DFA by adding to it a `trash state' $\tst_i$
looping on itself with every character $\sigma\in\Aiabc$, and also adding the missing transitions that all lead to the trash state $\tst_i$.
The DFA $\A_{0}$ checks that an input starts with the initial configuration on $\boldsymbol{x}$ and ends with the accepting configuration:\\[3pt]
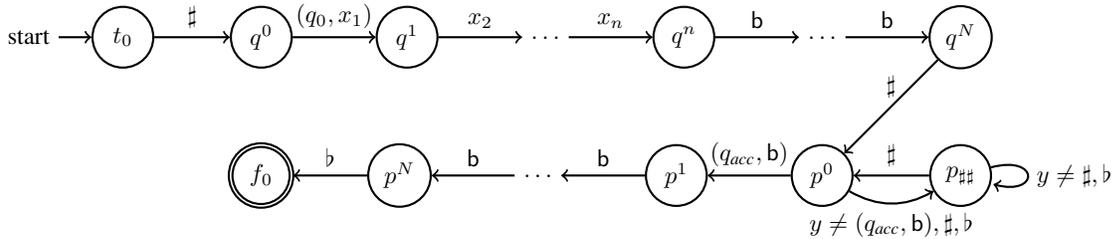
\centerline{
\begin{tikzpicture}[->,thick,scale=0.92,node distance=2cm, transform shape]
\node[state, initial] (1) {$t_0$};
\node[state, right of  =1] (2) {$q^0$};
\node[state, right  =1.2cm of 2] (3) {$q^1$};
\node[right  of= 3](4){$\ldots$};
\node[state, right  of= 4] (5) {$q^{n}$};
\node[right of=5](6){$\ldots$};
\node[state, right  of= 6] (7) {$q^{\ppn}$};
\node[state, below of= 7] (8) {$p_{\sharp\sharp}$};
\node[state,left of= 8](9){$p^0$};
\node[state, left =1.2cm of 9] (10) {$p^1$};
\node[left of=10](11){$\ldots$};
\node[state, left of= 11] (12) {$p^{\ppn}$};
\node[state, accepting, left of =12] (f) {$f_0$};
\draw (1) edge[above] node{$\sharp$} (2)
(2) edge[above] node{$(q_0,x_1)$} (3)
(3) edge[above] node{$x_2$} (4)
(4) edge[above] node{$x_n$} (5)
(5) edge[above] node{$\B$} (6)
(6) edge[above] node{$\B$} (7)
(8) edge [ loop right ] node{$y\neq\sharp,\flat$} (8)
(7) edge[above] node{$\sharp$} (9)
(9) edge[above] node{$(\qa,\B)$} (10)
(8) edge[above] node{$\sharp$} (9)
(9) edge[bend right,below] node{$y\neq (\qa,\B),\sharp,\flat$} (8)
(10) edge[above] node{$\B$} (11)
(11) edge[above] node{$\B$} (12)
(12) edge[above] node{$\flat$} (f)
;
\end{tikzpicture}
}\\
When $1 \le i \le \ppn$,
the DFA $\A_{i}$ checks, for all $j< k$, whether the $i$th symbol of $\conf^j$ changes `according to $\Tmtran$' in passing to $\conf^{j+1}$. The non-trash part of its transition function $\delta^i$ is as follows, for $1<i<N$. (For $i=1$ and $i=N$,  some adjustments are needed.) For all $u,u',v,w,w',y,z\in\Tmabc\cup(Q\times\Tmabc)$,
\begin{align*}
& \delta^i_\sharp(t_i)=q^{i-1},\quad
\delta^i_u(q^j)=q^{j-1},\ \mbox{for $2\le j\le i-1$,}\quad
\delta^i_u(q^{1})=r_u,\quad
\delta^i_v(r_u)=r_{uv},\\
& \delta^i_w(r_{uv})=q^{i+1}_{\Tmtran(u,v,w)},\quad
 \delta^i_y(q^j_z)=q^{j+1}_z,\ \mbox{for $i+1\le j\le N$}, \quad \delta^i_\sharp(q^{N}_z)=p^{i-1}_z\\
& \delta^i_y(p^j_z)=p^{j-1}_z, \mbox{ for $2\le j\le i-1$},\
\delta^i_\flat(q^{N}_z)=f_i,\ \
\delta^i_{u'}(q^{1}_z)=p_{u'z},\ \
\delta^i_z(p_{u'z})=r_{u'z};\\
& \mbox{see below, where $z=\Tmtran(u,v,w)$ and $z'=\Tmtran(u',z,w')$:}
\end{align*}
%
\\
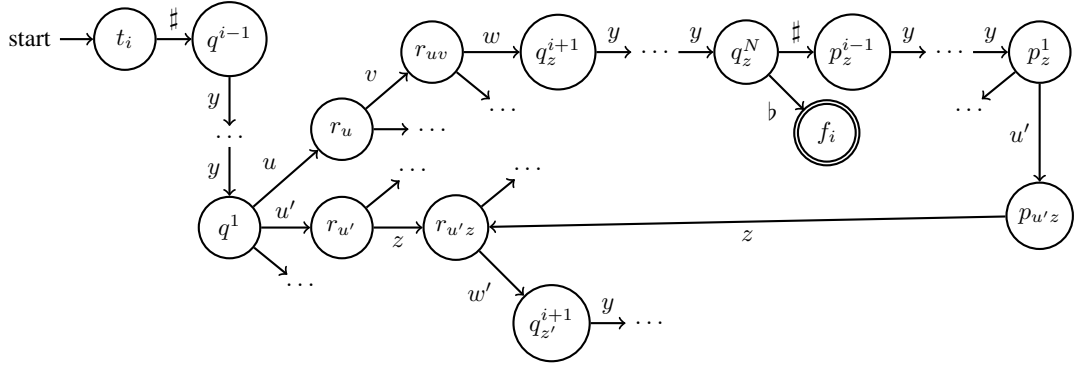
\centerline{
\begin{tikzpicture}[->,thick,scale=0.93,node distance=2cm, transform shape]
\node[state, initial] 	(ti) 	{$t_i$};
\node[state, right=5mm of ti]  (q0)  {$q^{i-1}$};
\node[below =7mm of q0]			(dotsq0)	{$\ldots$};
\node[state,below =7mm of dotsq0] 	(qi2) 	{$q^{1}$};
\node[below right=5mm of  qi2]	(dotsqi2) {$\ldots$};
\node[state, right=7mm of qi2]  (rup)  {$r_{u'}$};
\node[above right=5mm of rup]  (dotsrup)  {$\ldots$};
\node[state, above=5mm of  rup]  (ru)  {$r_{u}$};
\node[state, right=7mm of rup]  (rupz)  {$r_{u'z}$};
\node[above right=5mm of rupz]  (dotsrupz)  {$\ldots$};
\node[state, below right=9mm of rupz]  (qzp0)  {$q_{z'}^{i+1}$};
\node[right=5mm of qzp0]  (dotsqzp0)  {$\ldots$};
\node[right=5mm of ru]  (dotsru)  {$\ldots$};
\node[state, above=5mm of  dotsru]  (ruv)  {$r_{uv}$};
\node[below right=5mm of ruv]  (dotsruv)  {$\ldots$};
\node[state, right=8mm of ruv]  (qz0)  {$q_{z}^{i+1}$};
\node[right=5mm of qz0]  (dotsqz0)  {$\ldots$};
\node[state, right=5mm of dotsqz0]  (qzNi1)  {$q_{z}^{N}$};
\node[state, right=5mm of qzNi1]  (qzNi)  {$p_{z}^{i-1}$};
\node[right=5mm of qzNi]  (dotsqzNi)  {$\ldots$};
\node[state, right=5mm of dotsqzNi]  (qzN2)  {$p_{z}^{1}$};
\node[below left=5mm of  qzN2]	(dotsqzN2) {$\ldots$};
\node[state, below=14mm of qzN2]  (pupz)  {$p_{u'z}$};
\node[state,accepting, below right=7mm of qzNi1] (fi) {$f_i$};
\draw
(ti) 	edge[above] 	node{$\sharp$}  (q0)
(q0) 	edge[left]	 node{$y$} 	(dotsq0)
(dotsq0) 	edge[left]	 node{$y$} 	(qi2)
(qi2) 	edge[above] 	node{$u'$}  (rup)
(qi2) edge[above] 	node{$$}  (dotsqi2)
(rup) edge[below] 	node{$z$}  (rupz)
(rup) edge[above] 	node{$$}  (dotsrup)
(rupz) edge[above] 	node{$$}  (dotsrupz)
(rupz) edge[below left] 	node{$w'$}  (qzp0)
(qzp0) edge[above] 	node{$y$}  (dotsqzp0)
(qi2) 	edge[above left] 	node{$u$}  (ru)
(ru) edge[above] 	node{$$}  (dotsru)
(ru) 	edge[above left] 	node{$v$}  (ruv)
(ruv) 	edge[above] 	node{$w$}  (qz0)
(ruv) edge[above] 	node{$$}  (dotsruv)
(qz0) edge[above] 	node{$y$}  (dotsqz0)
(dotsqz0) 	edge[above] node{$y$}  (qzNi1)
(qzNi1) edge[above] 	node{$\sharp$}  (qzNi)
(qzNi) edge[above] 	node{$y$}  (dotsqzNi)
(dotsqzNi) edge[above] node{$y$}  (qzN2)
(qzN2) edge[left] node{$u'$}  (pupz)
(qzN2) edge[above] 	node{$$}  (dotsqzN2)
(pupz) edge[below] 	node{$z$}  (rupz)
(qzNi1) edge[below left] node{$\flat$}  (fi)
;
\end{tikzpicture}
}
\\
Finally, if $\ppn+1\le i\le \pp$  then $\A_i$ accepts all words over $\Aiabc$ with a single occurrence of $\flat$, which is the input's last character:\\
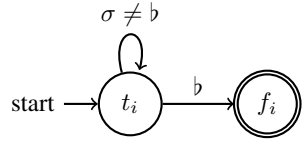
\centerline{
\begin{tikzpicture}[->,thick,scale=0.95, transform shape]
\node[state, initial] (1) {$t_i$};
\node[state, accepting, right = 1cm of  1] (2) {$f_i$};
\draw (1) edge[loop above] node{$\sigma\neq\flat$} (1)
(1) edge[above] node{$\flat$} (2)
;
\end{tikzpicture}
}\\
Note that $\A_{\pp-1}=\A_\pp$ as $\pp\geq\ppn+2$.
It is not hard to check that each $\A_i$ is a minimal DFA that does not contain nontrivial cycles and
the following holds:

\begin{lemma}\label{l:DFAs}
$\M$ accepts $\boldsymbol{x}$ iff $\bigcap_{i=0}^{\pp} \L(\A_i) \ne \emptyset$, in which case this language consists of a single word that encodes the accepting computation of $\M$ on $\boldsymbol{x}$.
%
\end{lemma}


Next, we require three sequences of DFAs $\aut^p_<$, $\aut^p_\equiv$ and $\aut^p_\MOD$, where $p>5$ is a prime number with  $p\not\equiv\pm 1\ (\text{mod}\ 10)$; see the picture below for $p=7$:\\[4pt]
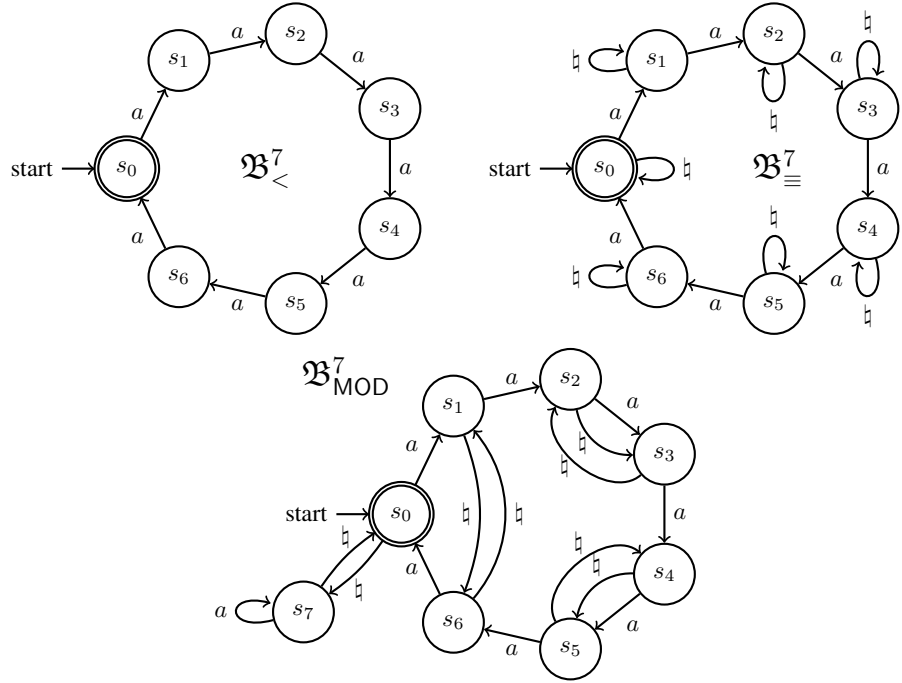
\centerline{
\begin{tikzpicture}[->,thick,scale=0.92,node distance=2cm, transform shape]
  \node[state,initial,accepting] (0)at (180:2cm) {$s_0$};
  \foreach \x in {1,...,6}{%
    \pgfmathparse{(-\x)*(360/7)+180}]
    \node[state] (\x) at (\pgfmathresult:2cm) {$s_\x$};
  }
 \draw (0) edge[left] node{$a$} (1)
(1) edge[above] node{$a$} (2)
(2) edge[above right] node{$a$} (3)
(3) edge[right] node{$a$} (4)
(4) edge[below right] node{$a$} (5)
(5) edge[below] node{$a$} (6)
(6) edge[below left] node{$a$} (0)
;
\node (0) at (0,0) {\Large $\aut^7_<$};
\end{tikzpicture}
\qquad
\begin{tikzpicture}[->,thick,scale=0.92,node distance=2cm, transform shape]
  \node[state,initial,accepting] (0)at (180:2cm) {$s_0$};
  \foreach \x in {1,...,6}{%
    \pgfmathparse{(-\x)*(360/7)+180}
    \node[state] (\x) at (\pgfmathresult:2cm) {$s_\x$};
  }
 \draw (0) edge[left] node{$a$} (1)
 (0) edge[loop right,right] node{$\natural$}(0)
(1) edge[above] node{$a$} (2)
(1) edge[loop left,left] node{$\natural$}(1)
(2) edge[above right] node{$a$} (3)
(2) edge[loop below,below] node{$\natural$}(2)
(3) edge[right] node{$a$} (4)
(3) edge[loop above,above] node{$\natural$}(3)
(4) edge[below right] node{$a$} (5)
(4) edge[loop below,below] node{$\natural$}(4)
(5) edge[below] node{$a$} (6)
(5) edge[loop above,above] node{$\natural$}(5)
(6) edge[below left] node{$a$} (0)
(6) edge[loop left,left] node{$\natural$}(6)
;
\node (0) at (0.5,0) {\Large$\aut^7_\equiv$};
\end{tikzpicture}
}
\\[4pt]
\centerline{
\begin{tikzpicture}[->,thick,scale=0.92,node distance=2cm, transform shape]
 \node[state,initial,accepting] (0)at (180:2cm) {$s_0$};
  \foreach \x in {1,...,6}{%
    \pgfmathparse{(-\x)*(360/7)+180}
    \node[state] (\x) at (\pgfmathresult:2cm) {$s_\x$};
  }
  \node[state, below left of =0] (7){$s_7$};
\draw (0) edge[above left] node{$a$} (1)
(0) edge[bend left=10,below] node{$\natural$} (7)
(1) edge[above] node{$a$} (2)
(1) edge[bend left=20,left] node{$\natural$} (6)
(2) edge[above right] node{$a$} (3)
(2) edge[bend right=40,left] node{$\natural$} (3)
(3) edge[right] node{$a$} (4)
(3) edge[bend left=80,left] node{$\natural$} (2)
(4) edge[below right] node{$a$} (5)
(4) edge[bend right=40,above] node{$\natural$} (5)
(5) edge[below] node{$a$} (6)
(5) edge[bend left=80,above] node{$\natural$} (4)
(6) edge[left] node{$a$} (0)
(6) edge[bend right=40,right] node{$\natural$} (1)
(7) edge[loop left,left] node{$a$} (7)
(7) edge[bend left=10,above ] node{$\natural$} (0)
;
\node (0) at (-2.8,2) {\Large$\aut^7_\MOD$};
\end{tikzpicture}
}\\[2pt]
In general, the first sequence
is $\aut^p_< = \bigl(\{s_{i} \mid i < p \},\{a\},\perm{\aut}{<}{p},s_0,\{s_0\} \bigr)$, where
$\perm{\aut}{<}{p}_a(s_i)=s_j$ if $i,j<p$ and $j\equiv i+1\ (\text{mod}\ p)$.
%
%
Then $\L(\aut^p_<)$ comprises all words of the form $(a^{p})^\ast$,
$\aut^p_<$ is the minimal DFA for $\L(\aut^p_<)$, and the syntactic monoid $M(\aut^p_<)$
 is the cyclic group of order $p$ (generated by the permutation $\smash{\perm{\aut}{<}{p}_a}$).

The second sequence is
%
$\aut^p_\equiv = \bigl(\{s_{i} \mid i < p \},\{a,\natural\},\perm{\aut}{\equiv}{p},s_0,\{s_0\} \bigr)$, where
$\perm{\aut}{\equiv}{p}_\natural(s_i)=s_i$ and
$\perm{\aut}{\equiv}{p}_a(s_i)=s_j$ if $i,j<p$ and $j\equiv i+1\ (\text{mod}\ p)$.
One can check that $\L(\aut^p_\equiv)$ comprises all words of $a$'s and $\natural$'s
where the number of $a$'s is divisible by $p$,
$\aut^p_\equiv$ is the minimal DFA for $\L(\aut^p_\equiv)$, and $M(\aut^p_\equiv)$
 is the cyclic group of order $p$ (generated by the permutation $\perm{\aut}{\equiv}{p}_a$).

%
%
%
%
%
%
%

The third sequence is
%
%
$\aut^p_\MOD = \bigl(\{s_{i} \mid i \leq p \},\{a,\natural\},\perm{\aut}{\MOD}{p},s_0,\{s_0\} \bigr)$, where
\begin{itemize}
\item
$\perm{\aut}{\MOD}{p}_a(s_p)=s_p$, and
$\perm{\aut}{\MOD}{p}_a(s_i)=s_j$ if $i,j<p$ and $j\equiv i+1\ (\text{mod}\ p)$;

\item
$\perm{\aut}{\MOD}{p}_\natural(s_0)=s_p$, $\perm{\aut}{\MOD}{p}_\natural(s_p)=s_0$, and
$\perm{\aut}{\MOD}{p}_\natural(s_i)=s_j$ whenever $1\leq i,j<p$ and $i\cdot j\equiv p-1\ (\text{mod}\ p)$,
that is, $j = -1/i$ in the finite field $\mathbb F_p$.
\end{itemize}
One can check that $\aut^p_\MOD$ is the minimal DFA for its language,
and the syntactic monoid $M(\aut^p_\MOD)$ is the permutation group generated by
$\perm{\aut}{\MOD}{p}_a$ and $\perm{\aut}{\MOD}{p}_\natural$.

\begin{lemma}\label{algebralemma}
For any prime $p>5$ with $p\not\equiv\pm 1\ (\text{mod}\ 10)$, the group $M(\aut^p_\MOD)$ is unsolvable, but all of its proper subgroups are solvable.

%
\end{lemma}
\begin{proof}
It is readily seen that the order of the permutation $\perm{\aut}{\MOD}{p}_\natural$ is $2$, that of $\perm{\aut}{\MOD}{p}_a$ is $p$, while the order of the inverse of $\perm{\aut}{\MOD}{p}_{\natural a}$ is the same as the
order of $\perm{\aut}{\MOD}{p}_{\natural a}$, which is $3$. So $M(\aut^p_\MOD)$ is unsolvable, for any prime $p$,
by the Kaplan--Levy criterion.
To prove that all proper subgroups of $M(\aut^p_\MOD)$ are solvable,
we show that $M(\aut^p_\MOD)$ is a subgroup of the
\emph{projective special linear group} $\text{\sc PSL}_2(p)$. If $p$ is a prime with $p>5$ and $p\not\equiv\pm 1\ (\text{mod}\ 10)$, then
all proper subgroups of $\text{\sc PSL}_2(p)$ are solvable~\cite<e.g.,>[Theorem~2.1]{DBLP:conf/bcc/King05}.
(So $M(\aut^p_\MOD)$  is in fact isomorphic to the unsolvable group $\text{\sc PSL}_2(p)$.)
Consider the set $P=\{0,1,\dots,p-1,\infty\}$ of all points of the projective line over the field $\mathbb F_p$.
By identifying $s_i$ with $i$ for $i<p$, and $s_p$ with $\infty$,
we may regard the elements of $M(\aut^p_\MOD)$ as $P\to P$ functions.
The group $\text{\sc PSL}_2(p)$ consists of all $P\to P$ functions of the form
$i\mapsto\frac{w\cdot i + x}{y\cdot i+z}$,
where $w\cdot z-x\cdot y=1$,  with the field arithmetic of $\mathbb F_p$ extended by $i+\infty=\infty$ for any $i\in P$, $0\cdot\infty=1$ and $i\cdot\infty=\infty$ for $i\ne 0$.
%
%
The two generators of $M(\aut^p_\MOD)$ are in $\text{\sc PSL}_2(p)$: take
$w=1$, $x=1$, $y=0$, $z=1$ for $\perm{\aut}{\MOD}{p}_a$,
and $w=0$, $x=1$, $y=p-1$, $z=0$ for $\perm{\aut}{\MOD}{p}_\natural$.
\end{proof}




Finally, we
define automata $\A_<$, $\A_\equiv$, $\A_\MOD$ over the tape alphabet
$\Aabc = \Aiabc\cup\{a_1,a_2,\natural\}$, where $a_1,a_2$ are fresh symbols.
We take, respectively,  $\aut^{\pp}_<$, $\aut^{\pp}_\equiv$, $\aut^{\pp}_\MOD$  and replace each transition $s_i\to_a s_j$ in them by a fresh copy of $\A_i$, for $i \le \pp$, as shown in the picture below:\\[3pt]
%
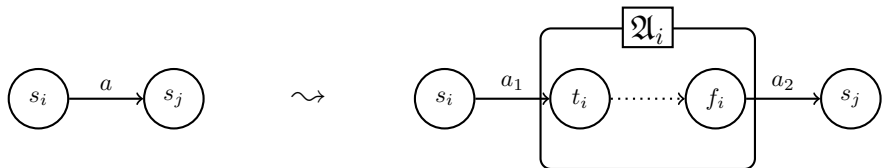
\centerline{
\begin{tikzpicture}[->,thick,scale=0.9,node distance=2cm, transform shape]
\node[state] (si) {$s_i$};
\node[state, right  of =si] (sj) {$s_{j}$};
\node[right  of= sj] (maps) {\Large\bf$\leadsto$};
\node[state, right  of= maps] (si2) {$s_i$};
\node[state, right of=si2] (q) {$t_i$};
\node[state,right of= q](f){$f_i$};
\node[state, right of =f] (sj2) {$s_j$};
\node[fit = (q)(f), basic box = black, header = $\A_i$] (A) {};
\draw (si) edge [above] node{$a$} (sj)
(si2) edge [above] node{$a_1$} (q)
(f) edge [above] node{$a_2$} (sj2)
(q) edge [dotted,above] node{$ $} (f)
;
\end{tikzpicture}
}\\
We make $\A_<$, $\A_\equiv$, $\A_\MOD$ deterministic by adding a trash state $\tst$
looping on itself with every $y\in\Aabc$, and adding the missing transitions leading to $\tst$.
It follows that $\A_<$, $\A_\equiv$, $\A_\MOD$ are minimal DFAs of size polynomial in $\ppn$ and $|\M|$, which can clearly be constructed in logarithmic space.

\begin{lemma}\label{l:empty}
$(i)$ $\L(\A_<)$ is $\FO(<)$-definable iff $\bigcap_{i=0}^{\pp} \L(\A_i)= \emptyset$.

$(ii)$ $\L(\A_\equiv)$ is $\FO(<,\equiv)$-definable iff $\bigcap_{i=0}^{\pp} \L(\A_i)= \emptyset$.

$(iii)$ $\L(\A_\MOD)$ is $\FO(<,\MOD)$-definable iff  $\bigcap_{i=0}^{\pp} \L(\A_i)= \emptyset$.
\end{lemma}
\begin{proof}
As $\A_<,\A_\equiv,\A_\MOD$ are minimal, we can replace $\simm$ by $=$ in the conditions
of Theorem~\ref{DFAcrit}.
For the ($\Rightarrow$) directions, given some  $w\in\bigcap_{i=0}^{\pp} \L(\A_i) $, in each case we show how to satisfy the corresponding condition of Theorem~\ref{DFAcrit}:
$(i)$ take $u=a_1wa_2$, $q=s_0$, and $k=\pp$;
$(ii)$ take $u=a_1wa_2$, $v=\natural^{|u|}$, $q=s_0$, and $k=\pp$;
$(iii)$ take $u=\natural$,  $v=a_1wa_2$, $q = s_0$, $k = \pp$ and $l = 3$.

$(\Leftarrow)$ We show that the corresponding condition of Theorem~\ref{DFAcrit}
implies non-emptiness of $\bigcap_{i=0}^{\pp} \L(\A_i) $.
To this end, we define a
$\Aabc^\ast\to\{a,\natural\}^\ast$ homomorphism by taking $h(\natural)=\natural$, $h(a_1)=a$, and $h(b)=\varepsilon$ for all other $b\in\Aabc$.

\smallskip
$(i)$ and $(ii)$: Let $\circ\in\{<,\equiv\}$ and suppose $q$ is a state in $\A^\pp_\circ$ and $u'\in\Aabc^\ast$ such that $q\ne\perm{\A}{\circ}{\pp}_{u'}(q)$ and
$q=\perm{\A}{\circ}{\pp}_{(u')^k}(q)$ for some $k$.
Let $S=\{s_0,s_1,\dots,s_{\pp-1}\}$.
 We claim that there exist $s\in S$ and $u\in \Aabc^\ast$ such that
 \begin{align}
 \label{unotid}
 & s\ne\perm{\A}{\circ}{\pp}_{u}(s),
 \end{align}
 \begin{align}
 \label{uins}
 & \perm{\A}{\circ}{\pp}_{x}(s)\in S,\quad\mbox{for every $x\in\{u\}^\ast$.}
 \end{align}
 Indeed, observe that none of the states along the cyclic $q\to_{(u')^k} q$ path $\Pi$ in $\A^\pp_\circ$ is $\tst$.
 So there is some state along $\Pi$ that is in $S$, as otherwise one of the $\A_i$ would contain a nontrivial cycle.
 Therefore, $u'$ must be of the form $w\natural^n a_1w'$ for some $w\in\Aiabc^\ast$, $n<\omega$ and $w'\in\Aabc^\ast$.
 It is easy to see that $s=\perm{\A}{\circ}{\pp}_{(u')^{k-1}w}(q)$ and $u=\natural^na_1w'w$ is as required in \eqref{unotid} and \eqref{uins}.

 As $M(\aut^\pp_\circ)$ is a finite group, $\bigl\{\perm{\aut}{\circ}{\pp}_{h(x)}\mid x\in\{u\}^\ast\bigr\}$
 forms a subgroup $\G$ in it (the subgroup generated by $\perm{\aut}{\circ}{\pp}_{h(u)}$).
 We show that $\G$ is nontrivial by finding its nontrivial homomorphic image.
 By \eqref{uins}, for any $x\in\{u\}^\ast$, the restriction $\perm{\A}{\circ}{\pp}_{x}\!\!\restriction_{S'}$ of
 $\perm{\A}{\circ}{\pp}_{x}$ to
 $S'=\bigl\{ \perm{\A}{\circ}{\pp}_{y}(s)\mid y\in\{u\}^\ast\bigr\}$ is an $S'\to S'$ function and
 $\perm{\A}{\circ}{\pp}_{x}\!\!\restriction_{S'}=\perm{\aut}{\circ}{\pp}_{h(x)}\!\!\restriction_{S'}$.
As  $M(\aut^\pp_\circ)$ is a group of permutations on a set containing $S'$,
$\perm{\aut}{\circ}{\pp}_{h(x)}\!\!\restriction_{S'}$ is a permutation of $S'$, for every $x\in\{u\}^\ast$.
Thus, $\bigl\{\perm{\aut}{\circ}{\pp}_{h(x)}\!\!\restriction_{S'}\mid x\in\{u\}^\ast\bigr\}$ is a homomorphic image of $\G$
that is nontrivial by~\eqref{unotid}.

As $\G$ is a nontrivial subgroup of the cyclic group $M(\aut^\pp_\circ)$ of order $\pp$ and $\pp$ is a prime, $\G=M(\aut^\pp_\circ)$. Then there is $x\in\{u\}^\ast$ with $\perm{\aut}{\circ}{\pp}_{h(x)}=\perm{\aut}{\circ}{\pp}_{a}$ (a permutation containing the $\pp$-cycle $(s_0 s_1\dots s_{\pp-1})$ `around' all elements of $S$), and so $S'=S$
and $x=\natural^n a_1wa_2w'$ for some $n<\omega$, $w\in\Aiabc^\ast$, and $w'\in\Aabc^\ast$.
As $n=0$ when $\circ=<$ and $\perm{\A}{\equiv}{\pp}_{\natural^n}(s)$ for every $s\in S$,
$S'=S$ implies that $w\in\bigcap_{i=0}^{\pp-1} \L(\A_i)=\bigcap_{i=0}^{\pp} \L(\A_i)$.

$(iii)$ Suppose $q$ is a state in $\A^\pp_\MOD$ and $u',v'\in\Aabc^\ast$ such that
$q\ne\perm{\A}{\MOD}{\pp}_{u'}(q)$, $q\ne\perm{\A}{\MOD}{\pp}_{v'}(q)$, $q\ne\perm{\A}{\MOD}{\pp}_{u'v'}(q)$,
and $\perm{\A}{\MOD}{\pp}_{x}(q)=\perm{\A}{\MOD}{\pp}_{x(u')^2}(q)=\perm{\A}{\MOD}{\pp}_{x(v')^k}(q)=\perm{\A}{\MOD}{\pp}_{x(u'v')^l}(q)$ for some odd prime $k$ and number $l$ that is coprime to both $2$ and $k$.
Take $S=\{s_0,s_1,\dots,s_{\pp}\}$.
 We claim that there exist $s\in S$ and $u,v\in \Aabc^\ast$ such that
\begin{align}
 \label{allnotid}
 & s\ne\perm{\A}{\MOD}{\pp}_{u}(s),\ s\ne\perm{\A}{\MOD}{\pp}_{v}(s),\ s\ne\perm{\A}{\MOD}{\pp}_{uv}(s),\\
  \label{allins}
 & \perm{\A}{\MOD}{\pp}_{x}(s)\in S,\quad\mbox{for every $x\in\{u,v\}^\ast$,}
\end{align}
\begin{align}
  \label{allorder}
 & \perm{\A}{\MOD}{\pp}_{x}(s)=\perm{\A}{\MOD}{\pp}_{xu^2}(s)=\perm{\A}{\MOD}{\pp}_{xv^k}(s)=\perm{\A}{\MOD}{\pp}_{x(uv)^l}(s),\quad\mbox{for every $x\in\{u,v\}^\ast$.}
\end{align}
Indeed, by an argument similar to the one in the proof of $(i)$ and $(ii)$ above, we must have $u'=w_u\natural^n a_1w'_u$ and $v'=w_v\natural^m a_1w'_v$ for some $w_u,w_v\in\Aiabc^\ast$, $n,m<\omega$ and $w'_u,w'_v\in\Aabc^\ast$.
 For every $x\in\{u,v\}^\ast$, as both
  $\perm{\A}{\MOD}{\pp}_{xw_u}(q)$ and $\perm{\A}{\MOD}{\pp}_{xw_v}(q)$
 are in $S$, they must be the same state. Using this
it is not hard to see that $s=\perm{\A}{\MOD}{\pp}_{u'w_u}(q)$, $u=\natural^na_1w'_uw_u$ and $v=\natural^ma_1w'_vw_v$ are as required in \eqref{allnotid}--\eqref{allorder}.

 As $M(\aut^\pp_\MOD)$ is a finite group, the set $\bigl\{\perm{\aut}{\MOD}{\pp}_{h(x)}\mid x\in\{u,v\}^\ast\bigr\}$
 forms a subgroup $\G$ in it (the subgroup generated by $\perm{\aut}{\MOD}{\pp}_{h(u)}$ and $\perm{\aut}{\MOD}{\pp}_{h(v)}$). We show that $\G$ is unsolvable by finding an unsolvable homomorphic
 image of it. To this end, we let $S'=\bigl\{ \perm{\A}{\MOD}{\pp}_{y}(s)\mid y\in\{u,v\}^\ast\bigr\}$.
 Then \eqref{allins} implies that $S'\subseteq S$ and
 \begin{equation}\label{sclosed}
 \perm{\aut}{\MOD}{\pp}_{h(x)}(s')=\perm{\A}{\MOD}{\pp}_{x}(s')\in S',\quad
 \mbox{for all $s'\in S$ and $x\in\{u,v\}^\ast$,}
 \end{equation}
 and so the restriction $\perm{\A}{\MOD}{\pp}_{x}\!\!\restriction_{S'}$ of
 $\perm{\A}{\MOD}{\pp}_{x}$ to $S'$ is an $S'\to S'$ function and
 $\perm{\A}{\MOD}{\pp}_{x}\!\!\restriction_{S'}=\perm{\aut}{\MOD}{\pp}_{h(x)}\!\!\restriction_{S'}$.
 As  $M(\aut^\pp_\MOD)$ is a group of permutations on a set containing $S'$,
$\perm{\aut}{\MOD}{\pp}_{h(x)}\!\!\restriction_{S'}$ is a permutation of $S'$, for any  $x\in\{u,v\}^\ast$.
It follows that $\{\perm{\aut}{\MOD}{\pp}_{h(x)}\!\!\restriction_{S'}\mid x\in\{u,v\}^\ast\!\}$ is a homomorphic image of $\G$, which
is unsolvable by the Kaplan--Levy criterion: by \eqref{allnotid}, \eqref{allorder}, and $2$ and $k$ being primes,
the order of the permutation $\perm{\aut}{\MOD}{\pp}_{h(u)}\!\!\restriction_{S'}$ is $2$,
the order of $\perm{\aut}{\MOD}{\pp}_{h(v)}\!\!\restriction_{S'}$ is $k$,
and the order of $\perm{\aut}{\MOD}{\pp}_{h(uv)}\!\!\restriction_{S'}$ (which is the same as the order of its inverse)
is a $>1$ divisor of $l$, and so coprime to both $2$ and $k$.

As $\G$ is an unsolvable subgroup of $M(\aut^\pp_\MOD)$, Lemma~\ref{algebralemma} implies that
$\G=M(\aut^\pp_\MOD)$, so $\{u,v\}^\ast\not\subseteq\natural^\ast$. We claim that $S'=S$.
Indeed, let $x\in\{u,v\}^\ast$ be such that $\perm{\aut}{\MOD}{\pp}_{h(x)}=\perm{\aut}{\MOD}{\pp}_{a}$.
As $|S'|\geq 2$ by \eqref{allnotid}, $s\in \{s_0,\dots,s_{\pp-1}\}$, and so
$\{s_0,\dots,s_{\pp-1}\}\subseteq S'$ follows by \eqref{sclosed}.  As there is
$y\in\{u,v\}^\ast$ with $\perm{\aut}{\MOD}{\pp}_{h(y)}=\perm{\aut}{\MOD}{\pp}_{\natural}$, $s_\pp\in S'$ also follows by \eqref{sclosed}.
Finally, as $\{u,v\}^\ast\not\subseteq\natural^\ast$, there is $x\in\{u,v\}^\ast$ of the form
$\natural^na_1wa_2w'$, for some $n<\omega$, $w\in\Aiabc$ and $w'\in\Aabc^\ast$.
As $S'=S$, $\perm{\aut}{\MOD}{\pp}_{x}(s_i)\in S$ for every $i\leq p$, and so
$w\in\bigcap_{i=0}^{\pp} \L(\A_i)$.
\end{proof}


Now Theorem \ref{DFAhard} clearly follows from Lemmas~\ref{l:DFAs} and \ref{l:empty}.
\end{proof}


\section{Deciding FO-definability of 2NFAs in \textsc{PSpace}}\label{sec:2nfa}

Using the criterion of Theorem~\ref{DFAcrit}~$(i)$, \citeA{DBLP:journals/iandc/Stern85} showed that
deciding whether the language of any given DFA is $\FO(<)$-definable can be done in \PSpace.
In this section, we also apply the criteria of Theorem~\ref{DFAcrit} to provide
\PSpace-algorithms deciding whether the language of any given 2NFA is $\lang$-definable, for $\lang \in \{ \FO(<), \FO(<,\equiv), \FO(<,\MOD)\}$.

Let $\A = (Q, \Sigma, \delta, Q_0, F)$ be a 2NFA. Similarly to~\citeA{carton_et_al:LIPIcs:2015:5413},
we first construct an exponential-size DFA $\A'$ with $\L(\A)=\L(\A')$. To this end, for any $w \in \Sigma^+$, we introduce four binary relations $\mathsf{b}_{lr}(w)$, $\mathsf{b}_{rl}(w)$, $\mathsf{b}_{rr}(w)$, and $\mathsf{b}_{ll}(w)$ on $Q$ describing the \emph{left-to-right}, \emph{right-to-left}, \emph{right-to-right}, and \emph{left-to-left behaviour of} $\A$ \emph{on} $w$. Namely,
\begin{itemize}
\item $(q,q') \in \mathsf{b}_{lr}(w)$ if there is a run of $\A$ on $w$ from $(q, 0)$ to $(q', |w|)$;

\item $(q,q') \in \mathsf{b}_{rr}(w)$ if there is a run of $\A$ on $w$ from $(q, |w|-1)$ to $(q', |w|)$;

\item $(q,q') \in \mathsf{b}_{rl}(w)$ if, for some $a \in \Sigma$, there is a run on $aw$ from $(q, |aw|-1)$ to $(q', 0)$ such that no $(q'',0)$ occurs in it before $(q', 0)$;

\item $(q,q') \in \mathsf{b}_{ll}(w)$ if, for some $a \in \Sigma$, there is a run on $aw$ from $(q, 1)$ to $(q', 0)$ such that no $(q'',0)$ occurs in it before $(q', 0)$.
\end{itemize}
For $w = \varepsilon$ (the empty word), we define the $\mathsf{b}_{ij}(w)$ as the identity relation on $Q$.

\begin{figure}[t]
\begin{center}
\begin{tikzpicture}[->,thick,node distance=2cm, transform shape]
\node[state, initial] (q0) {$q_0$};
\node[state,  right of  =q0] (r) {$r$};
\node[state, right  of =r] (s) {$s$};
\node[state, right  of =s] (t) {$t$};
\node[state, right  of =t] (v) {$v$};
\node[state, below  of =v] (w) {$w$};
\node[state, left  of =w] (u) {$u$};
\node[state, left  of =u] (y) {$y$};
\node[state, left  of =y] (z) {$z$};
\node[state, left  of =z] (p) {$p$};
\node[state, below  of =y] (g) {$g$};
\node[state, left  of =g] (h) {$h$};
\node[state, below  of =w] (x) {$x$};
\node[state, accepting,left  of =x] (q) {$q$};
\draw (q0) edge[above] node{$a,1$} (r)
(r) edge[above] node{$b,1$} (s)
(s) edge[above] node{$a,0$} (t)
(t) edge[above] node{$a,1$} (v)
(t) edge[left] node{$a,-1$} (u)
(v) edge[left] node{$b,-1$} (w)
(w) edge[above] node{$a,-1$} (u)
(u) edge[above] node{$b,1$} (y)
(y) edge[above] node{$a,1$} (z)
(z) edge[above] node{$b,1$} (p)
(u) edge[left] node{$b,-1$} (g)
(g) edge[above] node{$a,-1$} (h)
(w) edge[left] node{$a,1$} (x)
(x) edge[above] node{$b,1$} (q)
;
\end{tikzpicture}
\caption{The 2NFA $\A$ for Example~\ref{ex:2nfa}.}
\label{ex:2NFA}
\end{center}
\end{figure}

\begin{example}\label{ex:2nfa}\em
For the 2NFA $\A$ over $\Sigma = \{a,b\}$ shown in Figure~\ref{ex:2NFA}, we have:
\begin{align*}
& \mathsf{b}_{lr}(ab) = \{ (q_0, s), (s,q), (t,q), (w,q), (y,p) \}, \quad \mathsf{b}_{rl}(ab) = \{(v,u), (u,h)\},\\
& \mathsf{b}_{rr}(ab) = \{(r,s), (u,y), (v,q), (z,p)\},\quad \mathsf{b}_{ll}(ab) = \{(s,u), (t,u), (w,u)\}. \hspace*{2cm} \dashv
\end{align*}
\end{example}
Now, let $\mathsf{b} = (\mathsf{b}_{lr}, \mathsf{b}_{rl}, \mathsf{b}_{rr}, \mathsf{b}_{ll})$, where the $\mathsf{b}_{ij}$ are the behaviours of $\A$ on some $w \in \Sigma^*$, in which case we can also write $\mathsf{b}(w)$, and let $\mathsf{b}' = \mathsf{b}(w')$, for some $w' \in \Sigma^*$. We define the composition $\mathsf{b} \cdot \mathsf{b}' = \mathsf{b}''$ with components $\mathsf{b}_{ij}''$ as follows. Let $X$ and $Y$ be the reflexive and transitive closures of the relations $\mathsf{b}_{ll}' \circ \mathsf{b}_{rr}$ and $\mathsf{b}_{rr} \circ \mathsf{b}_{ll}'$ on $Q$, respectively.
Then we set:
\begin{align*}
&  \mathsf{b}_{lr}'' =  \mathsf{b}_{lr} \circ X \circ \mathsf{b}_{lr}',\qquad &&
\mathsf{b}_{rl}'' =  \mathsf{b}_{rl}' \circ Y \circ \mathsf{b}_{rl}, \\
&  \mathsf{b}_{rr}'' =  \mathsf{b}_{rr}' \cup \mathsf{b}_{rl}' \circ Y \circ \mathsf{b}_{rr} \circ \mathsf{b}_{lr}',\qquad &&
 \mathsf{b}_{ll}'' = \mathsf{b}_{ll} \cup \mathsf{b}_{lr} \circ X \circ \mathsf{b}_{ll}' \circ \mathsf{b}_{rl}.
\end{align*}
One can check that $\mathsf{b}'' = \mathsf{b}(ww')$.
\begin{example}\em\label{ex:behavours}
Consider again the 2NFA $\A$ from Example~\ref{ex:2nfa}, where we computed $\mathsf{b}(ab)$. One can readily check that $\mathsf{b}(ab) \cdot \mathsf{b}(ab) = (\mathsf{b}_{lr}, \mathsf{b}_{rl}, \mathsf{b}_{rr}, \mathsf{b}_{ll})$, where
\begin{multline*}
\mathsf{b}_{lr} = 
\mathsf{b}_{lr}(ab) \circ (\{(s,y), (t,y), (w,y) \} \cup \{(q,q) \mid q \in Q \})\circ
\mathsf{b}_{lr}(ab) ={}
\shoveright{\{(q_0, q), (q_0, p) \},}\\
\shoveleft{\mathsf{b}_{rl} = \{ (v,h)\},}\\
\shoveleft{\mathsf{b}_{rr} = \mathsf{b}_{rr}(ab) \cup \mathsf{b}_{rl}(ab) \circ (\{ (r,u) \} \cup \{(q,q) \mid q \in Q \}) \circ \mathsf{b}_{rr}(ab) \circ \mathsf{b}_{lr}(ab) =}
\shoveright{\mathsf{b}_{rr}(ab) \cup \{ (v,p) \},}\\
\shoveleft{\mathsf{b}_{ll} = \mathsf{b}_{ll}(ab) \cup \mathsf{b}_{lr}(ab) \circ (\{(s,y), (t,y), (w,y) \} \cup \{(q,q) \mid q \in Q \}) \circ \mathsf{b}_{ll}(ab) \circ{}} \\
\mathsf{b}_{rl}(ab) = \mathsf{b}_{ll}(ab) \cup \{(q_0, h)).
\end{multline*}
Clearly, $\mathsf{b}(ab) \cdot \mathsf{b}(ab)$ coincides with $\mathsf{b}(abab) = (\mathsf{b}_{lr}, \mathsf{b}_{rl}, \mathsf{b}_{rr}, \mathsf{b}_{ll})$, where $\mathsf{b}_{lr} = \{(q_0, p), (q_0, q) \}$, $\mathsf{b}_{rl} = \{ (v,h) \}$, $\mathsf{b}_{rr} = \{ (r,s), (u,y), (v,q), (v,p) \}$ and $\mathsf{b}_{ll} = \{ (q_0, h), (s,u), (t, u), (w,u) \}$; see the picture above. \hfill $\dashv$
\end{example}

Define a DFA $\A' = (Q', \Sigma, \delta', q_0', F')$ by taking
\begin{align*}
& Q' = \bigl\{ (B_{lr}, B_{rr}) \mid B_{lr} \subseteq Q_0 \times Q, \ B_{rr} \subseteq Q \times Q \bigr\},\ \
q_0' = \bigl(\bigl\{(q,q) \mid q \in Q_0\bigr\}, \emptyset\bigr),\\
& F' = \bigl\{(B_{lr}, B_{rr}) \mid (q_0, q) \in B_{lr}, \text{ for some $q_0 \in Q_0$ and $q\in F$} \bigr\},\\
& \delta'_a\bigl((B_{lr}, B_{rr})\bigr) =(B_{lr}', B_{rr}'), \text{ with}\ B_{lr}' = B_{lr} \circ X(a) \circ \mathsf{b}_{lr}(a),\\
& \hspace*{6.2cm} B_{rr}' = \mathsf{b}_{rr}(a) \cup \mathsf{b}_{rl}(a) \circ Y(a) \circ B_{rr} \circ  \mathsf{b}_{lr}(a),
\end{align*}
where $X(a)$ and $Y(a)$ are the reflexive and transitive closures of $\mathsf{b}_{ll}(a) \circ B_{rr}$ and $B_{rr} \circ \mathsf{b}_{ll}(a)$ respectively.
\begin{example}\label{aut'}\em
We illustrate the construction of the DFA $\A'$ using the 2NFA $\A$ from Example~\ref{ex:2nfa}. We have $q_0' = (\{ (q_0, q_0) \}, \emptyset)$ and
\begin{align*}
& \delta'_a(q_0') = (\{(q_0, r)\}, \{ (q_0, r), (s,v), (t,v), (w,x), (y,z) \}) = q_1',\\
& \delta'_b(q_1') = (\{(q_0, s)\}, \{(r,s), (u,y), (x,q), (z,p) \} \cup \{(v,q) \}) = q_2',\\
& \delta'_a(q_2') = (\{ (q_0, z), (q_0, v) \}, \{ (q_0, r), (s,v), (t,v), (w,x), (y,z) \} \cup{} \{ (s,z), (w,z), (t,z) \}) = q_3',\\
& \delta'_b(q_3') = (\{ (q_0, q), (q_0, p) \},\{(r,s), (u,y), (x,q), (z,p) \} \cup \{(v, q), (v, p)\}) = q_4'.
\end{align*}
Note that $q_4' \in F'$. \hfill $\dashv$
\end{example}

Returning to our general construction, we observe that, for any $w\in\Sigma^\ast$,
\begin{align}
\nonumber
& \delta'_w\bigl((B_{lr}, B_{rr})\bigr) =(B_{lr}', B_{rr}')\ \text{iff}\ B_{lr}' = B_{lr} \circ X(w) \circ \mathsf{b}_{lr}(w) \text{ and}\\
\label{aprime-reach}
& \hspace*{5.3cm} B_{rr}' = \mathsf{b}_{rr}(w) \cup \mathsf{b}_{rl}(w) \circ Y(w) \circ B_{rr} \circ  \mathsf{b}_{lr}(w),
\end{align}
where $X(w)$ and $Y(w)$ are the reflexive and transitive closures of $\mathsf{b}_{ll}(w) \circ B_{rr}$ and $B_{rr} \circ \mathsf{b}_{ll}(w)$. (To illustrate, for $\A'$ in Example~\ref{aut'}, we have just shown that $\delta'_{abab}(q_0') = (B_{lr}', B_{rr}') = q_4'$, and $(B_{lr}', B_{rr}')$ can be computed by applying~\eqref{aprime-reach} to $q_0'$ and $\mathsf{b}(abab)$ defined in Example~\ref{ex:behavours}.)
Similarly to~\citeA{5392614,10.1016/0020-0190(89)90205-6} one can show that 
\begin{equation}\label{eq:2nfatodfa}
\L(\A) = \L(\A').
\end{equation}
Intuitively, $q_0' \to_w (B_{lr}, B_{rr})$ in $\A'$ and $q \in B_{lr}$ iff there exists a (two-way) run of $\A$ on $w$ from $(q_0, 0)$ to $(q, |w|)$ on $w$.

Next, we prove that, even though the size of $\A'$ is exponential in $\A$, we can still use Theorem~\ref{DFAcrit} to decide $\lang$-definability of $\L(\A)$ in \PSpace:

\begin{theorem}\label{thm:2NFA}
For $\lang \in \{ \FO(<), \FO(<,\equiv), \FO(<,\MOD)\}$,
deciding $\lang$-definability of $\L(\A)$, for any 2NFA $\A$, can be done in \PSpace.
\end{theorem}
\begin{proof}
Let $\A'$ be the DFA defined above for the given 2NFA $\A$.
By Theorem~\ref{DFAcrit} $(i)$ and~\eqref{eq:2nfatodfa}, $\L(\A)$ is not $\FO(<)$-definable iff there exist a word $u\in\Sigma^\ast$, a reachable state $q \in Q'$, and a number $k\leq|Q'|$ such that $q\not\simm \delta'_u(q)$ and $q= \delta'_{u^k}(q)$. We guess the required $k$ in binary, $q$ and a quadruple $\mathsf{b}(u)$ of binary relations  on $Q$. Clearly, they all can be stored in polynomial space in $|\A|$. To check that our guesses are correct, we  first check that $\mathsf{b}(u)$ indeed corresponds to some $u \in \Sigma^\ast$. This is done by guessing a sequence $\mathsf{b}_0, \dots, \mathsf{b}_n$ of distinct quadruples of binary relations on $Q$ such that $\mathsf{b}_0 = \mathsf{b}(u_0)$ and $\mathsf{b}_{i+1} = \mathsf{b}_{i} \cdot \mathsf{b}(u_{i+1})$, for some $u_0, \dots, u_n \in \Sigma$. (Any sequence with a subsequence starting after $\mathsf{b}_i$ and ending with $\mathsf{b}_{i+m}$, for some $i$ and $m$ such that $\mathsf{b}_i = \mathsf{b}_{i+m}$, is equivalent, in the context of this proof, to the sequence with such a subsequence removed.) Thus, we can assume that $n \leq 2^{O(|Q|)}$, and so $n$ can be guessed in binary and stored in \PSpace{}. So the stage of our algorithm checking that $\mathsf{b}(u)$ corresponds to some $u \in \Sigma^*$ makes $n$ iterations and continues to the next stage if $\mathsf{b}_n = \mathsf{b}(u)$ or terminates with an answer \no{} otherwise. Now, using $\mathsf{b}(u)$, we compute $\mathsf{b}(u^k)$ by means of a sequence $\mathsf{b}_0, \dots, \mathsf{b}_k$, where $\mathsf{b}_0 = \mathsf{b}(u)$ and $\mathsf{b}_{i+1} = \mathsf{b}_{i} \cdot \mathsf{b}(u)$. With $\mathsf{b}(u)$ ($\mathsf{b}(u^k)$), we compute $\delta'_{u}(q)$ (respectively, $\delta'_{u^k}(q)$) in \PSpace{} using \eqref{aprime-reach}. If $\delta'_{u^k}(q) \neq q$, the algorithm terminates with an answer \no{}. Otherwise, in the final stage of the algorithm, we check that $\delta'_{u}(q) \not \sim q$. This is done by guessing $v \in \Sigma^*$ such that $\delta'_v(q) = q_1$, $\delta'_v\bigl(\delta'_{u}(q)\bigr) = q_2$, and $q_1 \in F'$ iff $q_1 \not \in F'$. We guess such $v$ (if any) in the form of $\mathsf{b}(v)$ using an algorithm analogous to that for guessing $u$.

By Theorem~\ref{DFAcrit}~$(ii)$ and~\eqref{eq:2nfatodfa}, $\L(\A)$ is not $\FO(<,\equiv)$-definable iff there
there exist words $u,v\in\Sigma^\ast$, a reachable state $q \in Q'$, and a number $k\leq|Q'|$ such that
$q\not\simm\delta'_u(q)$, $q=\delta'_{u^k}(q)$, $|v|=|u|$, and $\delta'_{u^i}(q)=\delta'_{u^iv}(q)$, for all $i<k$.
We outline how to modify the algorithm for $\FO(<)$ above to check $\FO(<,\equiv)$-definability. First, we need to guess and check $v$ in the form of $\mathsf{b}(v)$ in parallel with guessing and checking $u$ in the form of $\mathsf{b}(u)$, making sure that $|v| = |u|$. For that, we guess a sequence of distinct pairs $(\mathsf{b}_0, \mathsf{b}_0'), \dots, (\mathsf{b}_n, \mathsf{b}_n')$ such that the $\mathsf{b}_i$ are as above, $\mathsf{b}_0' = \mathsf{b}(v_0)$ and $\mathsf{b}_{i+1}' = \mathsf{b}_{i}' \cdot \mathsf{b}(v_{i+1})$, for some $v_0, \dots, v_n \in \Sigma$. (Any such sequence with a subsequence starting after $(\mathsf{b}_i, \mathsf{b}_i')$ and ending with $(\mathsf{b}_{i+m}, \mathsf{b}_{i+m}')$, for some $i$ and $m$ such that $(\mathsf{b}_i, \mathsf{b}_i') = (\mathsf{b}_{i+m}, \mathsf{b}_{i+m}')$, is equivalent to the sequence with that  subsequence removed.) So $n \leq 2^{O(|Q|)}$.
For each $i<k$,
we can then compute $\delta'_{u^i}(q)$ and $\delta'_{u^iv}(q)$, using \eqref{aprime-reach}, and check whether
whether they are equal.

Finally,
by Theorem~\ref{DFAcrit}~$(iii)$ and~\eqref{eq:2nfatodfa}, $\L(\A)$ is not $\FO(<,\MOD)$-definable iff
 there exist $u,v\in\Sigma^\ast$, a reachable
state $q\in Q'$ and $k,l\leq |Q'|$ such that $k$ is an odd prime, $l>1$ and coprime to both $2$ and $k$,
$q\not\simm\delta'_u(q)$, $q\not\simm\delta'_v(q)$, $q\not\simm\delta'_{uv}(q)$, and
$\delta'_{x}(q)\simm\delta'_{xu^{2}}(q)\simm\delta'_{xv^{k}}(q)\simm\delta'_{x(uv)^{l}}(q)$, for all $x\in\{u,v\}^\ast$.
We start by guessing $u,v \in \Sigma^*$ in the form of $\mathsf{b}(u)$ and $\mathsf{b}(v)$, respectively.
Also, we guess $k$ and $l$ in binary and check that $k$ is an odd prime and $l$ is coprime to both $2$ and $k$.
By \eqref{aprime-reach}, $\delta'_x$ is determined by $\mathsf{b}(x)$, for any $x\in\{u,v\}^\ast$.
Thus, to check that $u$, $v$, $k$, $l$ are as required, we perform the following steps, for \emph{each} quadruple $\mathsf{b}$ of binary relations on $Q$. First, we check whether $\mathsf{b} = \mathsf{b}(x)$, for some $x \in \{u,v\}^\ast$ (we discuss the algorithm for this below). If this is not the case, we construct the \emph{next} quadruple $\mathsf{b}'$ and process it as $\mathsf{b}$ above. If it is the case,
we compute all the states $\delta'_{x}(q)$, $\delta'_{xu^{2}}(q)$, $\delta'_{xv^{k}}(q)$, $\delta'_{x(uv)^{l}}(q)$, $\delta'_{u}(q)$, $\delta'_{v}(q)$, $\delta'_{uv}(q)$,
and check their required (non)equivalences with respect to $\sim$,
using the same method as for checking $\delta'_{u}(q) \not\sim q$ above.
 If they do not hold, our algorithm terminates with an answer \no{}. Otherwise, we construct the \emph{next} quadruple $\mathsf{b}'$ and process it as this $\mathsf{b}$. When all possible quadruples $\mathsf{b}$ of binary relations of $Q$ have been processed, the algorithm terminates with an answer \yes{}.

Now, to check that a given quadruple $\mathsf{b}$ is equal to $\mathsf{b}(x)$, for some $x \in  \{u,v\}^\ast$, we simply guess a sequence $\mathsf{b}_0, \dots, \mathsf{b}_n$ of quadruples of binary relations on $Q$ such that $\mathsf{b}_0 = \mathsf{b}(w_0)$, $\mathsf{b}_n = \mathsf{b}$ and $\mathsf{b}_{i+1} = \mathsf{b}_{i} \cdot \mathsf{b}(w_{i+1})$, where $w_i \in \{u, v\}$. It follows from the argument above that it is enough to take $n \leq 2^{O(|Q|)}$.
\end{proof}


\section{Deciding FO-rewritability of $\LTL$ OMQs}\label{sec:LTL-genearal}

In this section, we use the results obtained above to establish the complexity of recognising the rewritability type of an arbitrary $\LTL_{\bool}\Xallop$ OMQ.


%
%


\begin{theorem}\label{hornExpSpacehard}
For any $\lang \in \{ \FO(<), \FO(<,\equiv), \FO(<,\MOD)\}$, deciding $\lang$-rewri\-tability of \textup{(}Boolean and specific\textup{)} $\LTL_{\bool}\Xallop$ OMQs over $\Xi$-ABoxes is \ExpSpace-complete. The lower bound holds already for $\LTL_\horn\Xnext$ OMAQs.
\end{theorem}
\begin{proof}
The upper bound follows from Theorem~\ref{thm:2NFA} and the proof of Theorem~\ref{Prop:rewr-def}. We now establish the matching lower bound for $\LTL_\horn\Xnext$ OMAQs. We only consider specific OMAQs, leaving the easier case of Boolean OMAQs to the reader. (In fact, $\ExpSpace$-hardness for Boolean OMAQs follows from $\ExpSpace$-hardness for specific OMAQs by Lemma~\ref{le:remove-bot} and Proposition~\ref{prop:specific-to-boolean++} $(i)$ to be proved in the next section). With this in mind, we first show how one can store and compute numerical values of polynomial length using $\LTL_\horn\Xnext$-ontologies.

A \emph{counter} is a set $\mathbb A=\{A^i_j \mid i=0,1, \ j=1,\dots,k\}$ of atomic concepts that will be used to store values between $0$ and $2^k-1$, which can be different at different time points. The counter $\mathbb A$ is \emph{well-defined} at a time point $n \in \Z$ in an interpretation $\I$ if $\I,n \models A^0_j \land A^1_j \to \bot$ and $\I,n \models A^0_j \lor A^1_j$, for any $j=1,\dots,k$. In this case, the \emph{value of} $\mathbb A$ at $n$ in $\I$ is given by the unique binary number $b_{k} \dots b_1$ for which $\I,n \models A^{b_1}_1\wedge\dots\wedge A^{b_k}_k$.
We require the following formulas, for $c = b_{k} \dots b_1$ and a well-defined counter $\mathbb A$:
\begin{itemize}
\item $[\mathbb A={c}] = A^{b_1}_1\wedge\dots\wedge A^{b_k}_k$ with $\I,n \models [\mathbb A={c}]$ iff the value of $\mathbb A$ is $c$;

\item $\displaystyle[\mathbb A{<c}] = \bigvee_{\substack{k\geq i\geq 1\\ b_i=1}}\big(A_i^0\wedge\bigwedge_{j=i+1}^kA_j^{b_j}\big)$ with $\I,n \models [\mathbb A < {c}]$ iff the value of $\mathbb A$ is $< c$;

\item $\displaystyle[\mathbb A{>c}] = \bigvee_{\substack{k\geq i\geq 1\\ b_i=0}}\big(A_i^1\wedge\bigwedge_{j=i+1}^kA_j^{b_j}\big)$ with $\I,n \models [\mathbb A > {c}]$ iff the value of $\mathbb A$ is $> c$.
\end{itemize}
%
%
%
We regard the set $(\Rnext \mathbb A)=\{\Rnext A^i_j \mid i=0,1, \ j=1,\dots,k\}$ as another counter that stores at $n$ in $\I$ the value stored by $\mathbb A$ at $n+1$ in $\I$.
This allows us to use formulas such as $[\mathbb A >c_1]\to[(\Rnext \mathbb A)={c_2}]$, which says that if the value of $\mathbb A$ at $n$ in $\I$ is greater than $c_1$, then the  value of $\mathbb A$ at $n+1$ in $\I$ is $c_2$.


Given two counters $\mathbb A$ and $\mathbb B$, we set
\begin{align*}
& [\mathbb A=\mathbb B] = \bigwedge_{j=1}^k\left((B_j^0\to A_j^0)\wedge(B_j^1\to A_j^1)\right),\\
& [\mathbb A={\mathbb B+1}] = \bigwedge_{i=1}^k\big(( B_i^0\wedge B_{i-1}^1\wedge\dots\wedge B_{1}^1\to A_i^1\wedge A_{i-1}^0\wedge\dots\wedge A_{1}^0)\wedge{}\\
& \hspace{5.5cm} \bigwedge_{j<i}((B_i^0\wedge B_{j}^0\to A_i^0)\wedge(B_i^1\wedge B_{j}^0\to A_i^1))\big).
\end{align*}
Then $\I,n \models [\mathbb A=\mathbb B]$ iff the values of $\mathbb A$ and $\mathbb B$ at $n$ in $\I$ coincide, and $\I,n \models [\mathbb A={\mathbb B+1}]$ iff the value of $\mathbb A$ at $n$ is equal to the value of $\mathbb B$ at $n$ plus one. In a similar way, we define the formula $[\mathbb A=\mathbb B-1]$.

Consider a deterministic Turing machine $\M$ with exponential space bound, which behaves as described in the proof of Theorem~\ref{DFAhard}. Given an input word $\boldsymbol{x} = x_1\dots x_n$, let $N$ be the number of tape cells needed for the computation of $\M$ on $\avec{x}$, and let $\pp$ be the first prime such that $\pp\geq\ppn+2$ and $\pp\not\equiv\pm 1\ (\text{mod}\ 10)$. 
Our aim is to  construct $\LTL_\horn\Xnext$ ontologies $\TO_{<}$, $\TO_{\equiv}$ and $\TO_\MOD$ of polynomial size that simulate the exponential-size, $O(\pp)$, DFAs $\A_<$, $\A_\equiv$ and  $\A_{\MOD}$ from the proof of Theorem~\ref{DFAhard}, whose languages are $\lang$-definable (for the corresponding $\lang$) iff $\M$ rejects $\boldsymbol{x}$. The polynomial size of the ontologies can be achieved due to the repetitive structure of the automata $\A_<$, $\A_\equiv$ and  $\A_{\MOD}$ as  we can capture an exponential number of transitions by using only polynomially-many axioms.

First we define $\TO_<$. Let $k=\lceil\log_2\pp\rceil+1$.
The ontology $\TO_<$ uses the following atomic concepts: the symbols in $\Sigma'=\Gamma\cup(Q\times\Gamma)\cup\{\sharp,\flat,a_1,a_2\}$ (see the proof of Theorem~\ref{DFAhard}) and additional symbols $S$, $T$, $Q$, $P$, $Q_a$, $R_a$, $R_{ab}$, $P_a$, $P_{ab}$, for $a,b\in\Sigma'$, $F$, $X$, $Y$, and $F_\textit{end}$. We also use counters $\mathbb A$ and $\mathbb L$ with atomic concepts $A^i_j$ and $L^i_j$, for $i=0,1$, $j=1,\dots,k$. Set $\Xi=\Sigma'\cup\{X,Y\}$.

In the DFA $\A_i$ from the proof Theorem~\ref{DFAhard}, we represent
%
%
\begin{itemize}\itemsep=0pt
\item the state $t_i$ as $[\mathbb A={i}]\land T$;

\item each state $q^j$ of $\A_i$ as $[\mathbb A={i}]\land Q\land[\mathbb L={j}]$;

\item each state $q^j_a$ of $\A_i$ as $[\mathbb A={i}]\land Q_a\land[\mathbb L={j}]$;

\item each state $p^j$ of $\A_0$ as $[\mathbb A={0}]\land P\land[\mathbb L={j}]$;

\item each state $p^j_{a}$ of $\A_i$ as $[\mathbb A={i}]\land P_{a}\land[\mathbb L={j}]$;

\item each state $p_{ab}$ of $\A_i$ as $[\mathbb A={i}]\land P_{ab}$;

\item each state $r_a$ of $\A_i$ as $[\mathbb A={i}]\land R_a$;

\item each state $r_{ab}$ of $\A_i$ as $[\mathbb A={i}]\land R_{ab}$;

\item $f_i$ as $[\mathbb A={i}]\land F$.
\end{itemize}
We refer to these formulas and also $[\mathbb A={i}]\land S$ representing $s_i$ in $\A_<$ as \emph{state formulas}.

The ontology $\TO_<$ simulating $\A_<$ consists of  the following axioms, which are equivalent to polynomially-many $\LTL_\horn\Xnext$ axioms (see Lemma \ref{simulationlemma}):
\begin{itemize}
\item $a\wedge b\to\bot$, for distinct $a,b\in\Xi$; \hfill $(\star_1)$

\item $X\to[(\Rnext\mathbb A)={0}]\land\Rnext{ S}$ to simulate the initial state of $\A_<$; \hfill $(\star_2)$

\item $[\mathbb A={0}]\land S\land Y\to F_\textit{end}$ to simulate the accepting state of $\A_<$; \hfill $(\star_3)$

\item the axioms
\begin{align*}
&[\mathbb A<\pp]\land S\land a_1\to[(\Rnext\mathbb A)=\mathbb A]\land\Rnext T\land[(\Rnext\mathbb L)={\mathbb A}],\\
&[\mathbb A <\pp-1]\land F\land a_2\to[(\Rnext\mathbb A)={\mathbb A+1}]\land\Rnext S,\\
&[\mathbb A=\pp-1]\land F\wedge a_2\to[(\Rnext\mathbb A)={ 0}]\land\Rnext S;
\end{align*}
describing the behaviour of $\A_<$ in states $s_i$ and $f_i$;

\item the axioms
describing the transitions of $\A_i$, $0\le i\leq N$, that are given in Appendix~\ref{sec:LTL-genearal-app};

\item[--] and the following axioms for $a\ne\flat$:
\begin{align*}
&[\mathbb A>N]\land[\mathbb A<\pp]\land T\land a\to{} [(\Rnext\mathbb A)=\mathbb A]\land \Rnext T,\\
&[\mathbb A>N]\land[\mathbb A<\pp]\land T\land \flat\to[(\Rnext\mathbb A)=\mathbb A]\land \Rnext F 
\end{align*}
simulating the transitions of $\A_i$, for $N< i< \pp$.
\end{itemize}
Next, we define the ontology $\TO_{\equiv}$ by adding to $\TO_<$ the axiom
$$
[\mathbb A <p]\land S\land \natural\to[(\Rnext \mathbb A)=\mathbb A]\land\Rnext S
$$
simulating the $\natural$-transitions in $\A_\equiv$. We also extend $\Xi$ with $\natural$.

To define $\TO_{\MOD}$, more work is needed. First, we extend $\TO_<$ with
\begin{itemize}
\item the following axioms regarding $\A_{\pp}$:
\begin{align*}
&[\mathbb A={\pp}]\land S\land a_1\to[(\Rnext\mathbb A)={\pp}]\land\Rnext T,\\
&[\mathbb A={\pp}]\land F\land a_2\to[(\Rnext\mathbb A)={\pp}]
\land \Rnext S,
\end{align*}
\item and the following axioms handling $\natural$:
\begin{align*}
&[\mathbb A=0]\land S\land \natural\to[(\Rnext\mathbb A)={\pp}]\land \Rnext S,\\
&[\mathbb A={\pp}]\land S\land \natural \to[(\Rnext\mathbb A)=0]\land S,\\
&[\mathbb A >0]\land[\mathbb A <\pp]\land S\land \natural \to[(\Rnext\mathbb A)=\mathbb J]\land\Rnext S.
\end{align*}
\end{itemize}
Here, $\mathbb J$ is a new counter that stores the value $j=-1/i$ in the field $\mathbb F_{\pp}$, which is required to make sure that, for $i\neq 0,\pp$, we have
$$
\TO_\MOD\models [\mathbb A = i]\land S\land \natural\to [(\Rnext\mathbb A) = j]\land\Rnext S.
$$
We achieve this as follows.
We compute the number $r$ such that $ir=1 (\text{mod}\, N')$ using the following modified version of Penk's algorithm~\cite<e.g.,>[Exercise~4.5.2.39]{DBLP:books/lib/Knuth98}. The algorithm starts with $u = \pp$, $v = i$, $r = 0$, $s = 1$. In the course of the algorithm, $u$ and $v$ decrease, with  the following conditions being met: $\text{GCD}(u, v) = 1$, $u = ri\, (\text{mod}\, \pp)$, and $v = si\, (\text{mod}\, \pp)$.
The algorithm repeats the following steps until $v=0$:
\begin{itemize}
\item if $v$ is even, replace it with $v/2$, and replace $s$ with either $s/2$ or $(s+\pp)/2$, whichever is a whole number;

\item if $u$ is even, replace it with $u/2$, and replace $r$ with either $r/2$ or $(r+\pp)/2$, whichever is a whole number;

\item if $u,v$ are odd and $u>v$, replace $u$ with $(u-v)/2$ and $r$ with either $(r-s)/2$ or $(r-s+\pp)/2$, whichever is a whole number;

\item if $u,v$ are odd and $v\ge u$, replace $v$ with $(v-u)/2$ and $s$ with either $(s-r)/2$ or $(s-r+\pp)/2$, whichever is a whole number.
\end{itemize}
The binary length of the larger of $u$ and $v$ is reduced by
at least one bit, guaranteeing that the procedure terminates
in at most $2k$ iterations while maintaining the conditions.
At termination, $v = 0$ as otherwise a reduction is still possible.
If $u = 1$,
we get $1 = ri\, (\text{mod}\, \pp)$ and $r= 1/i$ in the field $\mathbb F_{\pp}$, so we can set $j=\pp-r$.

To compute the value of $j$, we need to halve the number in a counter, compare two counters (using an additional counter), add and subtract (using extra counters for carries). This can be done by means of $O(k)$ counters (a fixed number of counters per $O(k)$ steps of the algorithm) with polynomially-many additional axioms. So we compute $j$ when required and store it in counter $\mathbb J$. Appendix~\ref{sec:LTL-genearal-app} provides a full list of counters and axioms we need.



For $\lang \in \{ \FO(<),\FO(<,\equiv), \FO(<,\MOD)\}$, we use $\A_\lang$ and $\TO_\lang$ to denote the corresponding automaton and ontology defined above. Observe that, by the proof of Theorem~\ref{DFAhard}, 
\begin{equation}\label{eq:alang-rejects}
 \L(\A_\lang) \text{ is } \lang\text{-definable} \quad \text{ iff } \quad \M \text{ rejects } \avec{x}.
\end{equation}
The connection between $\A_\lang$ and $\TO_\lang$ is explained by the following lemma.
\begin{lemma}\label{simulationlemma}
 Let $\Abox$ be a $\Xi$-ABox and let $\Psi$ be a state formula. Then

$(i)$ $\Abox$ is inconsistent with $\TO_\lang$ iff there is $i$ such that $a(i),b(i)\in\Abox$ for different $a,b\in \Xi$\textup{;}

$(ii)$ if $\Abox$ is consistent with $\TO_\lang$, then $\TO_\lang,\Abox\models\Psi(l)$
%
%
iff $\Abox$ contains a subset
\begin{equation}\label{A:subset}
\{X(l-m-1), b_1(l-m), b_2(l-m+1), b_3(l-m+2),\dots, b_{m}(l-1)\},
\end{equation}
where $m\ge 0$, $b_k\in\Sigma'$ for all $k\in[1,m]$, and $\A_\lang$, having read the word $b_1\dots b_{m}$, is in the state represented by $\Psi$.
\end{lemma}
\begin{proof}
We obtain $(i)$ because the only axiom with $\bot$ is $(\star_1)$ and, for consistent $\Abox$ and $\TO_\lang$, $b\in\Xi$ and $n\in\mathbb Z$, we have $(\TO,\Abox) \models b(n)$ iff $b(n)\in\Abox$.

$(ii)$ $(\Leftarrow)$ If there is such a subset of $\Abox$, then $(\TO_\lang,\Abox) \models\big([\mathbb A=0]\land S\big)(l-m)$. One can check by induction on $j$ that if the automaton is in a state $q$ after reading $b_1\ldots b_{j-1}$ and $q$ is represented by a state formula $\Psi'$, then $(\TO,\Abox) \models \Psi'(l-m+j)$.

$(\Rightarrow)$ If $(\TO_\lang,\Abox) \models A_{j_1}^{\iota_1}(l)$, for some $A_{j_1}^{\iota_1}\in\mathbb A$, then $(\TO_\lang,\Abox) \models b(l-1)$, for some $b\in\Xi$. There are two possibilities: either $b=X$ or $b\in\Sigma'$ and there exists $A_{j_2}^{\iota_2}\in\mathbb A$ such that  $(\TO_\lang,\Abox) \models A_{j_2}^{\iota_2}(l-1)$. So there is a unique subset of $\Abox$ of the form~\eqref{A:subset}.
By induction on $j\in[1,m+1]$, we can prove that there exists a unique state formula $\Psi_j$ such that $(\TO_\lang,\Abox) \models\Psi_j(l-m+j)$ and $\Psi_j$ represents the state $\A_\lang$ is in after reading $b_1\dots b_{j-1}$.
\end{proof}
%
%
To complete the proof of Theorem~\ref{hornExpSpacehard}, we need one more lemma.
\begin{lemma}\label{rewrtodefine}
Let $\q_\lang(x) = (\TO_\lang,F_\textit{end}(x))$. For the signature $\Xi$ above, $\L(\A_\lang)$ is $\lang$-definable iff $\L_\Xi(\q_\lang(x))$ is $\lang$-definable.
\end{lemma}
\begin{proof}
Recall that the alphabet of $\L_\Xi(\q_\lang(x))$ is $\Gamma_\Xi=\Sigma_\Xi\cup\Sigma_\Xi'$.
As $(\star_3)$ is the only rule that produces the target concept $F_{end}$ and $F_{end}\notin\Xi$,  $k$ is a certain answer to $\q_\lang(x)$ over a $\Xi$-ABox $\Abox$ iff either $\Abox$ is inconsistent with $\TO_\lang$ or $(\TO_\lang,\Abox) \models \big([\mathbb A=0]\land S\land Y\big)(k)$  iff,
by Lemma~\ref{simulationlemma}, there are $a(i), b(i)\in\Abox$, for $a,b \in \Xi$ with $a\ne b$, or $\Abox$ contains a subset 
$$
\{X(k-m-1), b_1(k-m),\dots,b_m(k-1),Y(k)\},
$$
where $b_1\dots b_m\in \L(\A_\lang)$.

Let $\Xi^{\{\}}$ and  $\L^{\{\}}(\A_\lang)$ stand for
$\Xi$ and, respectively, $\L(\A_\lang)$, in which every $a\in\Xi$ is replaced by the set $\{a\}$. It follows that $\L_\Xi(\q_\lang(x)) = \L_0 \cup \L_1$, where 
\begin{align*}
\L_0 = & \bigl\{\Abox a' \mathcal B \mid \Abox \mathcal B \in \Sigma_\Xi^*, a' \in \Sigma_\Xi'\bigr\} \cap \bigl\{\Abox a \mathcal B \mid \Abox a \mathcal B \in \Gamma_\Xi^*, |a|>1\bigr\}\,\\
\L_1= &\bigl\{w_1\{X\}w\{Y'\}w_2\mid w\in \L^{\{\}}(\A_\lang),\ w_1,w_2\in \bigl(\Xi^{\{\}}\cup\{\emptyset\}\bigr)^*\bigr\}.
\end{align*}
%
%
%
(Indeed, $\L_0$ describes the inconsistent ABoxes and $\L_1$ the consistent ones.) Clearly, the language $\L_0$ is $\lang$-definable. Let $\varphi$ be an $\lang$-formula defining it.
If $\L(\A_\lang)$ is definable by an $\lang$-formula, then so are $\L^{\{\}}(\A_\lang)$ and, by Lemma~\ref{expandL}, $\L_1$. Let $\psi$ be the $\lang$-formula defining $\L_1$. Then $\varphi\lor\psi$ defines $\L_\Xi(\q_\lang(x))$.
If $\L_\Xi(\q_\lang(x))$ is definable by an $\lang$-formula $\chi$, then $\chi\land \neg \varphi$ defines $\L_1$. Thus, by Lemma~\ref{expandL}, the language $\L^{\{\}}(\A_\lang)$ is $\lang$-definable, and so is $\L(\A_\lang)$.
\end{proof}

By Theorem~\ref{Prop:rewr-def} $(ii)$, $\q_\lang(x)$ is $\lang$-rewritable over $\Xi$-ABoxes iff $\L(\A_\lang)$ is $\lang$-definable. By \eqref{eq:alang-rejects}, $\L(\A_\lang)$ is $\lang$-definable iff $\M$ rejects $\avec{x}$, which completes the proof of Theorem~\ref{hornExpSpacehard}.
\end{proof}

We also observe that $\LTL_\horn\Xnext$ ontologies can be encoded by positive existential queries  mediated by covering axioms that are available in $\LTL_\krom$:

\begin{theorem}\label{thm:kromlowerexp}
Deciding $\lang$-rewritability of \textup{(}Boolean and specific\textup{)} $\LTL_\krom$ OMPEQs over $\Xi$-ABoxes is \ExpSpace-complete.
\end{theorem}
\begin{proof}
By Theorem~\ref{hornExpSpacehard}, we only need to show the lower bound, which can be done by reduction of $\LTL_\horn\Xnext$ OMAQs $\q=(\TO,A)$ to $\LTL_\krom\Xnext$ OMPEQs.
By Remark~\ref{rem1}, we can assume that the axioms of $\TO$ take the form $\avec{C}\to \bot$ or $\avec{C}\to B$,
for some $\avec{C}=C_1\land\dots\land C_n$ and atomic $B$.
We construct an $\LTL_\krom\Xnext$ OMPQ $\q' = (\TO', \varkappa)$ that is $\Xi$-equivalent to $\q$ by taking $\TO'$ with the axioms $B\land\bar B\to\bot$ and $\top\to B\vee\bar B$, for all $B\in\sig(\q)$, where $\bar B$ is a fresh atom, and
$$
\varkappa ~=~ A \ \ \vee
\bigvee_{\avec{C}\to \bot \text{ in }\TO} \hspace*{-3mm} \Rdiamond\Ldiamond \avec{C} \ \ \vee
\bigvee_{\avec{C}\to B \text{ in } \TO} \hspace*{-3mm} \Rdiamond\Ldiamond(\avec{C}\land \bar B).
$$
Intuitively, $\bar B$ represents the negation of $B$ and $\varkappa$ is equivalent to the formula 
$$
\bigl[\,\bigwedge_{\avec{C}\to \bot \text{ in } \TO} \Rbox\Lbox (\avec{C} \to \bot) \land \bigwedge_{\avec{C}\to B \text{ in } \TO} \Rbox\Lbox (\avec{C} \to B)\,\bigr] \to A.
$$
It is readily seen that, for any $\Xi$-ABox $\Abox$, the certain answer to $\q$ over $\Abox$ is \yes\ iff the answer to $\q'$ over $\Abox$ is \yes, and $k$ is a certain answer to $\q(x)$ over $\Abox$ iff it is also a certain answer to $\q'(x)$.
It follows that $\q'$ is $\lang$-rewritable over $\Xi$-ABoxes iff $\q$ is $\lang$-rewritable.
\end{proof}


\section{Deciding $\lang$-rewritability of Linear $\LTL_\horn\Xnext$ OMPQs}\label{sec:linear}

As well known, deciding FO-rewritability of monadic datalog queries is 2\textsc{ExpTime}-complete \cite{DBLP:conf/stoc/CosmadakisGKV88,DBLP:conf/lics/BenediktCCB15,DBLP:conf/pods/KikotKPZ21}, which becomes \PSpace{} for the important class of linear monadic queries~\cite{DBLP:conf/stoc/CosmadakisGKV88,DBLP:journals/ijfcs/Meyden00}.
In this section, we focus on linear $\LTL_\horn\Xnext$ OMPQs. First, in Section~\ref{sec:incons-reductions}, we show that it suffices to consider $\bot$-free OMQs only and that deciding $\lang$-rewritability of specific $\LTL_\horn\Xallop$ OMPQs is polynomially reducible to the same problem for Boolean $\LTL_\horn\Xallop$ OMPQs and the other way round.
Then, in Section~\ref{sec:linear-2}, for any linear $\LTL_\horn\Xnext$ OMAQ $\q$, we construct in polynomial space a DFA $\A'$ such that $\q$ is $\lang$-rewritable iff $\L(\A')$ is $\lang$-definable. So, by Theorem~\ref{thm:2NFA}, deciding $\lang$-rewritability of linear $\LTL_\horn\Xnext$ OMAQs can be done in \PSpace{}. An essential part of this proof is the construction of a (polynomial-size) 2NFA $\A_\q^\Xi$ that recognises a certain encoding of the language of $\q$. We also show that any DFA can be simulated by a linear $\LTL_\horn\Xnext$ OMAQ, which yields a \PSpace{} lower bound for deciding $\lang$-rewritability.
Section~\ref{sec:lin-ompq} gives semantic criteria of $\FO(<)$- and $\FO(<,\equiv)$-rewritiability of $\LTL_\horn\Xnext$ OMPQs and a \PSpace{} algorithm for checking these criteria based on $\A_\q^\Xi$.

\subsection{Two useful reductions}\label{sec:incons-reductions}

We start with two technical observations. The first one rids ontologies of $\bot$.

\begin{lemma}\label{le:remove-bot}
Let $\TO$ be an $\LTL\Xallop_\bool$ ontology, let $\TO'$ result from $\TO$ by removing every axiom of the form $C_1\land\dots\land C_k\to \bot$, and let $\TO''$ result from $\TO$ by replacing every axiom of the form $C_1\land\dots\land C_k\to \bot$ with $C_1\land\dots\land C_k\to A'$, $A' \to \Rnext A'$, $A' \to \Lnext A'$, $A' \to A$, for a fresh atom $A'$. Let $\Xi$ be a signature that does not contain the newly introduced atoms $A'$.

$(i)$ Every Boolean OMAQ $\q=(\TO, A)$ is $\Xi$-equivalent to $\q' = (\TO'', A)$. Every specific OMAQ $\q(x) = (\TO, A(x))$ is $\Xi$-equivalent to $\q'(x) = (\TO'', A(x))$.

$(ii)$ Every Boolean OMQ $\q=(\TO, \varkappa)$ is $\Xi$-equivalent to $\q'' = (\TO',\varkappa')$, where
\begin{equation*}
\varkappa'=\varkappa\vee\bigvee_{C_1\land\dots\land C_k\to \bot\in\TO}\Rdiamond\Ldiamond(C_1\land\dots\land C_k)
\end{equation*}
Every specific OMQ $\q(x) = (\TO, \varkappa(x))$ is $\Xi$-equivalent to $\q''(x) = (\TO',\varkappa'(x))$.
\end{lemma}
\begin{proof}
We only show the first claim in $(i)$; the other claims are similar and left to the reader. Let $\Abox$ be any $\Xi$-ABox. Suppose the certain answer to $\q'$ over $\Abox$ is \no. This means that there is a model $\I$ of $\TO''$ and $\Abox$ such that $\I,n\not\models A$ for all $n \in\Z$. Then $\I$ is also a model of $\TO$ and $\Abox$. Indeed, if $\I,n \models C_1\land\dots\land C_k$, for some axiom $C_1\land\dots\land C_k\to \bot$ in $\TO$ and $n \in \Z$, then $\I,n \models A'$, and so $\I,n \models A$, which is a contradiction. It follows that the answer to $\q$ over $\Abox$ is \no. Conversely, suppose the answer to $\q$ over $\Abox$ is \no. Let $\I$ be a model of $\TO$ and $\Abox$ such that $\I,n\not\models A$ for all $n \in \Z$. Extend $\I$ to the fresh atoms $A'$ by setting $\I,n \not\models A'$. Then $\I$ is a model of $\TO''$ and $\Abox$, as required.
\end{proof}


The next proposition, which will be used in the proofs of Theorems~\ref{hornExpSpacehard}  and~\ref{th:linear-omaq-pspace}, shows that deciding $\lang$-rewritability of specific $\LTL_\horn\Xallop$-OMPQs is polynomially reducible to deciding $\lang$-rewritability of Boolean $\LTL_\horn\Xallop$-OMPQs. Recall from~\cite<e.g.,>{DBLP:journals/ai/ArtaleKKRWZ21} that, for any $\LTL_\horn\Xallop$-ontology $\TO$ and any ABox $\Abox$ consistent with $\TO$, there is a \emph{canonical} (or \emph{minimal}) \emph{model} $\mathcal C_{\TO, \Abox}$ of $\TO$ and $\Abox$ such that, for any positive concept $\varkappa$ and any $k \in \Z$,
\begin{equation}
%
\label{eq:can-def}
(\TO, \Abox) \models \varkappa(k)  \quad \text{ iff } \quad \mathcal C_{\TO, \Abox} \models \varkappa(k).
\end{equation}

\begin{proposition}\label{prop:specific-to-boolean++}
Let $\TO$ be an $\LTL_\horn\Xallop$-ontology without occurrences of $\bot$, $A$ an atom, $\varkappa$ a positive concept, and $\Xi$ a signature. Let $X,X'$ be fresh atomic concepts and $\Xi_X = \Xi \cup \{X\}$. Then the following hold:
\begin{description}
\item[$(i)$] The specific OMAQ $\q(x)=(\TO,A(x))$ is $\lang$-rewritable over $\Xi$-ABoxes iff the Boolean OMAQ $\q_X = (\TO\cup\{A\land X\to X'\},X')$  is $\lang$-rewritable over $\Xi_X$-ABoxes.

\item[$(ii)$] The specific OMPQ $\q(x) = (\TO, \varkappa(x))$ is $\lang$-rewritable over $\Xi$-ABoxes iff the Boolean OMPQ $\q_X = (\TO, \varkappa \land X)$ is $\lang$-rewritable over $\Xi_X$-ABoxes.
\end{description}
\end{proposition}
\begin{proof}
We only prove $(ii)$.
Suppose $\rq(x)$ is an $\lang$-rewriting of $\q(x)= (\TO, \varkappa(x))$ over $\Xi$-ABoxes. We show that $\exists x\, (\rq(x)\land X(x))$ is an $\lang$-rewriting of $\q_X$ over $\Xi_X$-ABoxes, that is, for every $\Xi_X$-ABox $\Abox$, we have $\SA\models \exists x\, (\rq(x)\land X(x))$ iff $\mathcal C_{\TO, \Abox} \models (\varkappa \land X) (n)$, for some $n \in \Z$. If $\SA\models \exists x\, (\rq(x)\land X(x))$, then $\SA\models \rq(n)$ and $\SA\models X(n)$, for some $n \in \tem (\Abox)$. Since $\rq(x)$ is a rewriting of $\q(x)$, we have $\mathcal C_{\TO, \Abox} \models \varkappa(n)$, and since $X$ does not occur in $\TO$, we must have $X(n) \in \Abox$.
Conversely, suppose $\mathcal C_{\TO, \Abox} \models (\varkappa \land X) (n)$, for some  $n \in \Z$. Then clearly $X(n) \in \Abox$ and, by~\eqref{eq:can-def}, $n$ is a certain answer to $\q(x)$ over $\Abox$, from which $\SA\models \exists x\, (\rq(x)\land X(x))$.

Suppose $\rq$ is an $\lang$-rewriting of $\q_X$ over $\Xi_X$-ABoxes.
Fix a variable $x$ that does not occur in $\rq$ and let $\rq^-(x)$ be the result of replacing every occurrence of $X(y)$ in $\rq$ with $(x=y)$. We show that $\rq^-(x)$ is an $\lang$-rewriting of $\q(x)$ over $\Xi$-ABoxes.
Given a $\Xi$-ABox $\Abox$, for any $k\in\tem(\Abox)$, we have
$$
\mathcal C_{\TO, \Abox} \models \varkappa(k) \quad  \text{ iff } \quad
\mathcal C_{\TO,  \Abox \cup \{X(k)\}} \models (\varkappa \land X)(k) \quad \text{ iff } \quad \mathfrak S_{ \Abox \cup \{X(k)\}} \models \rq \quad
 \text{ iff } \quad \mathfrak S_{\Abox} \models \rq^-(k)
$$
as required.
\end{proof}

\subsection{Deciding FO-rewritability of linear $\LTL_\horn\Xnext$ OMAQs}\label{sec:linear-2}

In this section, we use $\Abox$ to refer to both the ABox $\Abox$ and its representation as the word $w_{\Abox}$ over the alphabet $\Sigma_\Xi$.
For a linear $\LTL_\horn\Xnext$ ontology $\TO$, let $\textit{idb}(\TO)$ be the set of atoms that occur on the right-hand side of axioms in $\TO$. For an atom $A$ and $j \in \Z$, we define $\nxt^0 A = A$ and, inductively, $\nxt^{j+1} A = \Rnext \nxt^j A$ for $j \ge 0$, and $\nxt^{j-1} A = \Lnext \nxt^j A$ for $j \le 0$.
%
%
Let $\q = (\TO, \varkappa)$ be an $\LTL_\horn\Xnext$ OMPQ. For a type $\tp$ for $\q$ (see Proposition~\ref{Prop:rewr-def}), we denote by $\tp^\Xi$ its restriction to atoms in $\Xi$ and their negations. Given a model $\I$ of $\TO$ and $n \in \Z$, we denote by $\tp_\I(n)$ the type for $\q$ that is true in $\I$ at $n$. For a $\bot$-free $\TO$, we write $\tp_{\TO, \Abox}(n)$ instead of $\tp_{\mathcal C_{\TO, \Abox}}(n)$, where $\mathcal C_{\TO, \Abox}$ is the canonical model of $\TO$ and $\Abox$ with the key property~\eqref{eq:can-def}.

\begin{theorem}\label{th:linear-omaq-pspace}
For $\lang \in \{ \FO(<),\FO(<,\equiv), \FO(<,\MOD)\}$, deciding $\lang$-rewritabi\-lity of linear $\LTL_\horn\Xnext$ OMAQs over $\Xi$-ABoxes is \PSpace-complete.
\end{theorem}
\begin{proof}
To show the upper bound, it suffices, by Lemma~\ref{le:remove-bot}~$(i)$ and Proposition~\ref{prop:specific-to-boolean++}, to consider Boolean $\LTL_\horn\Xnext$ OMAQs $\q = (\TO, B)$ with a $\bot$-free $\TO$.
In view of Remark~\ref{rem1}, we can also assume that the axioms in $\TO$ are of two types:
\begin{align}
&C_1 \land \dots \land C_k \to A',\label{eq:varrho1}\\
&C_1 \land \dots \land C_k \land \nxt^i A \to A',\label{eq:varrho2}
\end{align}
where $k \geq 0$, $C_1, \dots, C_k$ contain no IDB atoms, $A \in \textit{idb}(\TO)$ and $i \in \{-1, 0, 1\}$.



We define a quadruple $\A_\TO^\Xi = (Q, \Sigma_\Xi, \delta, Q_0)$---a 2NFA without final states---giving the transition function $\delta$ as a set of transitions of the form $q \to_{a,d} q'$. Namely, we set $Q_0 = \{q_{0}\}$, $Q = \bigcup_{\alpha \in \TO} Q_\alpha \cup \{q_0, q_h\} \cup \{q_A \mid A \in \textit{idb}(\TO) \}$ and
$$
\delta = \bigcup_{\alpha \in \TO} \delta_\alpha \cup \{q_0 \to_{a, 1} q_0 \mid a \in \Sigma_\Xi\}.
$$
The states in $Q_\alpha$ and transitions in $\delta_\alpha$ are defined as follows. If $\alpha \in \TO$ is of the form~\eqref{eq:varrho1} and $C_i = \nxt^{j_i} A_i$, $1 \leq i \leq k$, then $Q_\alpha = \{q_\alpha\} \cup Q_\alpha'$ and $\delta_\alpha = \{ q_0 \to_{a, 0} q_\alpha \mid a \in \Sigma_\Xi\} \cup \delta_\alpha'$, where $Q_{\alpha}'$ and $\delta_{\alpha}'$ are defined below. If $j_1 < 0$ (the cases $j_1 = 0$ and $j_1 > 0$ are analogous), then $\delta'_{\alpha}$ is such that $\A_\TO^\Xi$ makes $|j_1|$ steps to the left by reading any symbols from $\Sigma_\Xi$. If after that the 2NFA reads any symbol $a$ with $A_1 \not \in a$ (remember that $C_1 = \nxt^{j_1} A_1$), it moves  to the `dead-end' state $q_h$. Otherwise, it makes $|j_1|$ steps to the right and repeats the same process for $C_2 = \nxt^{j_2} A_2$, etc. After executing the transitions for $C_k = \nxt^{j_k} A_k$ and provided that $q_h$ has been avoided, the 2NFA comes to state $q_{A'}$.
%
For $\alpha$ of the form~\eqref{eq:varrho2}, $Q_\alpha$ is the same as above but $\delta_\alpha = \{ q_A \to_{a, 0} q_\alpha \mid a \in \Sigma_\Xi\} \cup \delta_\alpha'$ for the same $\delta_\alpha'$ as above, leading to either $q_h$ or $q_{A'}$.

In what follows, $\mathsf{b}_{\bullet}(\Abox)$ and $\mathsf{b}(\Abox)$, for $\bullet \in \{lr, rr, rl, ll\}$ and $\Abox \in \Sigma_\Xi^*$, are defined with respect to $\A_\TO^\Xi$ (see Section~\ref{sec:2nfa} taking into account that the final states of the 2NFA are not relevant in the definition of $\mathsf{b}_{\bullet}(\Abox)$). Let $X_\Abox(\ell)$ be the reflexive and transitive closure of $\mathsf{b}_{ll}(\Abox^{>\ell}) \circ \mathsf{b}_{rr}(\Abox^{\leq \ell})$, for $0 \leq \ell < |\Abox|$.
Let $N = M + 2 M^2$, where $M$ is the number of occurrences of $\Rnext$ and $\Lnext$ in $\TO$. The proof of the following technical result can be found in Appendix~\ref{lem:23}:

\begin{lemma}\label{th:derivations-runs}
Let $\Abox \in \Sigma_\Xi^*$ be of the form $\emptyset^N \mathcal B \emptyset^N$. 
Then
%
%
$A \in \tp_{\TO, \mathcal A}^{\sig(\TO)}(\ell)$ iff there exists a run $(q_0, 0), \dots, (q, \ell), (q_A, i)$ of $\A_\TO^\Xi$ on $\Abox$, for all $\ell$ with $N \leq \ell < |\mathcal A| - N$.
\end{lemma}

As $\A_\TO^\Xi$ has a run $(q_0, 0), \dots, (q, \ell), (q_A, i)$ on $\Abox$ iff $(q_0, q_A) \in \mathsf{b}_{lr}(\Abox^{\leq \ell}) \circ X_\Abox(\ell)$, for all $\ell < |\Abox|$ and $A \in \sig(\TO)$, we immediately obtain that
\begin{equation}\label{th:atomic-type-2nfa}
\tp_{\TO, \mathcal A}^{\sig(\TO)}(\ell) = \{ A \mid (q_0, q_A) \in \mathsf{b}_{lr}(\Abox^{\leq \ell}) \circ X_\Abox(\ell) \}.
\end{equation}
Define a 2NFA $\A_\q^\Xi = (\Sigma_\Xi, Q', \delta', Q_0,  F)$ with $Q' = Q \cup \{q_B\}$, $\delta' = \delta \cup \{ q_B \to_{a, 1} q_B \mid a \in \Sigma_\Xi \}$, and $F = \{q_B\}$. Using Lemma~\ref{th:derivations-runs}, we obtain:
\begin{equation}
\L_\Xi(\q) = \{ \Abox \in \Sigma_\Xi^* \mid \emptyset^N \Abox \emptyset^N \in \L(\A_\q^\Xi)\}.\label{eq:emptyNlang}
\end{equation}
Our aim is to construct in polynomial space a DFA $\A'$ with $\L_\Xi(\q) = \L(\A')$ whose $\lang$-definability can be decided in \PSpace{}. We construct $\A'$ from $\A_\q^\Xi$ in the same way as in Section~\ref{sec:2nfa} except the definition of $q_0'$ and $F'$, which is now as follows: $q_0' = (\{(q_0,q_0)\}, \mathsf{b}_{rr}(\emptyset^N))$ and $F' = \{(B_{lr}, B_{rr}) \mid (q_0, q_1) \in B_{lr} \circ X \}$, where $X$ is the reflexive and transitive closure of $\mathsf{b}_{ll}(\emptyset^N) \circ B_{rr}$. By~\eqref{eq:emptyNlang}, we have $\L_\Xi(\q) = \L(\A')$, and it is readily seen that $\A'$ is constructible from $\q$ in \PSpace. That $\lang$-definability of $\A'$ is decidable in \PSpace, follows from the proof of Theorem~\ref{thm:2NFA}.

\smallskip

We now establish a matching lower bound. By Lemma~\ref{le:remove-bot} and Proposition~\ref{prop:specific-to-boolean++} $(i)$, it suffices to show it for \emph{specific} linear $\LTL_\horn\Xnext$ OMAQs $\q(x) = (\TO, F_{end}(x))$, which
will be done by reduction of $\lang$-rewritability for DFAs $\A=(Q,\Omega,\delta,q_0,F)$. We set $\Xi = \Omega \cup \{ X, Y\}$ with fresh $X$, $Y$ and construct a linear $\LTL_\horn\Xnext$ ontology $\TO$ with $\textit{idb}(\TO) \subseteq \{\bar q \mid q \in Q \} \cup \{ F_{end}\}$ (treating $\bar q$ as an atomic concept) that simulates the behaviour of the DFA $\A$ by means of the axioms $X \to \Rnext \bar q_0$, $\bar q \land Y \to F_{end}$ for all $q \in F$, $\bar q \land A \to \Rnext \bar r$ for all transitions  $q \to_A r$ in $\delta$, $A \land C \to \bot$ for all distinct $A, C \in \Xi$.
Then
$\L(\A)$ is $\lang$-definable iff $\L_\Xi(\q(x))$ is $\lang$-definable, which is proved similarly to Lemma~\ref{rewrtodefine}.
\end{proof}


\subsection{Deciding FO-rewritability of Linear $\LTL_\horn\Xnext$ OMPQs}\label{sec:lin-ompq}

We next show that $\FO(<)$- and $\FO(<,\equiv)$-definability of linear $\LTL_\horn\Xnext$ OMPQs can be recognised in \PSpace.
By Lemma~\ref{le:remove-bot} and Proposition~\ref{prop:specific-to-boolean++}, it suffices to do this for Boolean OMPQs $\q = (\TO, \varkappa)$ with $\bot$-free $\TO$, in which case we can assume that $\varkappa = \Ldiamond \Rdiamond \varkappa'$. Let $\Type_\q$ be the set of all types for $\q$.

We start with $\FO(<)$-definability. Recall that we established the \PSpace{} upper bound for deciding $\FO(<)$-definability of the language of a given DFA $\A$ in two steps. First, in Theorem~\ref{DFAcrit} $(i)$, we gave a criterion in terms of words in the alphabet of $\A$, and then, in Theorem~\ref{thm:2NFA}, we showed how to check that criterion in \PSpace{}. Similarly, in Theorem~\ref{th:fo-horn-crit} below, we prove a criterion of $\FO(<)$-rewritability of OMPQs we are dealing with in terms of $\Sigma_\Xi$-ABoxes. Then, in Theorem~\ref{th:lin-ompq-fo-pspace}, we show how this criterion can be checked in \PSpace{}.

\begin{theorem}\label{th:fo-horn-crit}
Let $\q = (\TO, \varkappa)$ be an OMPQ with a $\bot$-free $\LTL_\horn\Xallop$-ontology $\TO$. Then $\q$ is not $\FO(<)$-rewritable over $\Xi$-ABoxes iff there exist $\Abox, \mathcal{B}, \mathcal{D} \in \Sigma_\Xi^*$ and $k \geq 2$ such that the following conditions hold:
\begin{description}
\item[$(i)$] $\neg \varkappa \in \tp_{\TO, \Abox \mathcal{B}^{k} \mathcal{D}}(|\Abox|-1) = \tp_{\TO, \Abox \mathcal{B}^{k} \mathcal{D}}(|\Abox \mathcal{B}^{k}|-1)$\textup{;}

\item[$(ii)$] $\varkappa \in \tp_{\TO, \Abox \mathcal{B}^{k+1} \mathcal{D}}(|\Abox \mathcal{B}|-1)  = \tp_{\TO, \Abox \mathcal{B}^{k+1} \mathcal{D}}(|\Abox \mathcal{B}^{k+1}|-1)$.
\end{description}
Moreover, we can find $\Abox, \mathcal{B}, \mathcal{D}$ and $k$ such that $|\Abox|, |\mathcal{D}|, k \leq 2^{O(|\q|)}$.
\end{theorem}
\begin{proof}
Define a DFA $\A = (Q,\Sigma_\Xi,\delta, q_{-1},F)$ by taking $Q = 2^{\Type_\q}$, $q_{-1} = \Type_\q$, $F = \{q \in Q \mid \varkappa \in \tp \text{ for all } \tp \in q\}$, and  $\delta(q, a) = \{ \tp \mid \tp' \to_a \tp \text{
 for some }\tp' \in q\}$, where $\to_a$ was defined in the proof of Proposition~\ref{Prop:rewr-def}. 
As in that proof we can show that $\L_\Xi(\q) = \L(\A)$.
We write $q \Rightarrow_{\Abox} q'$ to say that, having started in state $q$ and read $\Abox \in \Sigma_\Xi^*$, the DFA $\A$ is in state $q'$.

We require the following property of $\A$. For a set $\{\tp_i \mid i \in I\}$ of types for $\q$, let $\bigoplus_{i \in I} \tp_i = \bigcap_{i \in I} \tp_i^+ \cup \bigcup_{i \in I} \tp_i^-$, where $\tp_i^+$ and $\tp_i^-$ are the sets of positive and negated concepts in $\tp_i$, respectively.
%
%
Suppose now $q_{-1} \Rightarrow_{\Abox_0} q_0 \Rightarrow_{\Abox_1} \dots \Rightarrow_{\Abox_{n-1}} q_{n-1}\Rightarrow_{\Abox_n} q_n$ is a run of $\A$ on $\Abox = \Abox_0 \dots \Abox_n$, and let $\bar q_i = \{ \tp \in q_i \mid \tp \to_{\Abox_{i+1} \dots \Abox_n} \tp', \text{ for some }\tp' \in q_n \}$.
Then
\begin{equation}\label{eq:fo-horn-crit2}
\tp_{\TO, \Abox}(i) = \bigoplus  \bar q_i, \quad \text{ for }-1 \leq i \leq n.
\end{equation}
%

$(\Rightarrow)$ Suppose $\q$ is not $\FO(<)$-rewritable. By applying Theorem~\ref{DFAcrit} $(i)$ to $\A$, we find $\Abox, \mathcal{B}, \mathcal{D} \in \Sigma_\Xi^*$ and $k \geq 2$ such that $q_{-1} \Rightarrow_{\Abox} q_0$, $q_0 \Rightarrow_{\mathcal B} q_1$, $q_0 \Rightarrow_{\mathcal B^k} q_0$ and $q_0 \Rightarrow_{\mathcal D} q_0'$, $q_1 \Rightarrow_{\mathcal D} q_1'$, for some $q_0, q_1, q_0',q_1' \in Q$ such that $q_0' \not \in F$ and $q_1' \in F$. Since $q_0' \not \in F$, 
by~\eqref{eq:fo-horn-crit2}, we obtain
$\neg \varkappa \in \tp_{\TO, \Abox \mathcal{B}^{k} \mathcal{D}}(|\Abox|-1) = \tp_{\TO, \Abox \mathcal{B}^{k} \mathcal{D}}(|\Abox \mathcal{B}^{k}|-1)$ as required in $(i)$. And since $q_1' \in F$,~\eqref{eq:fo-horn-crit2} yields $\varkappa \in \tp_{\TO, \Abox \mathcal{B}^{k+1} \mathcal{D}}(|\Abox \mathcal B|-1) = \tp_{\TO, \Abox \mathcal{B}^{k+1} \mathcal{D}}(|\Abox \mathcal{B}^{k+1}|-1)$, as required in $(ii)$.

 $(\Leftarrow)$ Assuming $(i)$ and $(ii)$, let $q_0, q_1, q_2$ be states in $\A$ with $q_{-1} \Rightarrow_{\Abox} q_0 \Rightarrow_{\mathcal B} q_1 \Rightarrow_{\mathcal B^{k-1}} q_2 \Rightarrow_{\mathcal B} q_2'$. Let $q_3, q_3'$ be such that $q_2 \Rightarrow_{\mathcal D} q_3$ and $q_2' \Rightarrow_{\mathcal D} q_3'$. It follows by~\eqref{eq:fo-horn-crit2} that $q_3 \not \in F$ and $q_3' \in F$. Observe that, if we had $q_0 = q_2$, we could conclude that $\q$ is not $\FO(<)$-rewritable, as the conditions of aperiodicity for $\A$ (see the proof of $(\Rightarrow)$) would be satisfied. Since we are not guaranteed that, we use the following property of the canonical models that follow from $(i)$ and $(ii)$: (a) 
 $\tp_{\TO, \Abox \mathcal{B}^k \mathcal{D}}(|\Abox \mathcal{B}^k|-1) = \tp_{\TO, \Abox \mathcal{B}^{kj} \mathcal{D}}(|\Abox \mathcal{B}^{kj}|-1)$, for any $j \geq 1$; (b) 
 $\tp_{\TO, \Abox \mathcal{B}^{k+1} \mathcal{D}}(|\Abox \mathcal{B}^{k+1}|-1) = \tp_{\TO, \Abox \mathcal{B}^{kj+1} \mathcal{D}}(|\Abox \mathcal{B}^{kj+1}|-1)$, for any $j \geq 1$. Take some $i, j \geq 1$ that satisfy $q_0 \Rightarrow_{\Abox \mathcal B^{ki}} q_4 \Rightarrow_{\mathcal B} q_4'  \Rightarrow_{\mathcal B^{kj}} q_4 \Rightarrow_{\mathcal B} q_4'$, for some $q_4, q_4' \in Q$. By $(i)$, $(ii)$, (a) and (b), we have $q_5 \not \in F$ and $q_5' \in F$ for such $q_5$ and $q_5'$ that $q_4 \Rightarrow_{\mathcal D} q_5$ and $q_4' \Rightarrow_{\mathcal D} q_5'$. Therefore, $\q$ is not $\FO(<)$-rewritable, as the conditions of aperiodicity for $\A$ are satisfied (as in the $(\Rightarrow)$-proof with $\Abox$, $\mathcal B$, $\mathcal D$ and $k$ being $\Abox \mathcal B^{ki}$, $\mathcal B$, $\mathcal D$ and $kj$, respectively).

To establish the bounds on the size of $\Abox$, $\mathcal{D}$ and $k$, we first notice that there is $\Abox$ with $|\Abox| \leq 2|\Type_\q|^2$. Indeed, consider the sequence
$$
\big(\tp_{\TO, \Abox \mathcal B^k \mathcal D}(0), \tp_{\TO, \Abox \mathcal B^{k+1} \mathcal D}(0)\big), \dots, \big(\tp_{\TO, \Abox \mathcal B^k \mathcal D}(|\Abox|-2), \tp_{\TO, \Abox \mathcal B^{k+1} \mathcal D}(|\Abox|-2)\big).
$$
If the $i$-th member of this sequence is equal to its $j$-th member, for $i < j$, then we take $\Abox'= \Abox^{<i}\Abox^{\geq j}$, where $\Abox^{<i}$ is the prefix of $\Abox$ before $i$ and $\Abox^{\geq j}$ the suffix of $\Abox$ starting at $j$. Then $\tp_{\TO, \Abox' \mathcal B^k \mathcal D}(|\Abox'|-1) = \tp_{\TO, \Abox \mathcal B^k \mathcal D}(|\Abox|-1)$ and $\tp_{\TO, \Abox' \mathcal B^{k+1} \mathcal D}(|\Abox' \mathcal B|-1) = \tp_{\TO, \Abox \mathcal B^k \mathcal D}(|\Abox \mathcal B|-1)$, and conditions $(i)$ and $(ii)$ are satisfied with $\mathcal A'$ in place of $\mathcal A$. In the same way we obtain the upper bound for $\mathcal D$. To show that there exists $k \leq 2|\Type_\q|^2$, we consider the sequence
\begin{multline*}
\big(\tp_{\TO, \Abox \mathcal B^k \mathcal D}(|\Abox \mathcal B|-1), \tp_{\TO, \Abox \mathcal B^{k+1} \mathcal D}(|\Abox \mathcal B^2|-1)\big), \dots, \\
\big(\tp_{\TO, \Abox \mathcal B^k \mathcal D}(|\Abox \mathcal B^{k-1}|-1), \tp_{\TO, \Abox \mathcal B^{k+1} \mathcal D}(|\Abox \mathcal B^{k}|-1)\big).
\end{multline*}
Clearly, if the $i$-th member of this sequence is equal to its $j$-th member, for $i < j$, then conditions $(i)$ and $(ii)$ are satisfied with $k-(j-i)$ in place of $k$.
\end{proof}

In the theorem above, we did not claim that there is $\mathcal B$ with $|\mathcal B| \leq 2^{O(|\q|)}$. However, this is indeed the case for linear $\LTL_\horn\Xnext$-ontologies, as follows from the proof of the next result:

\begin{theorem}\label{th:lin-ompq-fo-pspace}
Deciding $\FO(<)$-rewritability of OMPQs $\q = (\TO, \varkappa)$ with a linear $\LTL_\horn\Xnext$-ontology $\TO$ over $\Xi$-ABoxes can be done in \PSpace{}.
\end{theorem}
\begin{proof}
By Theorem~\ref{th:fo-horn-crit}, we need to check the existence of $\Abox, \mathcal{B}, \mathcal{D}$, $k \geq 2$, such that $|\Abox|, |\mathcal{D}|, k \leq 2^{O(|\q|)}$ and conditions $(i)$ and $(ii)$ hold. Without loss of generality, we assume that $\Abox$ has a prefix $\emptyset^N$ and $\mathcal{D}$ has a suffix $\emptyset^N$. (As before, $N = M + 2 M^2$, where $M$ is the number of occurrences of $\Rnext$ and $\Lnext$ in $\TO$.)

We start by guessing numbers $N_{\Abox} = |\Abox|$, $N_{\mathcal D} = |\mathcal D|$ and $k$. 
We guess two types $\tp_0$ and $\tp_1$ that represent $\tp_{\TO, \Abox \mathcal{B}^k \mathcal{D}}(N)$ and $\tp_{\TO, \Abox \mathcal{B}^k \mathcal{D}}(|\Abox|-1)$, respectively, and three types $\tp_0'$, $\tp_0''$, $\tp_1'$ that represent $\tp_{\TO, \Abox \mathcal{B}^{k+1} \mathcal{D}}(N)$, $\tp_{\TO, \Abox \mathcal{B}^{k+1} \mathcal{D}}(|\Abox|-1)$ and, respectively, $\tp_{\TO, \Abox \mathcal{B}^{k+1} \mathcal{D}}(|\Abox \mathcal B|-1)$.
Next, we compute $\mathsf{b}(\emptyset^N)$ and guess $\mathsf{b}(\Abox)$, $\mathsf{b}(\mathcal B)$, $\mathsf{b}(\mathcal D)$. Note that, given $\mathsf{b}(\mathcal B)$, we are able to compute $\mathsf{b}(\mathcal X)$ for each $\mathcal X \in \{ \mathcal B^i \mid 1 \leq i \leq k+1 \}$.
Now, we guess $\Abox$---symbol by symbol---by means of a sequence of pairs
$$(\mathsf{b}(\Abox^{\leq 0}), \mathsf{b}(\Abox^{>0})), \dots, (\mathsf{b}(\Abox^{\leq N_\Abox-1}), \mathsf{b}(\Abox^{> N_\Abox-1}))
$$
such that $\mathsf{b}(\Abox^{\leq i}) \cdot \mathsf{b}(\Abox^{> i}) = \mathsf{b}(\Abox)$, for all $i$, and there exist $a_i \in \Sigma_\Xi$ with $\mathsf{b}(\Abox^{\leq i+1}) = \mathsf{b}(\Abox^{\leq i}) \cdot \mathsf{b}(a_i)$ and $\mathsf{b}(\Abox^{> {i}}) = \mathsf{b}(a_i) \cdot \mathsf{b}(\Abox^{> i+1})$; we also require that $a_i = \emptyset$ for $i < N$. Observe that, by~\eqref{th:atomic-type-2nfa}, the pairs of the sequence with $i \geq N$ together with $\mathsf{b}(\mathcal B)$ and $\mathsf{b}(\mathcal D)$ give $\tp_{\TO, \Abox \mathcal{B}^k \mathcal{D}}^{\sig(\TO)}(i)$. When computing $\tp_{\TO, \Abox \mathcal{B}^k \mathcal{D}}^{\sig(\TO)}(N)$, we check whether it is subsumed by $\tp_0$ (if not, the algorithm terminates with an answer \no). We also need to check that
$\varkappa' \in \tp_{\mathcal C_{\TO, \{A(0) \mid A \in \tp_0\}}}(0)$ \text{implies} $\varkappa' \in \tp_0$,
for each $\varkappa'$ of the form $\Lbox \varkappa''$, $\Ldiamond \varkappa''$ from $\sub_{\q}$ (if not, the algorithm terminates and returns \no). We have now checked that the type $\tp_0$ is potentially guessed correctly (subject to further checks). We can apply the same method to check that $\tp_0'$ is potentially guessed correctly. For the remaining $N < i < N_\Abox$, since $\tp_{\TO, \Abox \mathcal{B}^k \mathcal{D}}(i)$ is determined by $\tp_{\TO, \Abox \mathcal{B}^k \mathcal{D}}^{\sig(\TO)}(i)$ and $\tp_{\TO, \Abox \mathcal{B}^k \mathcal{D}}(i-1)$, we are able to compute $\tp_{\TO, \Abox \mathcal{B}^k \mathcal{D}}(|\Abox|-1)$ or obtain a conflict, e.g., $\neg \Rdiamond A \in \tp_{\TO, \Abox \mathcal{B}^k \mathcal{D}}(i-1)$ and $A \in \tp_{\TO, \Abox \mathcal{B}^k \mathcal{D}}^{\sig(\TO)}(i)$. In the latter case, the algorithm terminates answering \no. In the former case, we check if $\tp_{\TO, \Abox \mathcal{B}^k \mathcal{D}}(|\Abox|-1)$ is equal to $\tp_1$, in which case $\tp_1$ is guessed correctly, and if not, the algorithm terminates answering \no. In the same way, we check if $\tp_0''$ is guessed correctly using $\mathcal{C}_{\TO, \Abox \mathcal{B}^{k+1} \mathcal{D}}$.

Now, we show how to check that the types $\tp_{\TO, \Abox \mathcal{B}^{k} \mathcal{D}}(i)$, for \mbox{$|\Abox| \leq i <  |\Abox \mathcal{B}^{k}|$},  are correct, that $\tp_1'$ is guessed correctly, and that the types $\tp_{\TO, \Abox \mathcal{B}^{k+1} \mathcal{D}}(i)$ with $|\Abox \mathcal B| \leq i <  |\Abox \mathcal{B}^{k+1}|$ are correct. We only demonstrate the algorithm for $\tp_{\TO, \Abox \mathcal{B}^{k} \mathcal{D}}(i)$. Observe that $\varkappa' \in \tp_{\TO, \Abox \mathcal{B}^{k} \mathcal{D}}(i)$ iff $\varkappa' \in \tp_{\TO, \Abox \mathcal{B}^{k} \mathcal{D}}(j)$ iff $\varkappa' \in \tp_1$, for any $\varkappa'$ of the form $\Box \varkappa''$ and $\Diamond \varkappa''$ from $\sub_{\q}$ and $|\Abox|-1 \leq i,j <  |\Abox \mathcal{B}^{k}|$. To do the required check, we need to guess a sequence of pairs
\begin{equation}\label{eq:b-sequence}
\big(\mathsf{b}(\mathcal B^{\leq 0}), \mathsf{b}(\mathcal B^{>0})\big), \dots, \big(\mathsf{b}(\mathcal B^{\leq |\mathcal B|-1}), \mathsf{b}(\mathcal B^{> |\mathcal B|-1})\big)
\end{equation}
such that $\mathsf{b}(\mathcal B^{\leq i}) \cdot \mathsf{b}(\mathcal B^{> i}) = \mathsf{b}(\mathcal B)$, for all $i$, and there are $a \in \Sigma_\Xi$ with $\mathsf{b}(\mathcal B^{\leq i+1}) = \mathsf{b}(\mathcal B^{\leq i}) \cdot \mathsf{b}(a)$ and $\mathsf{b}(\mathcal B^{> {i}}) = \mathsf{b}(a) \cdot \mathsf{b}(\mathcal B^{> i+1})$. While we do not have any bound on $|\mathcal B|$ yet, we can easily observe that any sequence~\eqref{eq:b-sequence} with repeating members at positions $0 \leq i' < i'' \leq |\mathcal B|-1$ is equivalent for the purposes of this proof to the sequence with all the members $i', \dots, i''-1$ removed. Thus, we can assume that $|\mathcal B|\leq 2^{O(|\q|)}$, if $\mathcal B$ exists at all.
By~\eqref{th:atomic-type-2nfa}, using an element $i$ of this sequence, we can compute $\tp_{\TO, \Abox \mathcal{B}^k \mathcal{D}}^{\sig(\TO)}(|\Abox \mathcal{B}^{j}|+i)$, for all $j < k$. We only need to check that such an atomic type is not in conflict with the temporal concepts in $\tp_1$, e.g., $\Lbox A \in \tp_1$ and $\neg A \in \tp_{\TO, \Abox \mathcal{B}^k \mathcal{D}}^{\sig(\TO)}(|\Abox \mathcal{B}^{j}|+i)$. If a conflict is detected for some $i$ and $j$,  the algorithm answers \no. Here, we also verify that $\tp_{\TO, \Abox \mathcal{B}^{k} \mathcal{D}}(|\Abox \mathcal{B}^{k}|-1) = \tp_1$ and $\tp_{\TO, \Abox \mathcal{B}^{k+1} \mathcal{D}}(|\Abox \mathcal{B}^{k+1}|-1) = \tp_1'$.
Finally, we check that all the types $\tp_{\TO, \Abox \mathcal{B}^{k} \mathcal{D}}(|\Abox \mathcal{B}^{k}|+i)$ and $\tp_{\TO, \Abox \mathcal{B}^{k+1} \mathcal{D}}(|\Abox \mathcal{B}^{k+1}|+i)$ are correct, for $0 \leq i < N_{\mathcal D} - N$. Details are left to the reader.
\end{proof}

A criterion for $\FO(<,\equiv)$-definability of linear $\LTL_\horn\Xallop$ OMPQs (cf.\ Theorem~\ref{DFAcrit} $(ii)$) is given by the next theorem whose (rather technical) proof can be found in Appendix~\ref{le:26}:

\begin{theorem}\label{th:foe-horn-crit}
Let $\q = (\TO, \varkappa)$ be an OMPQ with a $\bot$-free $\LTL_\horn\Xallop$-ontology $\TO$. Then $\q$ is not $\FO(<,\equiv)$-rewritable over $\Xi$-ABoxes iff there are $\Abox, \mathcal{B}, \mathcal{D} \in \Sigma_\Xi^*$ and $k \geq 2$, such that $(i)$ and $(ii)$ from Theorem~\ref{th:fo-horn-crit} hold and
there are $\mathcal{U}, \mathcal{V} \in \Sigma_\Xi^*$, such that $\mathcal B = \mathcal V \mathcal U$, $|\mathcal{U}|= |\mathcal{V}|$,
\begin{description}
\item[$(iii)$]
 $\tp_{\TO, \Abox \mathcal{B}^{k} \mathcal{D}}(|\Abox \mathcal{B}^i|-1) = \tp_{\TO, \Abox \mathcal{B}^{k} \mathcal{D}}(|\Abox \mathcal{B}^i \mathcal V|-1)$, for all $i < k$, and

\item[$(iv)$] 
 $\tp_{\TO, \Abox \mathcal{B}^{k+1} \mathcal{D}}(|\Abox \mathcal{B}^i|-1) = \tp_{\TO, \Abox \mathcal{B}^{k+1} \mathcal{D}}(|\Abox \mathcal{B}^i \mathcal V|-1)$, for all $i$, $1 \leq i \leq k$.
\end{description}
\end{theorem}
This result allows us to obtain a \PSpace{} algorithm by a straightforward modification of the proof of Theorem~\ref{th:lin-ompq-fo-pspace}. Thus, we obtain:
\begin{theorem}\label{th:lin-ompq-foe-pspace}
Deciding $\FO(<,\equiv)$-rewritability of OMPQs $\q = (\TO, \varkappa)$ with a linear $\LTL_\horn\Xnext$-ontology $\TO$ over $\Xi$-ABoxes can be done in \PSpace{}.
\end{theorem}

At present, we do not know how to transform Theorem~\ref{DFAcrit} $(iii)$ into \PSpace-checkable conditions on the canonical models and ABoxes, so the complexity of deciding $\FO(<,\MOD)$-rewritability of linear $\LTL_\horn\Xnext$ OMPQs remains open.



\section{$\FO(<)$-rewritability of $\LTL_\krom\Xnext$ OMAQs and $\LTL_\core\Xnext$ OMPQs}\label{sec:others}

Our final aim is to look for non-trivial classes of OMQs deciding $\FO(<)$-rewritability of which could be `easier' than \PSpace. Syntactically, the simplest type of axioms~\eqref{axiom1} are binary clauses $C_1 \to C_2$ and $C_1 \land C_2 \to \bot$, known as \core{} axioms, which together with $C_1 \lor C_2$ form the class Krom. In the atemporal case, the W3C standard language \OWLQL{}, designed specifically for ontology-based data access, admits core clauses only and uniformly guarantees FO-rewritability~\cite{CDLLR07,ACKZ09}.

In this section, we use NFAs with $\varepsilon$-transitions that can be defined as 2NFAs where backward transitions $q \to_{a,-1} q'$ are disallowed and transitions of the form $q \to_{a,0} q'$ hold for all $a \in \Sigma$, in which case we write $q \to_{\varepsilon} q'$.

As we saw in the proof of Theorem~\ref{thm:kromlowerexp}, OMPEQs with disjunctive axioms can simulate $\LTL_\horn\Xnext$ OMAQs, and so are too complex for the purposes of this section. On the other hand, $\LTL_\krom\Xnext$ OMAQs and $\LTL_\core\Xnext$ OMPQs are all $\FO(<,\equiv)$-rewritable~\cite{DBLP:journals/ai/ArtaleKKRWZ21}. Below, we focus on deciding $\FO(<)$-rewritability of OMQs in these classes.

\subsection{$\LTL_\krom\Xnext$ OMAQs}

\begin{theorem}\label{thm:coNP}
Deciding $\FO(<)$-rewritability of Boolean and specific $\LTL_\krom\Xnext$ OMAQs over $\Xi$-ABoxes is \coNP-complete.
\end{theorem}
\begin{proof}

Suppose $\q=(\TO,A)$ ($\q(x)=(\TO,A(x))$) is a Boolean (respectively, specific) $\LTL_\krom\Xnext$ OMAQ. 
Using the form of Krom axioms, one can show~\cite<e.g.,>{DBLP:journals/ai/ArtaleKKRWZ21} that, for any ABox $\Abox$ and $l\in\mathbb Z$ (respectively, $l\in \tem(\Abox)$), we have $(\TO,\Abox) \models A(l)$ iff at least one of the following holds:
$(i)$~there is $B(k)\in\Abox$ such that $\TO\models B\to\nxt^{l-k}A$ (the $\nxt^n$ notation was defined in Section~\ref{sec:linear-2}); $(ii)$~$\TO$ and $\Abox$ are inconsistent, i.e., there exist $k_1\le k_2$,  $B(k_1)\in\Abox$ and $C(k_2)\in\Abox$ such that $\TO\models B\to\Rnext^{k_2-k_1}\neg C$.




Let $\textit{lit}(\q)=\{C,\neg C \mid C\in \sig(\q)\}$. For any $L_1,L_2\in \textit{lit}(\q)$, we construct a unary NFA $\A_{L_1L_2}$ of size $O(|\q|)$ that accepts the language $\L_{L_1L_2}=\{a^n\mid\TO\models L_1\to\Rnext^n L_2, \ n \ge 0\}$ over the alphabet $\{a\}$.
The set of its states is $\textit{lit}(\q)$, $L_1$ is the initial state, $L_2$ the only accepting state, and the transitions are
%
$L \to_a L'$ if $\TO\models L \to\Rnext L'$, and $L \to_\varepsilon L'$ if $\TO\models L \to L'$.
For $\Xi\subseteq \sig(\q)$, we define two sets: $\Xi^\exists_A=\{B\in\Xi\mid(\TO,\{B(0)\})\models\exists x\, A(x)\}$ and $\Xi^\forall_A=\{B\in\Xi\mid (\TO,\{B(0)\}) \models\forall x\, A(x)\}$.

\begin{lemma}
$(a)$ The Boolean OMAQ $\q$ is $\FO(<)$-rewritable over $\Xi$-ABoxes iff, for any $B,C\in\Xi\setminus\Xi^\exists_A$, the language $\L_{B\neg C}$ is $\FO(<)$-definable.

$(b)$ The specific OMAQ $\q(x)$ is $\FO(<)$-rewritable over $\Xi$-ABoxes iff the following conditions are satisfied\textup{:}
\begin{itemize}
\item[$(b_1)$] for all $B\in\Xi$, the languages $\L_{BA}$ and $\L_{\neg A\neg B}$ are $\FO(<)$-definable\textup{;}

\item[$(b_2)$] for all $B,C\in\Xi\setminus\Xi^\forall_A$ such that at least one of the $\L_{BA}$ and $\L_{\neg A\neg C}$ is finite, the language $\L_{B\neg C}$ is $\FO(<)$-definable.
\end{itemize}
%
\end{lemma}
\begin{proof}
$(a,\, \Rightarrow)$ If $\q$ is $\FO(<)$-rewritable, then $\L_{\Xi}(\q)$ over the alphabet $\Sigma_\Xi$ is $\FO(<)$-definable, and so is the language $\L_{\Xi}(\q)\cap \L(\{B\}\emptyset^*\{C\})$, for any $B,C \in \Xi$. For $B,C \in \Xi \setminus \Xi^\exists_A$, we have $\{B\}\emptyset^n\{C\}\in \L_{\Xi}(\q)$ iff $\TO\models B\to\Rnext^{n+1}\neg C$ iff $a^{n+1}\in\L_{B\neg C}$. Therefore, $\L_{B\neg C}$ is $\FO(<)$-definable.

$(a,\, \Leftarrow)$ For a $\Xi$-ABox $\Abox$, the certain answer to $\q$ is \yes\ iff either there is $B(k)\in\Abox$, for some $B\in\Xi^\exists_A$, or there are $B,C\in\Xi\setminus\Xi^\exists_A$ and $k\le l$ such that $B(k),C(l)\in\Abox$ and $\TO\models B\to\Rnext^{k-l}\neg C$. As these conditions are $\FO(<)$-definable, $\q$ is $\FO(<)$-rewritable.

$(b,\, \Rightarrow)$ If $\q(x)$ is $\FO(<)$-rewritable, then $\L_{\Xi}(\q(x))$ over the alphabet $\Gamma_\Xi$ is $\FO(<)$-definable, and so are the languages $\L_{\Xi}(\q(x))\cap \L(\{B\}\emptyset^*\emptyset')$ 
and \mbox{$\L_{\Xi}(\q(x))\cap \L(\emptyset'\emptyset^*\{B\})$}, for any $B \in \Xi$. We have $\{B\}\emptyset^n\emptyset'\in \L_{\Xi}(\q(x))$ iff $\TO\models B\to\Rnext^{n+1}A$ and $\emptyset'\emptyset^*\{B\}\in \L_{\Xi}(\q(x))$ iff $\TO\models B\to\Lnext^{n+1}A$. Therefore, $\L_{BA}$ and $\L_{\neg A\neg B}$ are $\FO(<)$-definable.

Let $B,C\in\Xi\setminus\Xi^\forall_A$ and $\L_{BA}$ be finite. There is $l\in\mathbb Z$ with $(\TO,\{C(0)\})\not\models A(l)$ and there is $k$ with $k>n$ for all $a^n\in \L_{BA}$. For $m>k+|l|$, we have $(\TO,\{B(0),C(m)\})\models A(m+l)$ iff $\TO\models B\to \Rnext^m \neg C$. So $\L_{B\neg C}$ is $\FO(<)$-definable. The case when $\L_{\neg A\neg C}$ is finite is similar.

$(b,\,\Leftarrow)$ Assuming that conditions $(b_1)$ and $(b_2)$ hold, we define formulas $\varphi_{B \neg C}$, for any $B, C \in \Xi$. If $\L_{B \neg C}$ is $\FO(<)$-definable,  then set $\varphi_{B \neg C} = \exists x, y\, (B(x) \land C(y) \land \psi(x,y))$, where $\psi(x,y)$ is an $\FO(<)$-formula saying that $a^{y-x} \in \L_{B \neg C}$.
Suppose $B, C \not \in \Xi_A^\forall$ and $\L_{B \neg C}$ is not $\FO(<)$-definable. It follows from $(b_2)$ that  both $\L_{BA}$ and $\L_{\neg A\neg C}$ are infinite.
By $(b_1)$ and the folklore fact that every star-free language over a unary alphabet is either finite or cofinite,
we have $n_1, n_2 \in \mathbb N$ such that $a^k \in \L_{BA}$ for all $k \geq n_1$ and $a^k \in \L_{\neg A\neg C}$ for all $k \geq  n_2$. Then we set $\varphi_{B \neg C} = \exists x, y\, (B(x) \land C(y) \land \psi(x,y))$, where $\psi(x,y)$ is an $\FO(<)$-formula saying that $y-x < n_1+n_2$ and $a^{y-x} \in \L_{A\neg C}$. Finally, for $B, C \in \Xi$ such that either $B$ or $C$ is not in $\Xi_A^\forall$ and $\L_{B \neg C}$ is not $\FO(<)$-definable, we set $\varphi_{B \neg C} = \bot$. For $B \in \Xi$, let $\varphi_{BA}(x) = \exists y\, (B(y) \land \psi(y,x))$, where $\psi(y,x)$ is an $\FO(<)$-formula saying that $a^{x-y} \in \L_{BA}$, which exists by $(b_1)$. Similarly, let $\varphi_{\neg A\neg B}(x) = \exists y\, (B(y) \land \psi(y,x))$, where $\psi(y,x)$ is an $\FO(<)$-formula saying that $a^{y-x} \in \L_{\neg A \neg B}$. We claim that
$$
\varphi(x) = \bigvee_{B \in \Xi} (\varphi_{BA}(x) \lor \varphi_{\neg A \neg B}(x)) \lor \bigvee_{B, C \in \Xi} \varphi_{B \neg C}
$$
is an $\FO(<)$-rewriting of $\q(x)$ over $\Xi$-ABoxes. Indeed,  let $(\TO,\Abox) \models A(l)$ for $l \in \tem(\Abox)$. If $(i)$ at the beginning of the proof of Theorem~\ref{thm:coNP} holds, then we have $\SA \models \varphi_{BA}(l)$ or $\SA \models \varphi_{\neg A \neg B}(l)$, so $\SA \models \varphi(l)$. If $(ii)$ holds, consider the $B$ and $C$ given by it. If $\L_{B\neg C}$ is $\FO(<)$-definable, then $\SA \models \varphi_{B \neg C}$ and $\SA \models \varphi(l)$ as required. If $\L_{B\neg C}$ is not $\FO(<)$-definable and $B, C \not \in \Xi_A^\forall$, consider $B(k_1) \in \Abox$ and $C(k_2) \in \Abox$ such that $k_2-k_1 \in \L_{B \neg C}$ given by $(ii)$. If $k_2 - k_1 < n_1 + n_2$, then $\SA \models \varphi_{B \neg C}$ and $\SA \models \varphi(l)$ as required. If, on the contrary, $k_2 - k_1 \geq n_1 + n_2$, we have  either $l-k_1 \in \L_{BA}$ or $k_2 - l \in \L_{\neg A \neg B}$. Then $\SA \models \varphi_{BA}(l)$ or $\SA \models \varphi_{\neg A \neg B}(l)$, so $\SA \models \varphi(l)$. Finally, if either $B$ or $C$ is in $\Xi_A^\forall$, we have either $\SA \models \varphi_{BA}$ or $\SA \models \varphi_{CA}$ or $\SA \models \varphi_{\neg A \neg B}$ or $\SA \models \varphi_{\neg A \neg C}$, so $\SA \models \varphi(l)$.
The proof that $\SA \models \varphi(l)$ implies $(\TO,\Abox) \models A(l)$ is similar and left to the reader.
\end{proof}

Thus, to check $\FO(<)$-rewritability of $\q$ and $\q(x)$, 
it suffices to check $\FO(<)$-definability, emptiness and finiteness of the languages of the form $\L_{L_1L_2}$, for $L_1, L_2 \in \lit(\q)$. Emptiness and finiteness can be checked in \NL{} in the size of $\A_{L_1L_2}$.
%
Using~Stockmeyer \& Meyer's~\citeyear[Theorem 6.1]{DBLP:conf/stoc/StockmeyerM73}, one can show that deciding $\FO(<)$-definability of the language of a unary NFA is \coNP-complete, which gives the required upper bound.
To establish \coNP-hardness, for any given unary NFA $\A=(Q,\{a\},\delta,Q_0,F)$ with $Q=\{Q_0,\dots,Q_n\}$, we define an $\LTL_{\core}\Xnext$ ontology $\TO_\A$ with $\sig(\TO_\A) = Q \cup \{X,Y\}$ and the axioms $X\to \Rnext Q_0$, $Q_i\wedge Y\to \bot$, for every $Q_i\in F$, and $Q_i\to \Rnext Q_j$, for every transition $Q_i \to_a Q_j$ in $\A$. Let $A$ be a fresh concept name.
The OMAQ $\q= (\TO_\A,A)$ (respectively, $\q(x) = (\TO_\A,A(x))$)  is $\FO(<)$-rewritable over $\{X,Y\}$-ABoxes iff $\L(\A)$ is $\FO(<)$-definable because $(\TO,\Abox)\models A(l)$ for some $l \in \Z$ (respectively, $l \in \tem(\Abox)$), for an $\{X,Y\}$-ABox $\Abox$, iff $\Abox$ is inconsistent with $\TO_\A$ iff there are $X(i),Y(j)\in\Abox$ with $a^{j-i-1}\in\L(\A)$.
\end{proof}

Our next result deals with a weaker (core) ontology language but more expressive queries.


\subsection{$\LTL_\core\Xnext$ OMPEQs}

\begin{theorem}\label{thm:corepi-upper}
Deciding $\FO(<)$-rewritability of Boolean and specific $\LTL_\core\Xnext$ OMPEQs over $\Xi$-ABoxes is $\Pi^p_2$-complete. 
\end{theorem}
\begin{proof}
By Proposition~\ref{prop:specific-to-boolean++} $(ii)$ and Lemma~\ref{le:remove-bot}, it suffices to consider Boolean OMPEQs $\q = (\TO, \varkappa)$ with a $\bot$-free $\TO$. Also, for the same technical reasons as in Section~\ref{sec:lin-ompq}, we can assume that $\varkappa$ takes the form $\Ldiamond \Rdiamond \varkappa'$.
We first observe that checking $\FO(<)$-definability of $\L_\Xi(\q)$ can be reduced to checking $\FO(<)$-definability of finitely many simpler languages.
For $n \geq 0$, let
$$
W_{n, \Xi}=\{a_1 \ldots a_k \in \Sigma_\Xi^*\mid |a_i| \geq 1, \
\sum_{i=1}^k| a_i|\le n\}.
$$
With each $\mathcal B = a_1\ldots a_k \in W_{|\varkappa|, \Xi}$ we associate the languages
$$
\L_{\mathcal B}^1 =  \L((\emptyset^* a_1)\ldots (\emptyset^* a_k)\emptyset^*)\quad \text{and} \quad
\L_{\mathcal B}=  \L_{\mathcal B}^1 \cap \L_\Xi(\q).
$$
For $\mathcal U=u_1\dots u_k$ and $\mathcal V=v_1\dots v_l$ in $\Sigma^*_\Xi$, we write $\mathcal U \le \mathcal V$ if $k=l$ and $u_i\subseteq v_i$, for all $i$.
%
Let $\L_{\mathcal B}^\uparrow=\{\mathcal V \in\Sigma_\Xi^*\mid \exists\, \mathcal U \in\L_{\mathcal B}\ \mathcal U \le \mathcal V\}$. We show that
\begin{equation}\label{eq:unicorn}
  \textstyle \L_\Xi(\q) ~=~ \bigcup_{\mathcal B \in W_{|\varkappa|, \Xi}}\L_{\mathcal B}^{\uparrow}.
\end{equation}
Let $\Abox \in \L_\Xi(\q)$, and so $(\TO, \Abox) \models \exists x\, \varkappa(x)$. Observe that, for any $\Abox$ and $j \in \Z$, $(\TO,\Abox) \models \varkappa(j)$ iff $(\TO,\Abox') \models\varkappa(j)$, for some $\Abox' \subseteq \Abox$ with $|\Abox| \leq |\varkappa|$. The latter statement is shown by induction on the construction of $\varkappa$, where the base case $\varkappa = A$ follows from the proof of Theorem \ref{thm:coNP}, and left to the reader. This observation implies that $(\TO, \Abox') \models \exists x\, \varkappa(x)$ for some $\Xi$-ABox $\Abox' \subseteq \Abox$ with $|\Abox'| \leq |\varkappa|$. Let $\mathcal B$ be the result of removing all $\emptyset$ from $\mathcal A'$ (viewed as a word). Clearly, $\mathcal B \in W_{|\varkappa|, \Xi}$ and $\Abox \in \L_{\mathcal B}^\uparrow$. The converse inclusion follows from the fact that $\Abox \in \L_\Xi(\q)$ implies $\Abox' \in \L_\Xi(\q)$ for any $\Abox \subseteq \Abox'$.


\begin{lemma}\label{lwstarfree}
The language $\L_\Xi(\q)$ is $\FO(<)$-definable iff $\L_{\mathcal B}$ is $\FO(<)$-definable, for every $\mathcal B  \in W_{|\varkappa|, \Xi}$.
\end{lemma}
\begin{proof}
($\Rightarrow$) If $\L_\Xi(\q)$ is $\FO(<)$-definable, then so is $\L_{\mathcal B}$ as $\L^1_{\mathcal B}$ is $\FO(<)$-definable.

($\Leftarrow$) Suppose $\L_{\mathcal B}$ is $\FO(<)$-definable for any $\mathcal B \in W_{|\varkappa|, \Xi}$. 
By~\eqref{eq:unicorn}, it suffices to prove that $\L_{\mathcal B}^\uparrow$ is $\FO(<)$-definable.
For $0 \leq  l \leq k$, let $\L^1_{\mathcal B, l} = \L\big((\emptyset^*a_{k - l+1})\ldots (\emptyset^* a_k) \emptyset^*\big)$. Note that $\L^1_{\mathcal B, 0} = \L(\emptyset^*)$ and $\L^1_{\mathcal B, k} = \L^1_{\mathcal B}$.
We prove by induction on $l$ that, for any $\L \subseteq \L^1_{\mathcal B, l}$, if $\L$ is $\FO(<)$-definable, then $\L^\uparrow$ is $\FO(<)$-definable. Let $l = 0$ and suppose $\L$ is $\FO(<)$-definable. Then $\L^1_{\mathcal B, l} = \L(\emptyset^*)$, and so $\L$ is a finite or cofinite subset of $\L(\emptyset^*)$. Either way the language $\L^\uparrow$ is $\FO(<)$-definable.
Now, suppose $l > 0$ and $\L \subseteq \L^1_{\mathcal B, l}$ is $\FO(<)$-definable.
Let $\A=(Q,\Gamma_\Xi,\delta,q_0,F)$ be a minimal DFA accepting $\L$. Let $Q_{\emptyset}=\{q\in Q\mid\exists i \, \delta_{\emptyset^i}(q_0)=q\}$. 
 For $p\in Q_{\emptyset}$, let $\L_p$ be the language accepted by the automaton $(Q_\emptyset,\{\emptyset\},\delta|_{Q_{\emptyset}},\{q_0\},\{p\})$ and let $\L_p'$ be the language accepted by the automaton $(Q\setminus Q_{\emptyset},\Gamma_\Xi,\delta|_{Q\setminus Q_\emptyset},\delta_{\avec{a}_{k-l+1}}(p),F)$. Clearly, $\L_p'\subseteq\L^1_{\mathcal B,  l-1}$ and both $\L_p$ and $\L_p'$ are $\FO(<)$-definable. Since $\L_p'\subseteq\L(\emptyset^*)$ and by IH, the languages $\L_p^\uparrow$ and $\L_p'\mathstrut^\uparrow$ are $\FO(<)$-definable, and so $\L^\uparrow=\bigcup_{p\in Q_\emptyset}(\L_p^\uparrow\cdot(\bigcup_{\avec{a} \supseteq \avec{a}_{k-l+1}}\{\avec{a}\})\cdot\L_p'\mathstrut^\uparrow)$ is $\FO(<)$-definable as well.
\end{proof}

Now we give a criterion of checking $\FO(<)$-definability of $\L_w$ (cf.~Theorem~\ref{th:fo-horn-crit}).
\begin{lemma}\label{th:cloud}
Let $w = a_1 \dots a_k \in W_{|\varkappa|, \Xi}$. Then $\L_w$ is not $\FO(<)$-definable iff there are words $\Abox = (\emptyset^{i_1} a_1) \dots (\emptyset^{i_{l-1}} a_{l-1}) \emptyset^{i_l}$, $\mathcal D = (\emptyset^{i_l'} a_l) (\emptyset^{i_{l+1}} a_{l+1}) \dots (\emptyset^{i_k} a_k) \emptyset^{i_{k+1}}$, $\mathcal B = \emptyset^n$ and $k \geq 2$  such that $(i)$ and $(ii)$ from Theorem~\ref{th:fo-horn-crit} hold. Moreover, we can find $\Abox, \mathcal{B}, \mathcal{D}$ and $k$ such that $|\Abox|, |\mathcal{B}|, |\mathcal{D}|, k \leq 2^{O(|\q|)}$.
\end{lemma}
\begin{proof}
We only outline modifications needed to the proof of Theorem~\ref{th:fo-horn-crit} to obtain this result and the specific form of $\Abox$, $\mathcal B$ and $\mathcal D$.
Consider the automaton $\A$ defined in the proof of~Theorem~\ref{th:fo-horn-crit} and denote by $\A_{j}$, for $1 \leq j \leq k+1$, a copy of $\A$ restricted to the alphabet $\emptyset$. We construct an automaton $\A_w$ by taking a disjoint union of all the $\A_{j}$ and adding a transition $q \to_{a_j} q'$, for $1 \leq j \leq k$, from $q$ in $\A_{j}$ to $q' \in \A_{j+1}$ for each pair $(q,q')$ such that $\A$ contains an ${a_j}$-transition from the original of $q$ to the original of $q'$. The initial state of $\A_w$ is $q_{-1}$ from $\A_1$ and the final states are those in $\A_{k+1}$. It is straightforward to see that $\L_w = \L(\A_w)$. The proof of Theorem~\ref{th:fo-horn-crit} works  for $\A_w$ in place of $\A$. That $\mathcal B$ consists of $\emptyset$ only follows from the fact that non-trivial cycles in $\A_w$ can only be with $\emptyset$-symbols.
\end{proof}

We observe that (binary encoding of) $\Abox \mathcal B^k \mathcal D$ and $\Abox \mathcal B^{k+1} \mathcal D$ in the lemma above can be guessed and stored in polynomial time. Thus, it remains to show  that conditions $(i)$ and $(ii)$ from Theorem~\ref{th:fo-horn-crit}  can be checked by an \NP{}-oracle. The (more or less standard) proof of the following lemma is given in Appendix~\ref{standard}. 

\begin{lemma}\label{checkNP}
Given $a_1, \dots, a_l \in \Sigma_\Xi$ with $|a_i| = 1$, for $1 \le i \le l$, binary numbers $i_1,\dots, i_{l+1}, j$, a $\bot$-free $\LTL_\core\Xnext$-ontology $\TO$ and a positive existential temporal concept $\varkappa$, checking whether $\mathcal C_{\TO, \Abox} \models \varkappa(j)$ for $\Abox = \emptyset^{i_1}a_1 \dots \emptyset^{i_l}a_l \emptyset^{i_{l+1}}$ can be done in \NP{}.
\end{lemma}

We can now complete the proof of the upper bound. 
By Lemmas~\ref{lwstarfree} and~\ref{th:cloud}, $\L_\Xi(\q)$ is not $\FO(<)$-definable iff there exist $a_1 \dots a_k \in W_{|\varkappa|, \Xi}$, $\Abox = (\emptyset^{i_1} a_1) \dots (\emptyset^{i_{l-1}} a_{l-1}) \emptyset^{i_l}$, $\mathcal D = (\emptyset^{i_l'} a_l) (\emptyset^{i_{l+1}} a_{l+1}) \dots (\emptyset^{i_k} a_k) \emptyset^{i_{k+1}}$, $\mathcal B = \emptyset^n$, $k \geq 2$, such that $|\Abox|, |\mathcal{B}|, |\mathcal{D}|, k \leq 2^{O(|\q|)}$, $(i)$ $\neg \varkappa \in \tp_{\TO, \Abox \mathcal{B}^{k} \mathcal{D}}(|\Abox|-1) = \tp_{\TO, \Abox \mathcal{B}^{k} \mathcal{D}}(|\Abox \mathcal{B}^{k}|-1)$\textup{;} and $(ii)$ $\varkappa \in \tp_{\TO, \Abox \mathcal{B}^{k+1} \mathcal{D}}(|\Abox \mathcal{B}|-1)  = \tp_{\TO, \Abox \mathcal{B}^{k+1} \mathcal{D}}(|\Abox \mathcal{B}^{k+1}|-1)$. We can check non-$\FO(<)$-definability of $\L_\Xi(\q)$ by guessing the required $\Abox$, $\mathcal B$, $\mathcal D$, $k$ and the four types involved in the conditions $(i)$ and $(ii)$. That the four types are indeed correct can be checked in polynomial time by using the \NP{}-oracle provided by Lemma~\ref{checkNP}.

%
The proof of the matching lower bound is by reduction of $\forall\exists\text{CNF}$---the satisfiability problem for fully quantified Boolean formulas in CNF with the prefix $\forall\exists$---which is known to be $\Pi^p_2$-complete~\cite<e.g.,>{Arora&Barak09}. 
By Lemma~\ref{le:remove-bot} and Proposition~\ref{prop:specific-to-boolean++} $(ii)$, we can only consider specific OMQs.
Suppose we are given a closed QBF $\varphi = \forall X_1\dots\forall X_n\exists Y_1\dots\exists Y_m\, \psi=\forall \avec{x}\exists \avec{y} \, \psi$ with a CNF $\psi$.
Define an $\LTL_\core\Xnext$ OMPEQ $\q_\varphi(x)=(\TO_{\varphi},\varkappa_\varphi(x))$ and $\Xi$ such that $\q_\varphi(x)$ is $\FO(<)$-rewritable over $\Sigma$-ABoxes iff $\forall \avec{x}\exists \avec{y}\,\psi$ is true.
Let $\Xi$ consist of the atomic concepts $A$, $B$, $A^j_i$, for $1\le i\le m$, $0\le j\le p_i-1$, where $p_i$ is the $i$-th prime number, $X^0_k, X^1_k$, for $1\le k\le n$, $Y^0_i,Y^1_i$, for $1\le i \le m$. The ontology $\TO_{\varphi}$ has the following axioms for all such $i$ and $k$:
\begin{align*}
& A\to A_i^0, \quad \ \, A_i^j\to\Rnext A_i^{(j+1)\!\!\!\!\!\mod\!\!\!\ p_i},\quad \text{for }0\le j\le p_i-1,\\
& A_i^0\to Y_i^0, \quad A_i^{1}\to Y_i^1,\qquad
X_k^0\to\Rnext X_k^0, \quad X_k^1\to\Rnext X_k^1, \quad B\to\Rnext\Rnext B.
\end{align*}
%
%
The size of $\TO_{\varphi}$ is polynomial in $n+m$.
Let $\psi'$ be the result of replacing all $X_i$ in $\psi$ with $X_i^1$, all $\neg X_i$ with $X_i^0$, and similarly for the $Y_i$.
We set
$$
\varkappa_\varphi = A\land \bigwedge_{i=0}^{n}(X_i^0\vee X_i^1) \land (\Lnext B\vee\Rdiamond \psi').
$$
To show that $\q_\varphi(x)$ is as required, suppose $\forall \avec{x}\exists \avec{y} \, \psi$ is true. Let $(\TO_\varphi,\Abox)\models\varkappa_\varphi(t)$, for some $\Abox$ and $t$. Then $A(t)\in\Abox$ and $(\TO_\varphi,\Abox)\models\bigwedge_{i=0}^{n}(X_i^0\vee X_i^1)(t)$. So, for any $i$, there is $s\le t$ with  $X_i^0(s) \in \Abox$ or $X_i^1(s) \in \Abox$.
Let $\mathfrak a_1\colon\{X_1,\dots,X_n\}\to\{0,1\}$ be such that $(\TO_\varphi,\Abox)\models X_i^{\mathfrak a_1(X_i)}(s)$ for all $s>t$ and $i$. Take an assignment $\mathfrak a_2\colon\{Y_1,\dots,Y_m\}\to\{0,1\}$ that makes $\psi$ true. There is a number $r>0$ such that $r = \mathfrak a_2(Y_i)\ (\text{mod}\ p_i)$ for all $i$. Then $(\TO_\varphi,\Abox)\models Y_i^{\mathfrak a_2(i)}(t+r)$, $(\TO_\varphi,\Abox)\models\psi'(t+r)$, and so $(\TO_\varphi,\Abox)\models\Rdiamond \psi'(t)$.
Thus, the sentence
$$
A(x)\wedge \bigwedge_{i=0}^{n}\exists s_i\left((s_i\leqslant x)\wedge(X_i^0(s_i)\vee X_i^1(s_i))\right)
$$
is an $\FO(<)$-rewriting of $\q_\varphi(x)$ over $\Xi$-ABoxes.

If  $\forall \avec{x}\exists \avec{y} \, \psi$ is false, there is an assignment $\mathfrak a\colon\{X_1,\dots,X_n\}\to\{0,1\}$ such that $\psi$ is false under any assignment of the $Y_i$. Suppose $\L_\Xi(\q_\varphi(x))$ over $\Gamma_\Xi$ is $\FO(<)$-definable.
Let $a_{\mathfrak a}=\{A\}\cup\bigcup_{i=1}^n\{X_i^{\mathfrak a(X_i)}\} \in \Sigma_\Xi$. Consider $\Abox_l=\{B(0)\}\cup \bigcup_{Z\in a_{\mathfrak a}}\{ Z(l) \}$ for some $l>0$. Observe that, since $(\TO_\varphi,\Abox)\not\models \Rdiamond \psi'(l)$, $l$ is a certain answer to $\q_\varphi(x)$ over $\Abox_l$ iff $(\TO_\varphi,\Abox)\models B(l-1)$. It follows that $\L(\{B\}(\emptyset\emptyset)^*a'_{\mathfrak a}) = \L_\Xi(\q_\varphi(x)) \cap \L(\{B\}\emptyset^*a'_{\mathfrak a})$ (recall that $a' \in \Gamma_\Xi$ for each $a \in \Sigma_\Xi$). Clearly, $\L(\{B\}\emptyset^*a'_{\mathfrak a})$ is $\FO(<)$-definable, and so $\L(\{B\}(\emptyset\emptyset)^*a'_{\mathfrak a})$ is $\FO(<)$-definable, which is not the case \cite[Theorem~IV.2.1]{Straubing94}.
\end{proof}

\subsection{$\LTL_\core\Xnext$ OMPQs}

If we increase the expressive power of $\LTL_\core\Xnext$ OMPEQs $\q = (\TO,\varkappa)$ by allowing $\Box$-operators in $\varkappa$, the problem of deciding $\FO(<)$-rewritability becomes  \PSpace-complete, as established by Theorem~\ref{thm:ompqsforcore} below. The upper bound follows from Theorem \ref{th:lin-ompq-fo-pspace} and the next observation showing that, even though  core disjointness constraints $C_1 \land C_2 \to \bot$ may have IDB concepts $C_1$ and $C_2$, there is always an equivalent linear $\LTL_\horn\Xnext$ ontology. 


\begin{proposition}\label{th:red-core-to-lin}
For any Boolean \textup{(}specific\textup{)} $\LTL_\core\Xnext$ OMPQ and any signature $\Xi$, one can construct in polynomial time a $\Xi$-equivalent Boolean \textup{(}specific\textup{)} linear $\LTL_\horn\Xnext$ OMPQ.
\end{proposition}
\begin{proof}
We only consider Boolean OMQs $\q = (\TO, \varkappa)$ as the case of specific ones is similar.
First, for each atom $A$ in $\TO$, we introduce a fresh atom $\bar A$ and, for each axiom $C_1 \to C_2$ in $\TO$, we add to $\TO$ the axiom $\bar C_2 \to \bar C_1$, where $\bar C$ is the result of replacing $A$ in $C$ by $\bar A$; we also replace each axiom $C_1 \land C_2 \to \bot$ in $\TO$ with  $C_1 \to \bar C_2$. Then we rename each atom $A$ ($\bar A$) in $\q = (\TO, \varkappa)$ to $A'$ (respectively, $\bar A'$), for a fresh $A'$ (respectively, $\bar A'$). Denote by $C'$ the temporal concept obtained by replacing $A$ by $A'$ in $C$. Finally, we add the axioms $A \to A'$ and $A \land \bar A' \to \bot$ to $\TO$, for $A \in \Xi$, denoting the result by $\q' = (\TO', \varkappa')$. It is easy to see that $\q'$ is linear because all the IDB atoms of $\TO'$ are of the form $A'$ or $\bar A'$. For example, let $\TO = \{ \Lnext A \to B,\, \Rnext D \to C,\, C \land B \to \bot\}$ and $\Xi = \{ A, B, D\}$. Then  $\TO'$ contains the axioms $\Lnext A' \to B'$, $\bar B' \to \Lnext \bar A'$, $\Rnext D' \to C'$, $\bar C' \to \Rnext \bar D'$, $C' \to \bar B'$ together with $X \to X'$ and $X \land \bar X' \to \bot$, for each $X \in \Xi$. Clearly, $\q$ and $\q'$ are $\Xi$-equivalent. For example, over $\Abox = \{A(0), D(2)\}$, both $\q$ and $\q'$ return \yes{} as both $(\TO, \Abox)$ and $(\TO',\Abox)$ are inconsistent.
\end{proof}

\begin{theorem}\label{thm:ompqsforcore}
Deciding $\FO(<)$-rewritability of Boolean and specific $\LTL_\core\Xnext$ OMPQs is $\PSpace$-complete
\end{theorem}
\begin{proof}
The upper bound follows from Proposition~\ref{th:red-core-to-lin} and Theorem \ref{th:lin-ompq-fo-pspace}.
To prove the lower one, we reduce the \PSpace-complete problem of deciding the emptiness of the intersection of a set of DFAs~\cite{Kozen77} 
 to OMPQ rewritability. Let $\A_1,\dots,\A_n$ with $\A_i=(Q_i,\Sigma,\delta^i,q^i_0,F_i)$ be a sequence of DFAs that do not accept the empty word, have a common input alphabet, and disjoint sets $Q_i=\{q_0^i,\dots,q_{j_i}^i\}$ of states.

Let $\nabla_i$ be the set of atoms $N^i_{q, a, r}$, for $q,r\in Q_i$,  $a\in\Sigma$, such that $\fu^i_{a}(q)=r$. 
%
%
Consider the ontology $\TO$ with atomic  concepts $\{X,Y,B\}\cup\bigcup_{1 \leq i \leq n} \nabla_i$ and the following axioms, for $ 1 \leq i,l \leq n$, $q, r, s, t \in Q_i$, $q', r' \in Q_l$, $a,b \in \Sigma$:
\begin{align*}
(1) \quad & N^i_{q,a,r} \land N^l_{q', b, r'} \to\bot, &&\quad \text{ if either } a \neq b, \text{ or }i=l\text{ and }(q,r) \neq (q',r');\\
(2) \quad & N^i_{q,a,r} \land \Rnext N^i_{s,b,t} \to\bot, &&\quad \text{ if } r\neq s;\\
(3) \quad & X\land \Rnext N^i_{q,a,r} \to\bot, &&\quad \text{ if } q \neq q_0^i;\\
(4) \quad & N^i_{q,a,r} \land\Rnext Y \to\bot, &&\quad \text{ if } r\notin F_i;\\
(5) \quad & X\land\Rnext Y\to\bot;&&\\
(6) \quad & Y\to\Rnext Y;&&\\
(7) \quad & B\to\Rnext \Rnext B.&&
\end{align*}
Let
\begin{equation*}
\varkappa ~=~ \Lnext B \wedge X\wedge\Box_F\Big(\big(\bigwedge_{1 \leq i \leq n}\bigvee_{\fu^i_a(q) = r}N^i_{q,a,r}\big)\vee Y\Big).
\end{equation*}
We claim that the OMPQs $\q(x)=(\TO,\varkappa(x))$ and $\q=(\TO,\varkappa)$ are $\FO(<)$-rewritable over $\Xi$-ABoxes, for $\Xi = \sig(\q)$, iff  $\bigcap_{1 \leq i \leq n}\L(\A_i)=\emptyset$.

($\Leftarrow$) If $\bigcap_{1 \leq i \leq n}\L(\A_i)=\emptyset$, then, for any $\Xi$-ABox $\Abox$ and $k \in \tem(\Abox)$, we have $\TO,\Abox\models\varkappa(k)$ iff $\Abox$ is inconsistent with $\TO$ because the formula $X\wedge\Box_F\Big(\big(\bigwedge_{1 \leq i \leq n}\bigvee_{\fu^i_a(q) = r}N^i_{q,a,r}\big)\vee Y\Big)$ cannot be true at any place in a consistent ABox.
Let $\varphi$ be the disjunction of the formulas $\exists x\, (C(x) \land D(x))$, for all axioms $C\land D\to\bot$ of the form (1), the formulas $\exists x\, (C(x) \land D(x+1))$, for all axioms $C\land \Rnext D\to\bot$ of the forms (2) and (3), and $\exists x,y\, ((y<x+2)\land Y(y) \land C(x))$, for all axioms $C\land \Rnext Y\to\bot$ of the forms (4) and (5).
Then $(x = x) \land \varphi$ is an $\FO(<)$-rewriting of $\q(x)$, and $\varphi$ is an $\FO(<)$-rewriting of $\q$.

($\Rightarrow$) Let  $w=a_1\ldots a_k\in\bigcap_{1 \leq i \leq n}\L(\A_i)$,  $q(i,j)=\fu^i_{a_1\ldots a_j}(q^i_0)$,
$\avec{n}_j = \bigcup_i\{ N^i_{q(i,j-1),a_j,q(i,j)}\}$;
%
let $\L_1=\L(\{B\} (\emptyset)^* \{ X \} \avec{n}_1 \dots \avec{n}_k \{ Y \})$, and let $\L'_1=\L(\{B\} (\emptyset)^* \{ X' \} \avec{n}_1 \dots \avec{n}_k \{ Y \})$.
Clearly, $\L_1$ and $\L_1'$ are $\FO(<)$-definable. If $\L_\Xi(\q)$ is $\FO(<)$-definable, then so is $\L_2=\L_1\cap\L_\Xi(\q)$. However, $\L_2=\L(\{B\} (\emptyset\emptyset)^* \{ X \} \avec{n}_1 \dots \avec{n}_k \{ Y \})$ is not $\FO(<)$-definable. Similarly, $\L_2'=\L_1'\cap\L_\Xi(\q(x))$ is not $\FO(<)$-definable. So $\L_\Xi(\q)$ and $\L_\Xi(\q(x))$ are not $\FO(<)$-definable.
\end{proof}


\section{Conclusions}\label{sec:conclusion}

The problems we investigate in this article originate in the area of ontology-based access to temporal data. Classical atemporal ontology-based data access (OBDA), which over the past 15 years has become one of the most impressive applications of Description Logics and Semantic Technologies, is based on the idea of rewriting ontology-mediated queries (OMQs) into query languages supported by conventional database management systems (DBMSs). For relational data, standard target languages for rewritings are SQL---that is, essentially $\FO$-formulas---and datalog, which allows recursive queries over data.
The idea of rewriting has led to numerous and profound results that either uniformly classify OMQs according to their $\FO$- and datalog-rewritability or establish the computational complexity of recognising $\FO$- and datalog-rewritability of OMQs in expressive languages and design practical decision and rewriting algorithms. In classical database theory, $\FO$- and linear-datalog-rewritability of datalog queries has been an active research area since the 1980s.

Unfortunately, those results and developed techniques are not applicable to OMQs over temporal data, where the timestamps are linearly ordered by the precedence relation $<$ and OMQs may contain temporal constructs. First, as well known, the interaction between temporal and description logic operators tends to dramatically increase the complexity of OMQ answering, which makes the uniform classification of OMQs according to their rewritability type much harder. Some initial steps in this direction have been made by~\citeA{DBLP:conf/ruleml/BorgwardtFK19,DBLP:journals/corr/abs-2111-06806}.
Second, even without the description logic constructs, pure one-dimensional temporal OMQs give rise to the complexity classes and target languages for $\FO$-rewritings that have not occurred in the OBDA context so far. For instance, any $\LTL_{\bool}\Xallop$ OMQ is rewritable into $\FO(<,\RPR)$---a class not appearing in the classical (atemporal) OBDA literature---that essentially requires recursion, which is weaker than linear datalog recursion but still not expressible in SQL.

In this article, our concern is determining the optimal rewritability type for OMQs given in linear temporal logic \LTL. In fact, we argue in the introduction that such OMQs provide an adequate formalism for querying sensor log data from various parts of complex equipment where there is no relevant interaction between those parts, and the results of measurements are qualitatively graded as, e.g., high, medium, low, etc. Our starting point is establishing a close connection between rewritability of \LTL{} OMQs and definability of regular languages by means of $\FO(<)$-formulas possibly containing  extra predicates and constructs. The computational complexity and definability of regular languages have been investigated since the late 1980s. The relevant $\FO$-languages identified are $\FO(<)$, $\FO(<,\equiv)$, $\FO(<,\MOD)$ and $\FO(\RPR)$, the first two of which are in $\ACz$ for data complexity, the third is in $\ACC$ and the last one in $\NCo$. In practice, $\FO(<,\MOD)$-rewritable OMQs can be implemented in SQL using the \texttt{count} operator, while $\FO(<,\equiv)$-rewritable ones do not need it. It is also known that recognising $\FO(<)$-definability of regular languages given by a DFA is \PSpace-complete; recognising definability by $\FO(<,\equiv)$ and $\FO(<,\MOD)$ formulas is known to be decidable, but the exact complexity has so far remained open.

The main technical results we obtain here are threefold. First, we settle the open problems just  mentioned by proving that deciding $\FO(<)$-, $\FO(<,\equiv)$- and $\FO(<,\MOD)$-definability of regular languages given by a DFA, NFA or 2NFA is \PSpace-complete. Second, we show that deciding  $\FO(<)$-, $\FO(<,\equiv)$- and $\FO(<,\MOD)$-rewritability of \LTL{} OMQs is \ExpSpace-complete. And finally, we identify a number of natural and practically important OMQ classes for which these problems are \PSpace-, $\Pi^p_2$- or \coNP{}-complete; these results could lead to feasible  algorithms to be used in temporal OBDA systems.
%

While this article makes steps towards the non-uniform approach to temporal OBDA, many interesting and challenging problems remain open. We discuss some of them below.

\smallskip
\noindent
\textbf{1.} Our results on linear Horn, core and Krom OMQs are only established for ontologies with $\Rnext$ and $\Lnext$. Some of the techniques used in the proofs do not go through in the presence of $\Rbox$ and $\Lbox$, and so it would be interesting to see if the same complexity results hold for the fragments with all of these operators. One could also consider adding the operators `since' and `until' to ontologies and/or queries in \LTL{} OMQs. General results, such as Theorem~\ref{hornExpSpacehard}, will not be affected by this, but it is an open question for the fragments mentioned above. Finally, we could not establish the complexity of deciding $\FO(<,\MOD)$-rewritability of linear $\LTL_{\horn}\Xnext$ OMPQs. It is likely to be \PSpace{}, but we did not manage to prove an appropriate criterion in the spirit of Theorems~\ref{th:fo-horn-crit} and~\ref{th:foe-horn-crit}.

\smallskip
\noindent
\textbf{2.}
In this article, we consider queries with at most one answer variable. More expressive query languages based on monadic first-order logic $\MFO(<)$ and allowing multiple answer variables have been suggested by~\citeA{DBLP:journals/ai/ArtaleKKRWZ21}. It would be interesting to understand the impact of replacing \LTL{} queries with $\MFO(<)$ queries in \LTL{} OMQs on their $\FO$-rewritability properties.

\smallskip
\noindent
\textbf{3.}
Another prominent temporal KR formalism that has great potential as an ontology and query language for temporal OBDA is metric temporal logic \textsl{MTL}, which was originally introduced for modelling and reasoning about real-time systems~\cite{DBLP:journals/rts/Koymans90,DBLP:journals/iandc/AlurH93}. Each operator in \textsl{MTL} is indexed by a temporal interval over which the operator works: for example, $\Diamond_{(0, 1.5]} A$ is true at $t$ iff $A$ holds at some $t'$ with $0 < t'- t \leq 1.5$. The interpretation domain is dense $\mathbb R$ or $\mathbb Q$ under the continuous semantics and the active domain of the data instance under the pointwise semantics~\cite{DBLP:conf/formats/OuaknineW08}.  \textsl{MTL} is more expressive and succinct than \LTL{} and is also suitable in scenarios where sensors report their measurements asynchronously. In the context of OBDA, \textsl{MTL} has recently been investigated by~\citeA{DBLP:journals/jair/BrandtKRXZ18,DBLP:conf/ijcai/RyzhikovWZ19,DBLP:conf/ijcai/WalegaGKK20,DBLP:conf/aaai/CucalaWGK21,DBLP:journals/corr/abs-2201-04596}. Target rewriting languages for \textsl{MTL} OMQs include $\FO(\text{DTC})$, $\FO(\text{TC})$ with (deterministic) transitive closure, and datalog($\FO$), which correspond to the complexity classes L, NL and P, respectively. At present, the problem of recognising the data complexity and optimal rewritability type of \textsl{MTL} OMQs is wide open.

\smallskip
\noindent
\textbf{4.}
In OBDA practice, we are concerned not only with the fact of $\FO$-rewritability of a given OMQ but also with the size and shape of the rewriting to be executed by a DBMS~\cite<e.g.,>{DBLP:journals/jacm/BienvenuKKPZ18,DBLP:conf/pods/BienvenuKKPRZ17}. The experiments with a few real-world use cases reported by~\citeA{DBLP:journals/jair/BrandtKRXZ18,DBLP:conf/time/0001CKKMR0Z19} indicate that temporal OMQs with a non-recursive ontology are scalable and efficient. But we are not aware of any theoretical results on the succinctness of $\FO$-rewritings for temporal OMQs.

\smallskip
\noindent
\textbf{5.}
Extending the results obtained above for 1D \LTL{} OMQs to various 2D combinations of \LTL{} with description logics (such as \DL{}, $\mathcal{EL}$ or $\mathcal{ALC}$), Schema.org or datalog could be especially challenging due to the interaction between the temporal and domain dimensions. In the Horn case, one might try to use a variant of the automata-theoretic approach developed by~\citeA{DBLP:conf/ijcai/LutzS17,DBLP:journals/corr/abs-1904-12533}.

\smallskip
\noindent
\textbf{6.}
Finally, from the application point of view, it is important to identify real-world use-cases for temporal OBDA, relevant classes of OMQ, and then develop OMQ rewriting and optimisation algorithms for those classes. Some work in this direction has recently been done for both \textsl{MTL} and \LTL~\cite{DBLP:journals/corr/abs-2201-04596,DBLP:journals/jair/BrandtKRXZ18,DBLP:conf/dlog/TahratBAGO20}. Although the results of this paper suggest  algorithms that can identify the best rewritability class (and so the most efficient database query language) for a given OMQ, implementing and optimising such algorithms is a serious challenge. Furthermore, the algorithms mentioned above need to be incorporated into a user-friendly OBDA system such as Ontop~\cite{DBLP:conf/semweb/Rodriguez-MuroKZ13,DBLP:conf/semweb/XiaoLKKKDCCCB20}.

%

\paragraph{Acknowledgements.}
This work was supported by the EPSRC U.K.\ grant EP/S032282 for the project  `\textsl{quant$^\textsl{MD}$}: Ontology-Based Management for Many-Dimensional Quantitative Data'\!.
We are grateful to the referees  of  this  article  for  their careful reading, valuable comments and suggestions.


\appendix



\newpage

\section{}

\subsection{Proof of Theorem~\ref{Prop:rewr-def}~$(ii)$}\label{app:thm5}

\noindent
\textbf{Theorem~\ref{Prop:rewr-def}~$(ii)$.}  
\emph{Let $\q = (\TO, \varkappa)$ be a Boolean and $\q(x) = (\TO, \varkappa(x))$ a specific OMQ. Then, for any $\lang \in \{ \FO(<), \FO(<,\equiv), \FO(<,\MOD)\}$ and $\Xi\subseteq\sig(\q)$, the OMQ $\q$ is $\lang$-rewritable over $\Xi$-ABoxes iff $\L_\Xi(\q)$ is $\lang$-definable\textup{;} similarly, $\q(x)$ is $\lang$-rewritable over $\Xi$-ABoxes iff $\L_\Xi(\q(x))$ is $\lang$-definable.}
\begin{proof}
For any $A\in\Xi$ and any $a\in\Sigma_\Xi$,  we set
$$
\chi_A(y) = \bigvee_{A \in a\in\Sigma_\Xi} a(y), \qquad \chi_a(y)=\bigwedge_{A \in a}A(y)\land\bigwedge_{A \notin a}\neg A(y),
$$
where $a(y)$ is a unary predicate associated with each $a\in\Sigma_\Xi$.
For any $\Xi$-ABox $\Abox$ and any $n \in \tem(\Abox)$, we have $\SA \models A(n)$ iff $\mathfrak S_{w_\Abox} \models \chi_A(n)$, and $\mathfrak S_{w_\Abox} \models a(n)$ iff $\SA \models \chi_a(n)$.
Thus, we obtain an $\lang$-sentence defining $\L_\Xi(\q)$ by taking an $\lang$-rewriting of $\q$ and replacing all atoms $A(y)$ in it with $\chi_A(y)$. Conversely, we obtain an $\lang$-rewriting of $\q$ by taking an $\lang$-sentence defining $\L_\Xi(\q)$ and replacing all $a(y)$ in it with $\chi_a(y)$.

Consider next $\q(x)$. Let $\varphi(x)$ be an $\lang$-rewriting of $\q(x)$ and let $\varphi'(x)$ be the result of replacing atoms $A(y)$ in $\varphi(x)$ with $\chi'_A(y) =\bigvee_{A \in a\in\Gamma_\Xi} a(y)$. Given an ABox $\Abox$ and $i \in \tem(\Abox)$,
we have $\SA \models \varphi(i)$ iff $\mathfrak S_{w_\Abox,i}\models\varphi'(i)$. 
A word $w=a_0\ldots a_n\in\Gamma_\Xi^*$ is in $\L_\Xi(\q(x))$ iff (a) there is $i$ such that $a_i\in\Sigma'_\Xi$, (b) $a_j \in \Sigma_\Xi$ for all $j\neq i$, and (c) $\mathfrak S_w\models\varphi'(i)$. Therefore, for the sentence
$$
\varphi'' = \exists x\, \Big(\varphi'(x) \land \forall y \, \big[ \big ((y = x) \to \bigvee_{a' \in \Sigma'_{\Xi}} a'(y)\big) \land \big((y \neq x) \to \bigvee_{a \in \Sigma_{\Xi}} a(y) \big)\big] \Big)
$$
and a word $w\in\Gamma_\Xi^*$, we have $\mathfrak S_w\models\varphi''$ iff $w=w_{\Abox,i}$ for some $\Abox$ and $i$ such that $\SA\models\varphi(i)$. It follows that $\varphi''$ defines $\L_\Xi(\q(x))$.

Now, let $\psi$ be an $\lang$-sentence defining $\L_\Xi(\q(x))$ and let $\psi'(x)$ be the result of replacing atoms $a(y)$ in $\varphi$, for $a\in\Sigma_\Xi$,  with $a(y)\land(x\neq y)$ and atoms $a'(y)$, for $a'\in\Sigma'_\Xi$, with $a(y)\land(x=y)$. For $w=a_0\ldots a_n\in\Sigma_\Xi^*$, we have $\mathfrak S_w\models\psi'(i)$ iff $\mathfrak S_{w_i}\models\psi$, where $w_i$ is $w$ with $a_i$ replaced by $a_i'$. Let $\psi''(x)$ be the result of replacing $a(y)$ in $\psi'(x)$ with $\chi_a(y)$. Then, for any $\Abox$ and $i\in\tem(\Abox)$, we have $\mathfrak S_\Abox\models\psi''(i)$ iff $\mathfrak S_{w_\Abox}\models \psi'(i)$ iff $\mathfrak S_{w_{\Abox,i}}\models\psi$, and so $\psi''(x)$ is a rewriting of $\q$.
\end{proof}


\subsection{Proof of Lemma~\ref{expandL}}\label{app:lem7}

\noindent
\textbf{Lemma~\ref{expandL}.}
\emph{Suppose $\lang \in \{ \FO(<), \FO(<,\equiv), \FO(<,\MOD)\}$ and $\Sigma$, $\Gamma$ and $\Delta$ are alphabets such that $\Sigma\cup\{x,y\}\subseteq\Gamma\subseteq\Delta$, for some $x,y \notin \Sigma$. Then a regular language $\L$ over $\Sigma$ is $\lang$-definable iff the regular language 
\[
\L'=\{w_1xwyw_2\mid w\in\L,\ w_1,w_2\in\Gamma^\ast\}
\]
%
 is $\lang$-definable over $\Delta$.
}
\begin{proof}
 Let $\L = \L(\A)$, for a minimal DFA $\A=(Q,\Sigma,\delta, q_0,F)$. Let $tr$ be the trash state\footnote{A \emph{trash state} is a state from which no accepting state is reachable. A  minimal DFA can have at most one trash state.} in $\A$ if any. Given alphabets $\Gamma,\Delta$, consider the DFA 
 \[
 \A'=\bigl(Q\cup\{tr,q_0',f\},\Delta,\delta', q'_0,\{f\}\bigr),
 \]
  where $\delta'$ consists of the following transitions:
 $(q,a,p)$ for $(q,a,p)\in \delta$, $p\ne tr$,
$(q,a,q_0')$ for $(q,a,tr)\in \delta$,
 $(q_0',a,q_0')$ for $a\in\Gamma\setminus\{ x\}$,
 $(q,x,q_0)$ for $q\in Q\cup\{q_0'\}$,
 $(q,a, q_0')$ for $q\in Q$ and $a\in\Gamma\setminus(\Sigma\cup\{x,y\})$,
 $(q,y, q_0')$ for $q\in Q\setminus F$,
 $(q,y,f)$ for $q\in F$,
 $(f,a, f)$ for $a\in\Gamma$,
$(q,a,tr)$ for  $q$ and $a\in\Delta\setminus\Gamma$,
 $(tr,a,tr)$ for $a\in\Delta$.
The DFA $\A'$ is illustrated in the picture below, where a transition labelled by a set stands for the corresponding transitions for each element of that set, the transitions starting from the frame around $\A$ represent the corresponding  transitions from every state in $\A$, and the transitions from states in $\A$ to $tr$ (shown as the dashed arrow in the picture) are redirected to $q'_0$.
It is readily checked that $\L(\A') = \{w_1xwyw_2\mid w\in\L,\ w_1,w_2\in\Gamma^\ast\}$.
\\
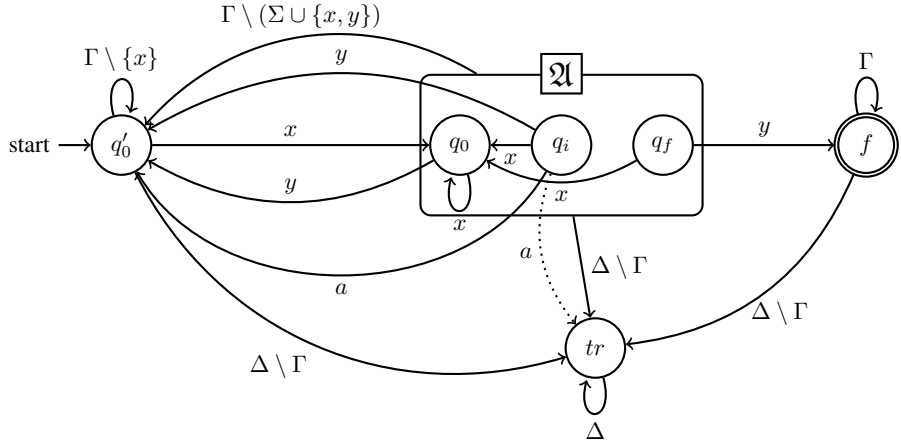
\centerline{
\begin{tikzpicture}[->,thick,scale=0.9,node distance=2cm, transform shape]
\node[state,initial](q0p){$q_0'$};
\node at (5,0)[state] (q0) {$q_0$};
\node at (6.5,0)[state](qi){$q_i$};
\node at (8,0)[state] (qf) {$q_{f}$};
\node at (11,0)[state,accepting](f){$f$};
\node at (7,-3)[state](tr){$tr$};
\node[fit = (q0)(qf), basic box = black, header = $\A$] (A) {};
\draw
(q0p) edge [loop above,above] node {$\Gamma\setminus\{x\}$} (q0p)
(q0p) edge [above] node{$x$} (q0)
(q0p) edge [bend right=40,below left] node{$\Delta\setminus\Gamma$} (tr)
(q0) edge [loop below, below] node{$x$} (q0)
(q0) edge [bend left,above] node {$y$} (q0p)
(qi) edge [below] node {$x$} (q0)
(qi) edge [bend right, above] node {$y$} (q0p)
(qi) edge [bend right,dotted, left] node {$a$} (tr)
(qi) edge [bend left=60, below] node {$a$} (q0p)
(A) edge [bend right=40,above] node {$\Gamma\setminus(\Sigma\cup\{x,y\})$} (q0p)
(A) edge [right] node {$\Delta\setminus\Gamma$} (tr)
(qf) edge [above] node{$y$} (f)
(qf) edge [bend left, below] node{$x$} (q0)
(f) edge [loop above,above] node{$\Gamma$} (f)
(f) edge [bend left,below] node{$\quad\ \Delta\setminus\Gamma$} (tr)
(tr) edge [loop below,below] node {$\Delta$} (tr)
;
\end{tikzpicture}
}
\\
We now show that $\L(\A)$ is $\lang$-definable iff the language $\L(\A')$ is $\lang$-definable. As the argument is effectively the same for all $\lang$, we only show it in one case.

($\Leftarrow$) If $\L(\A)$ is not $\FO(<)$-definable, then, by Theorem \ref{DFAcrit} $(i)$, there exist a state $q$, a number $k$, and a word $u\in\Sigma^*$ such that $q\not\simm \fu_u(q)$ and $q= \fu_{u^k}(q)$. One can readily check that the same $q$, $k$ and $u$ satisfy the same condition in $\A'$, and so $\L(\A')$.

 ($\Rightarrow$) If $\L(\A')$ is not $\FO(<)$-definable, then, by Theorem \ref{DFAcrit} $(i)$, there exist a state $q$, a number $k$, and a word $u\in\Delta^*$ such that $q\not\simm \fu_u(q)$ and $q= \fu_{u^k}(q)$. There are no transitions leaving $tr$ and the only transition leaving $f$ is to $tr$.  It follows that, when reading $u^k$ starting from $q$, $\A'$ can visit $f$ or $tr$. Suppose it visits $q_0'$. As the only way of leaving $q_0'$ not to $tr$ is via $x$, the word $u$ contains $x$. Let $u=u_1xu_2$. But then, for any $p\notin \{f, tr\}$, we have $\fu_u(p)=\fu_{u_2}(q_0)$, and so all $\fu_{u^i}(q)$ are the same, which is a contradiction. Thus, $\A'$ does not visit $q_0'$. It follows that $\fu_{u^i}(q)\in Q$ and $u\in\Sigma^*$. Then the same $q$, $k$, and $u$ satisfy the conditions of Theorem \ref{DFAcrit} $(i)$ for $\A$, and so $\L(\A)$ is not $\FO(<)$-definable.
\end{proof}


\subsection{Additional Axioms and Counters for the Proof of Theorem \ref{hornExpSpacehard}}\label{sec:LTL-genearal-app}

Below are the axioms describing the transitions of the automata $\A_i$.  
For $\A_0$, we use the axioms
\begin{align*}
&[\mathbb A=0] \wedge T\land\sharp\to [(\Rnext \mathbb A)=0] \wedge \Rnext{ Q}\wedge [(\Rnext\mathbb L)=0], \\
&[\mathbb A=0] \land Q\land [\mathbb L={0}]  \land (q_1,x_1) \to [(\Rnext \mathbb A)=0] \land \Rnext{Q}\land [(\Rnext\mathbb L)=1],\\
&\ldots\\
&[\mathbb A=0] \land Q\land [\mathbb L=n-1]  \land x_n\to [(\Rnext \mathbb A)=0] \land \Rnext{Q}\land [(\Rnext\mathbb L)=n],\\
&[\mathbb A=0] \land Q\land [\mathbb L > n-1]\land [\mathbb L{< N}]  \land \B\to [(\Rnext \mathbb A)=0] \land \Rnext{Q}\land [(\Rnext\mathbb L)=\mathbb L+1],\\
&[\mathbb A=0] \land  Q\land[\mathbb L=N]\land \sharp\to [(\Rnext \mathbb A)=0] \land \Rnext{P}\land [(\Rnext\mathbb L)=0], \\
&[\mathbb A=0] \land  P\land[\mathbb L=0]\land a\to [(\Rnext \mathbb A)=0] \land \Rnext{P_{\sharp\sharp}}, \quad \text{for $a\neq(\qa,\B),\sharp,\flat$}, \\
&[\mathbb A=0] \land  P\land[\mathbb L=0]\land (\qa,\B)\to [(\Rnext \mathbb A)=0] \land \Rnext{P}\land [(\Rnext\mathbb L)=1], \\
&[\mathbb A=0] \land  P_{\sharp\sharp}\land a\to [(\Rnext \mathbb A)=0] \land \Rnext{P_{\sharp\sharp}}, \quad \text{for $a\neq\sharp$}, \\
&[\mathbb A=0] \land  P_{\sharp\sharp}\land \sharp \to [(\Rnext \mathbb A)=0] \land \Rnext{P}\land [(\Rnext\mathbb L)=0], \\
&[\mathbb A=0] \land  P\land[\mathbb L>0]\land[\mathbb L<N]\land \B\to [(\Rnext \mathbb A)=0] \land \Rnext{P}\land [(\Rnext\mathbb L)=\mathbb L+1)], \\
&[\mathbb A=0] \land  P\land[\mathbb L=N]\land \flat\to [\mathbb A=0]\land \Rnext F.
\end{align*}
For $\A_i$ with $0<i\leq N$ and $a,b,c\in\Sigma'\setminus\{\sharp,\flat\}$, we need the axioms
\begin{align*}
&[\mathbb A=1]\land[\mathbb A <N+1]\land T\land \sharp\to [(\Rnext\mathbb A)=\mathbb A]\land \Rnext{ R_\sharp}, \\
&[\mathbb A>1]\land[\mathbb A <N+1]\land T\land \sharp\to [(\Rnext\mathbb A)=\mathbb A]\land \Rnext{ Q}\land [(\Rnext\mathbb L)=\mathbb A-1], \\
&[\mathbb A>1]\land[\mathbb A <N+1]\land Q\land[\mathbb L>1]\land a\to [(\Rnext\mathbb A)=\mathbb A]\land \Rnext{ Q}\land [(\Rnext\mathbb L)=\mathbb L-1], \\
&[\mathbb A>0]\land[\mathbb A <N+1]\land Q\land[\mathbb L={1}]\land a \to [(\Rnext\mathbb A)=\mathbb A]\land \Rnext{ R_a},\\
&[\mathbb A>0]\land[\mathbb A<N+1]\land R_a\land b\to [(\Rnext\mathbb A)=\mathbb A]\land \Rnext{ R_{ab}},\\
&[\mathbb A>0]\land[\mathbb A<N]\land R_{ab}\land c \to [(\Rnext\mathbb A)=\mathbb A]\land \Rnext{ Q_{\gamma(a,b,c)}}\land\Rnext  [\mathbb L=\mathbb A+1], \\
&[\mathbb A=N]\land R_{ab}\land \sharp \to [(\Rnext\mathbb A)=\mathbb A]\land \Rnext{ P_{\gamma(a,b,\sharp)}}\land\Rnext  [\mathbb L=N-1], \\
&[\mathbb A >0]\land[\mathbb A<N+1]\land Q_{a}\land[\mathbb L<N]\land b\to [(\Rnext\mathbb A)=\mathbb A]\land \Rnext{ Q_{a}}\land [(\Rnext\mathbb L)=\mathbb L+1],\\
&[\mathbb A =1]\land Q_{a}\land[\mathbb L=N]\land \sharp\to [(\Rnext\mathbb A)=\mathbb A]\land \Rnext{ P_{\sharp a}}\\
&[\mathbb A >1]\land[\mathbb A<N+1]\land Q_{a}\land[\mathbb L=N]\land \sharp\to [(\Rnext\mathbb A)=\mathbb A]\land \Rnext{ P_{a}}\land [(\Rnext\mathbb L)=\mathbb A-1]\\
&[\mathbb A >1]\land[\mathbb A<N+1]\land P_{a}\land[\mathbb L>1]\land b\to [(\Rnext\mathbb A)=\mathbb A]\land \Rnext{ P_{a}}\land [(\Rnext\mathbb L)=\mathbb L-1],\\
&[\mathbb A>0]\land[\mathbb A<N+1]\land P_{a}\land[\mathbb L={1}]\land b\to [(\Rnext\mathbb A)=\mathbb A]\land \Rnext{ P_{ba}} \\
&[\mathbb A>0]\land[\mathbb A<N+1]\land P_{ab}\land a\to [(\Rnext\mathbb A)=\mathbb A]\land \Rnext{ R_{ab}},\\
&[\mathbb A>0]\land[\mathbb A<N]\land Q_b\land[\mathbb L=N]\land \flat\to [(\Rnext\mathbb A)=\mathbb A]\land \Rnext F, \\
&[\mathbb A=N]\land R_{ab}\land \flat\to [(\Rnext\mathbb A)=\mathbb A]\land \Rnext F.
\end{align*}
To calculate the value of $j$ in the construction of $\TO_{\MOD}$, we use the following counters, formulas, and axioms.

For two counters $\mathbb X$ and $\mathbb Y$, set
$$
[\mathbb X=\mathbb Y/2] ~=~ X_k^0\land\bigwedge_{l=2}^k\left((Y_l^0\to X_{l-1}^0)\land(Y_l^1\to X_{l-1}^1)\right).
$$
We have $\I,n \models [\mathbb X=\mathbb Y/2]$ iff the values $x$ of $\mathbb X$ and $y$ of $\mathbb Y$ at $n$ in $\I$ satisfy $x=\lfloor y/2\rfloor$.
We define three new counters $\mathbb C^=_{\mathbb X\mathbb Y}$, $\mathbb C^-_{\mathbb X\mathbb Y}$, and $\mathbb C^+_{\mathbb X\mathbb Y}$, which come with the following axioms, for all $\iota_1,\iota_2,\iota_3\in\{0,1\}$, that should be added to the ontology:
\begin{align*}
&X_i^{\iota_1}\land Y_i^{\iota_2}\to (C^=_{\mathbb X\mathbb Y})_i^{(\iota_1+\iota_2+1)\!\!\!\!\mod 2},& \text{ for all $i\in[1,k]$,}\\
&X_1^{\iota_1}\land Y_1^{\iota_2}\to (C^+_{\mathbb X\mathbb Y})_1^{0}, &\\
&X_{i-1}^{\iota_1}\land Y_{i-1}^{\iota_2}\land(C^+_{\mathbb X\mathbb Y})_{i-1}^{\iota_3}\to (C^+_{\mathbb X\mathbb Y})_i^{(\iota_1\iota_2+\iota_1\iota_3+\iota_2\iota_3)\!\!\!\!\mod 2}, &\text{ for all $i\in[2,k]$,}\\
&X_1^{\iota_1}\land Y_1^{\iota_2}\to (C^-_{\mathbb X\mathbb Y})_1^{0}, &\\
&X_{i-1}^{\iota_1}\land Y_{i-1}^{\iota_2}\land(C^-_{\mathbb X\mathbb Y})_{i-1}^{\iota_3}\to (C^-_{\mathbb X\mathbb Y})_i^{(\iota_1\iota_2+\iota_1\iota_3+\iota_2\iota_3+\iota_2+\iota_3)\!\!\!\!\mod 2}, &\text{ for all $i\in[2,k]$.}
\end{align*}
Define the following formulas, where $\mathbb W,\mathbb X,\mathbb Y$ are some counters:
\begin{align*}
[\mathbb X>\mathbb Y]&=\bigvee_{i=1}^k \big(X_i^1\land Y_i^0\land\bigwedge_{j=i+1}^k(C^=_{\mathbb X\mathbb Y})_i^{1}\big),\\
[\mathbb X\ge\mathbb Y]&=[\mathbb X>\mathbb Y]\lor\bigwedge_{i=1}^k(C^=_{\mathbb X\mathbb Y})_i^{1},\\
[\mathbb W=\mathbb X+\mathbb Y]&=\bigwedge_{i=1}^k \bigwedge_{\iota_{1,2,3}\in\{0,1\}}\big(X_i^{\iota_1}\land Y_i^{\iota_2}\land(C^+_{\mathbb X\mathbb Y})_{i}^{\iota_3}\to W_i^{\iota_1+\iota_2+\iota_3\!\!\!\!\mod 2}\big),\\
[\mathbb W=\mathbb X-\mathbb Y]&=\bigwedge_{i=1}^k \bigwedge_{\iota_{1,2,3}\in\{0,1\}}\big(X_i^{\iota_1}\land Y_i^{\iota_2}\land(C^-_{\mathbb X\mathbb Y})_{i}^{\iota_3}\to W_i^{\iota_1+\iota_2+\iota_3\!\!\!\!\mod 2}\big).
\end{align*}
We have $\I,n \models [\mathbb X>\mathbb Y]$, $\I,n \models [\mathbb X\ge\mathbb Y]$, $\I,n \models [\mathbb W=\mathbb X+\mathbb Y]$, or $\I,n \models [\mathbb W=\mathbb X-\mathbb Y]$  iff the values $x$ of $\mathbb X$, $y$ of $\mathbb Y$, and $w$ of $\mathbb W$ at $n$ in $\I$ satisfy, respectively, the following conditions: $x>y$, $x\ge y$, $w=x+y$ for $x+y<2^k$, and $w=x-y$ for $x\ge y$.

In our ontology $\TO_\MOD$, we use counters $\mathbb U_l$, $\mathbb V_l$, $\mathbb R_l$, $\mathbb R^{+}_l$, $\mathbb R^-_l$, $\mathbb S_l$, $\mathbb S^-_l$, $\mathbb S^{+}_l$, $\mathbb D_l$, $\mathbb G_l$, $\mathbb H_l$, for $l\in[0,\dots,2k]$, along with some auxiliary counters $\mathbb C_{\mathbb X\mathbb Y}$. Intuitively, the counters with the index $l$ hold the values of the corresponding expressions after the $l$-th step of the algorithm according to the table below:\\[5pt]
\centerline{
\begin{tabular}{l|l} \toprule
$\mathbb U_l,\mathbb V_l,\mathbb R_l,\mathbb S_l$&$u,v,r,s$\\
$\mathbb R^{+}_l,\mathbb S^{+}_l$&$r+\pp,s+\pp$\\
$\mathbb R^{-}_l,\mathbb S^{-}_l$&$-r\!\!\!\mod \pp,-s\!\!\!\mod \pp$\\
$\mathbb D_l$&$|u-v|$\\
$\mathbb G_l$& the even number from the pair $((r-s)\!\!\!\mod \pp)$, $((r-s)\!\!\!\mod \pp)+\pp$\\
$\mathbb H_l$& the even number from the pair $((s-r)\!\!\!\mod \pp)$, $((s-r)\!\!\!\mod \pp)+\pp$\\
\bottomrule
\end{tabular}}
\\[5pt]
We add the following axioms (simulating the algorithm) to the ontology $\TO_{\MOD}$:
\begin{align*}
&[\mathbb A >0]\land[\mathbb A <\pp]\land S\land \natural \to [\mathbb U_0=\pp] \land [\mathbb V_0=\mathbb A]\land[\mathbb R_0=0]\land[\mathbb S_0=1],\\
&[\mathbb U_l>\mathbb V_l]\to[\mathbb D_l=\mathbb U_l-\mathbb V_l],\\
&[\mathbb V_l\ge\mathbb U_l]\to[\mathbb D_l=\mathbb V_l-\mathbb U_l],\\
&[\mathbb R^{+}_l=\mathbb R_l+\mathbb U_0]\land[\mathbb R^{-}_l=\mathbb U_0-\mathbb R_l]\land[\mathbb S^{+}_l=\mathbb S_l+\mathbb U_0]\land[\mathbb S^{-}_l=\mathbb U_0-\mathbb S_l],
\end{align*}
\begin{align*}
&[\mathbb R_l\ge\mathbb S_l]\land(((R_l)_1^0\land(S_l)_1^0)\vee((R_l)_1^1\land(S_l)_1^1))\to [\mathbb G_l=\mathbb R_l-\mathbb S_l]\land[\mathbb H_l=\mathbb S^+_l+\mathbb R^-_l],\\
&[\mathbb R_l\ge\mathbb S_l]\land(((R_l)_1^1\land(S_l)_1^0)\vee((R_l)_1^0\land(S_l)_1^1))\to [\mathbb G_l=\mathbb R_l+\mathbb S^-_l]\land[\mathbb H_l=\mathbb S^+_l-\mathbb R_l],\\
&[\mathbb S_l>\mathbb R_l]\land(((R_l)_1^0\land(S_l)_1^0)\vee((R_l)_1^1\land(S_l)_1^1))\to [\mathbb G_l=\mathbb R^+_l+\mathbb S^-_l]\land[\mathbb H_l=\mathbb S_l-\mathbb R_l],\\
&[\mathbb S_l>\mathbb R_l]\land(((R_l)_1^1\land(S_l)_1^0)\vee((R_l)_1^0\land(S_l)_1^1))\to [\mathbb G_l=\mathbb R^+_l-\mathbb S_l]\land[\mathbb H_l=\mathbb S_l+\mathbb R^-_l],\\
%
&[\mathbb V_l>0]\land(V_l)_1^0\land(S_l)_1^0\to [\mathbb U_{l+1}=\mathbb U_l]\land[\mathbb V_{l+1}=\mathbb V_l/ 2]\land[\mathbb R_{l+1}=\mathbb R_l]\land[\mathbb S_{l+1}=\mathbb S_l/2],\\
&[\mathbb V_l>0]\land(V_l)_1^0\land(S_l)_1^1\to [\mathbb U_{l+1}=\mathbb U_l]\land[\mathbb V_{l+1}=\mathbb V_l/2]\land[\mathbb R_{l+1}=\mathbb R_l]\land[\mathbb S_{l+1}=\mathbb S^+_l/2],\\
&(V_l)_1^1\land(U_l)_1^0\land(R_l)_1^0\to [\mathbb U_{l+1}=\mathbb U_l/2]\land [\mathbb V_{l+1}=\mathbb V_l]\land[\mathbb R_{l+1}=\mathbb R_l/2]\land[\mathbb S_{l+1}=\mathbb S_l],\\
&(V_l)_1^1\land(U_l)_1^0\land(R_l)_1^1\to [\mathbb U_{l+1}=\mathbb U_l/2]\land [\mathbb V_{l+1}=\mathbb V_l]\land[\mathbb R_{l+1}=\mathbb R^+_l/2]\land[\mathbb S_{l+1}=\mathbb S_l],\\
&(V_l)_1^1\land(U_l)_1^1\land[\mathbb U_l> \mathbb V_l]\to [\mathbb U_{l+1}=\mathbb D_l/ 2]\land[\mathbb V_{l+1}=\mathbb V_l]\land[\mathbb R_{l+1}=\mathbb H_l/2]\land[\mathbb S_{l+1}=\mathbb S_l],\\
&(V_l)_1^1\land(U_l)_1^1\land[\mathbb V_l\ge \mathbb U_l]\to [\mathbb U_{l+1}=\mathbb U_l]\land[\mathbb V_{l+1}=\mathbb D_l/ 2]\land[\mathbb R_{l+1}=\mathbb R_l]\land[\mathbb S_{l+1}=\mathbb G_j/2],\\
&[\mathbb V_l=0]\to[\mathbb J=\mathbb R^-_l].
\end{align*}

\subsection{Proof of Lemma~\ref{th:derivations-runs}}\label{lem:23}

\noindent
\textbf{Lemma~\ref{th:derivations-runs}.}
\emph{Let $\Abox \in \Sigma_\Xi^*$ be of the form $\emptyset^N \mathcal B \emptyset^N$. 
Then
%
%
$A \in \tp_{\TO, \mathcal A}^{\sig(\TO)}(\ell)$ iff there exists a run $(q_0, 0), \dots, (q, \ell), (q_A, i)$ of $\A_\TO^\Xi$ on $\Abox$, for all $\ell$ with $N \leq \ell < |\mathcal A| - N$.
}
\begin{proof}
We call a sequence $\mathfrak D$ of the form
\begin{multline}
(C_1^0 \land \dots \land C_{k_0}^0 \to A_1, n_1), (C_1^1 \land \dots \land C_{k_1}^1 \land \nxt^{i_1} A_1 \to A_2, n_2), \dots,\\
\label{eq:deriv} (C_1^m \land \dots \land C_{k_m}^m \land \nxt^{i_m} A_{m} \to A, n_{m+1})
\end{multline}
a \emph{derivation} of $A$ from $\TO$ and $\Abox$ if the axioms are from $\TO$ and the numbers $n_1, \dots, n_m, n_{m+1}$ are such that $n_{j+1} = n_j + i_j$ and $\Abox \models C_1^j \land \dots \land C_{k_j}^j(n_{j+1})$. We say that such a derivation ends at $n$ if $n_{m+1} = n$. It is straightforward to verify that 
$A \in \tp_{\TO, \Abox}^{\sig(\TO)}(\ell)$ iff there is a derivation of $A$ at $\ell$, for any $\ell \in \Z$.


Let $\Abox$ be of the form $\emptyset^N \mathcal B \emptyset^N$. We now show that, for any $\ell$ with $N \leq \ell < |\mathcal A|-N$,
\begin{multline}\label{eq:der-prop}
  \text{if there is a derivation of } A \text{ at } \ell, \text{ then there is a derivation of } A \text{ at } \ell \\
  \text{ such that }0 \leq n_j < |\Abox| \text{ for all } n_j \text{ in it}.
\end{multline}


%
\begin{proposition}\label{th:short-ders}
Let $\mathfrak D_1$, $\mathfrak D_2$, $\mathfrak D_3$ be derivations from $\TO$ and $\Abox$ of the form\textup{:}
\begin{align*}
& \mathfrak D_1 =  \dots, (C_1 \land \dots \land C_{k} \land \nxt^{i} A \to A_0, n_0),\\
& \mathfrak D_2 = (\nxt^{i_0} A_0 \to A_1, n_1), \dots, (\nxt^{i_{m-1}} A_{m-1} \to A_m, n_m),\\
& \mathfrak D_3 = (C_1' \land \dots \land C_{k'}' \land \nxt^{i} A_m \to A_{m+1}, n_{m+1}), \dots
\end{align*}
%
%
If $\mathfrak D_1 \mathfrak D_2 \mathfrak D_3$ is a derivation of $A$ at $\ell$, then there is a derivation $\mathfrak D_1 \mathfrak D_2' \mathfrak D_3$ of $A$ at $\ell$ from $\TO$ and $\Abox$ such that $\min \{ n_0, n_{m+1} \} - 2M^2 \leq n_j  \leq \max \{ n_0, n_{m+1} \} + 2M^2$ for all $n_j$ in $\mathfrak D_2'$.
\end{proposition}
\begin{proof}
%
%
%
Suppose $n_{m+1} > n_0$ (the opposite case is analogous). Let $j$ be the earliest number in $\mathfrak D_2$ such that
\begin{itemize}
\item either $n_j = n_{m+1}$ and $n_{j+k} = n_{m+1}$ for some $k \geq 0$,

\item or $n_j = n_{0}$ and $n_{j+k} = n_{0}$ for some $k \geq 0$.
\end{itemize}
If there is no such $j$, then Proposition~\ref{th:short-ders} holds with $\mathfrak D_2' = \mathfrak D_2$. Suppose the former case holds for the earliest $j$. Let $\mathfrak D_2 = \mathfrak D_4 \mathfrak D_5 \mathfrak D_6$, where $\mathfrak D_5$ is the subsequence of $\mathfrak D_2$ between $j$ (not inclusive) and $j+k$. Consider any quadruple $((A_{j'}, n_{j'}),(A_{j''}, n_{j''}), (A_{k''}, n_{k''}), (A_{k'}, n_{k'}))$ in $\mathfrak D_5$ with $j' \leq  j'' \leq k'' \leq k'$, $n_{j'} = n_{k'}$, $n_{j''} = n_{k''}$, $A_{j'} = A_{j''}$ and $A_{k'} = A_{k''}$. Clearly, $\mathfrak D_1 (\mathfrak D_4 \mathfrak D_5' \mathfrak D_6) \mathfrak D_3$ is also a derivation $A$ at $\ell$ from $\TO$ and $\Abox$, where
\begin{multline*}
\mathfrak D_5' = (\nxt^{i_j} A_j \to A_{j+1}, n_{j+1}), \dots, (\nxt^{i_{j'-1}} A_{j'-1} \to A_{j'}, n_{j'}), \\
\shoveright{(\nxt^{i_{j''}} A_{j''} \to A_{j''+1}, n_{j''+1} - d),\dots}\\
\shoveleft{\quad\quad\quad(\nxt^{i_{k''-1}} A_{k''-1} \to A_{k''}, n_{k''} - d), (\nxt^{i_{k'}} A_{k'} \to A_{k'+1}, n_{k'+1}), \dots,}\\
(\nxt^{i_{j+k-1}} A_{j+k-1} \to A_{j+k}, n_{j+k})
\end{multline*}
and $d = n_{j''} - n_{j'}$. After recursively applying to $\mathfrak D_5$ the transformation above for each quadruple $((A_{j'}, n_{j'}),(A_{j''}, n_{j''}), (A_{k''}, n_{k''}), (A_{k'}, n_{k'}))$, we obtain $\mathfrak D_5'$. It is easy to check that there exist no $n_1 \neq n_2$ and atoms $A,B$ such that $(\nxt^{i_1} A_{1} \to A, n_1), \dots, (\nxt^{i_2} A_{2} \to B, n_1)$ and $(\nxt^{i_3} A_{3} \to A, n_2), \dots, (\nxt^{i_4} A_{4} \to B, n_2)$ are in $\mathfrak D_5'$. Therefore, $|n_{j'} - n_{m+1}| \leq 2M^2$ for all numbers $n_{j'}$ in $\mathfrak D_5'$. If the latter case holds for the earliest $j$, we can transform the subsequence $\mathfrak D_5$ of $\mathfrak D_2$ between $j$ (not inclusive) and $j+k$ into the subsequence $\mathfrak D_5'$ with all numbers $|n_{j'} - n_{0}| \leq 2M^2$. Then we find $j$ in $\mathfrak D_6$ satisfying one of the two cases above and transform $\mathfrak D_6$ analogously. We proceed until there are no more $j$ satisfying either of the two cases and the result $\mathfrak D_2'$ of the transformation is as required by the proposition.
\end{proof}

To show~\eqref{eq:der-prop}, consider a derivation $\mathfrak D$ of $A$ at $\ell$, for $N \leq \ell < |\mathcal A|-N$, with the numbers $n_j$.
Take the first $n_j$ such that $n_j \geq |\mathcal B|+M$ or $n_j < 2M^2$. Suppose the former is the case. Since $\Abox_i = \emptyset$ for $|\emptyset^N \mathcal B| \leq i < |\Abox|$, there are $n_{j'}$, for $j' < j$,  such that $2M^2 \leq  n_{j'} < |\mathcal B|+M$ and a (sub)sequence $(\nxt^{i_{j'}} A_{j'} \to A_{j' + 1}, n_{j'+1}), \dots,$ $(\nxt^{i_{j-1}} A_{j-1} \to A_{j}, n_j)$ is in $\mathfrak D$. We expand this subsequence by taking all $(\nxt^{i_{j}} A_{j} \to A_{j+1}, n_j), \dots, (\nxt^{i_{j''-1}} A_{j''-1} \to A_{j''}, n_{j''})$, such that $j''$ is the first after $j$ such that $n_{j''}=n_{j'}$. Let $\mathfrak D = \mathfrak D_1 \mathfrak D_2 \mathfrak D_3$, where $\mathfrak D_2$ is the expanded sequence above. By applying Proposition~\ref{th:short-ders}, we obtain a derivation $\mathfrak D_1 \mathfrak D_2' \mathfrak D_3$ of $A$ at $\ell$, where all numbers $n_j$ in $\mathfrak D_1 \mathfrak D_2'$ are such that $2M^2 \leq n_j \leq n_{j'}+2M^2 < |\Abox|$. In case $n_j < 2M^2$, we analogously obtain a derivation of $A$ at $\ell$, where all numbers $n_j$ in $\mathfrak D_1 \mathfrak D_2'$ are such that $0 \leq n_{j'} - 2M^2 \leq n_j < |\mathcal B|+M$. By continuing to apply Proposition~\ref{th:short-ders} to $\mathfrak D_3$ the required number of times, we obtain a derivation of $A$ at $\ell$ satisfying~\eqref{eq:der-prop}.


This completes the proof of Lemma~\ref{th:derivations-runs} as, clearly, for any $\ell$ with $N \leq \ell < |\mathcal A| - N$,  there is a run $(q_0, 0), \dots, (q, \ell), (q_A, i)$ of $\A_\TO^\Xi$ on $\Abox$ iff there is a derivation of  $A$  at  $\ell$ such that $0 \leq n_j < |\Abox|$  for all $n_j$ in it.
\end{proof}


\subsection{Proof of Theorem~\ref{th:foe-horn-crit}}\label{le:26}

\noindent
\textbf{Theorem~\ref{th:foe-horn-crit}.}
\emph{Let $\q = (\TO, \varkappa)$ be an OMPQ with a $\bot$-free $\LTL_\horn\Xallop$-ontology $\TO$. Then $\q$ is not $\FO(<,\equiv)$-rewritable over $\Xi$-Aboxes iff there are $\Abox, \mathcal{B}, \mathcal{D} \in \Sigma_\Xi^*$ and $k \geq 2$ such that $(i)$ and $(ii)$ from Theorem~\ref{th:fo-horn-crit} hold and
there are $\mathcal{U}, \mathcal{V} \in \Sigma_\Xi^*$ such that $\mathcal B = \mathcal V \mathcal U$, $|\mathcal{U}|= |\mathcal{V}|$,
\begin{description}
\item[$(iii)$]
 $\tp_{\TO, \Abox \mathcal{B}^{k} \mathcal{D}}(|\Abox \mathcal{B}^i|-1) = \tp_{\TO, \Abox \mathcal{B}^{k} \mathcal{D}}(|\Abox \mathcal{B}^i \mathcal V|-1)$, for all $i < k$, and
\item[$(iv)$] 
 $\tp_{\TO, \Abox \mathcal{B}^{k+1} \mathcal{D}}(|\Abox \mathcal{B}^i|-1) = \tp_{\TO, \Abox \mathcal{B}^{k+1} \mathcal{D}}(|\Abox \mathcal{B}^i \mathcal V|-1)$, for all $i$ with $1 \leq i \leq k$.
\end{description}
}
\begin{proof}
Consider the DFA $\A=(Q,\Sigma,\delta,q_{-1},F)$ from the proof of Theorem~\ref{th:fo-horn-crit} such that $\L_\Xi(\q) = \L(\A)$.
$(\Rightarrow)$ Suppose $\q$ is not $\FO(<, \equiv)$-rewritable. By Theorem~\ref{DFAcrit} $(ii)$, there exist $\Abox,  \mathcal{V}, \mathcal{U}, \mathcal{D} \in \Sigma_\Xi^*$ with $|\mathcal{U}| = |\mathcal{V}|$ and $k \geq 2$ such that
$$
q_{-1} \Rightarrow_{\Abox} q_0 \Rightarrow_{\mathcal V} q_0 \Rightarrow_{\mathcal U} q_1 \Rightarrow_{\mathcal V} q_1 \Rightarrow_{\mathcal U} \dots
\Rightarrow_{\mathcal U} q_{k-1} \Rightarrow_{\mathcal V} q_{k-1} \Rightarrow_{\mathcal U} q_0,
$$
$q_0 \Rightarrow_{\mathcal D} r_0$, $q_1 \Rightarrow_{\mathcal D} r_1$ for some $q_0, \dots, q_{k-1}, r_0,r_1 \in Q$ with $r_0 \not \in F$ and $r_1 \in F$. That $(i)$ and $(ii)$ are satisfied for $\mathcal B = \mathcal{V} \mathcal{U}$ is shown as in the proof of Theorem~\ref{th:fo-horn-crit}. Then $(iii)$ and $(iv)$ easily follow from~\eqref{eq:fo-horn-crit2}.

$(\Leftarrow)$ Suppose $(i)$--$(iv)$ hold and set $\mathcal E(i_0, \dots, i_{j}) = \mathcal V^{i_0} \mathcal U \dots \mathcal V^{i_{j}} \mathcal U$.
Let $\mathcal F_{j'}(i_0, \dots, i_{j})$ be the prefix of $\mathcal E(i_0, \dots, i_{j})$ of the form $\mathcal V^{i_0} \mathcal U \dots \mathcal V^{i_{j'-1}} \mathcal U \mathcal V^{i_{j'}}$, for $j' \leq j$. By the properties of the canonical models, we then obtain the following, for any $0 \leq n \leq m$ and any $0 \leq  \ell < k$:
\begin{description}
\item[] 
(a) $
\tp_{\TO, \Abox \mathcal E(i_0, \dots, i_{km+k-1}) \mathcal D}(|\Abox \mathcal F_{kn+\ell}(i_0, \dots, i_{km+k-1})|-1) =
\tp_{\TO, \Abox \mathcal B^k \mathcal{D}}(|\Abox \mathcal{B}^{\ell}|-1),
$,

\item[] 
(b) $
\tp_{\TO, \Abox \mathcal E(i_0, \dots, i_{km+k-1}, i_0) \mathcal D}(|\Abox \mathcal F_{kn+\ell+1}(i_0, \dots, i_{km+k-1}, i_0)|-1) = \tp_{\TO, \Abox \mathcal B^{k+1} \mathcal{D} }(|\Abox \mathcal{B}^{\ell+1}|-1).$
\end{description}

The rest of the proof relies on the following observation:
\begin{proposition}\label{prop:cycle}
Let $\A$ be a DFA with a set of states $Q$, $|Q| \geq 3$, over an alphabet $\Sigma$. Then, for any $q \in Q$ and $w \in \Sigma^*$, there exists $q'$ such that $q \Rightarrow_{w^{|Q|!-1}} q' \Rightarrow_{w^{|Q|!}} q'$.
\end{proposition}

Take the DFA $\A$ from the proof of Theorem~\ref{th:fo-horn-crit}, assume without loss of generality that $|Q|\ge 3$, and, for $m \geq 0$, consider the sequence
\begin{multline*}
q_{-1} \Rightarrow_{\Abox \mathcal V^{|Q|!-1}}  q_0 \Rightarrow_{\mathcal V^{|Q|!}} q_0' \Rightarrow_{\mathcal U} q_0''  \Rightarrow_{\mathcal V^{|Q|!-1}}
 q_1 \Rightarrow_{\mathcal V^{|Q|!}} q_1' \Rightarrow_{\mathcal U} q_1'' \Rightarrow_{\mathcal V^{|Q|!-1}} \dots \\
q_{km+k-1} \Rightarrow_{\mathcal V^{|Q|!}} q'_{km+k-1} \Rightarrow_{\mathcal U} q_{km+k}.
\end{multline*}
By Proposition~\ref{prop:cycle}, $q_i = q_i'$ for $0 \leq i < km+k$. By taking an appropriate $m$, as in the proof of Lemma~\ref{th:fo-horn-crit}, we can find $i$ and $j$ such that
$$
q_{-1} \Rightarrow_{\Abox \mathcal V^{|Q|!-1} (\mathcal U \mathcal V^{|Q|!-1})^{ik}}  r_0 \Rightarrow_{\mathcal U \mathcal V^{|Q|!-1}} r_1 \Rightarrow_{\mathcal U \mathcal V^{|Q|!-1}} \dots
\Rightarrow_{\mathcal U \mathcal V^{|Q|!-1}} r_{jk + k-1} \Rightarrow_{\mathcal U \mathcal V^{|Q|!-1}} r_0
$$
and $r_\ell \Rightarrow_{\mathcal V^{|Q|!}} r_\ell$, for $0 \leq \ell < jk + k$.  It can be readily shown using (a) and (b) that $q_0' \not \in F$ and $q_1' \in F$ for such $q_0'$ and $q_1'$ that $r_0 \Rightarrow_{\mathcal D} q_0'$ and $r_1 \Rightarrow_{\mathcal D} q_1'$. Now, we have found a state $r_0$ in $\A$ that satisfies the condition of Theorem~\ref{DFAcrit} $(ii)$ with $u = \mathcal U \mathcal V^{|Q|!-1}$ and $v = \mathcal V^{|Q|!}$. Therefore, $\q$ is not $\FO(<,\equiv)$-rewritable.
\end{proof}


\subsection{Proof of Lemma~\ref{checkNP}}\label{standard}

\noindent
\textbf{Lemma~\ref{checkNP}.}
\emph{Given $a_1, \dots, a_l \in \Sigma_\Xi$ with $|a_i| = 1$, for $1 \le i \le l$, binary numbers $i_1,\dots, i_{l+1}, j$, a $\bot$-free $\LTL_\core\Xnext$-ontology $\TO$ and a positive existential temporal concept $\varkappa$, checking whether $\mathcal C_{\TO, \Abox} \models \varkappa(j)$ for $\Abox = \emptyset^{i_1}a_1 \dots \emptyset^{i_l}a_l \emptyset^{i_{l+1}}$ can be done in \NP{}.}

\begin{proof}
We first show that, for any ABox $\Abox$, we have $\mathcal C_{\TO, \Abox} \models \varkappa(j)$ iff there exist numbers $n$, $n'$ and $k$, $k'$ with $0 < n, k \leq |\{\Diamond \varkappa' \in \sub(\varkappa) \}|+3$, $0 \leq n' < n$, $0 \leq k' < k$, a set of numbers $\{j_{-k}, \dots, j_{-1}, j_0, j_1, \dots j_n\} \subseteq \Z$, and types $\tp_{-k}, \dots, \tp_{-1}, \tp_0, \tp_1, \dots \tp_{n}$ for $(\TO, \varkappa)$ such that:
\begin{itemize}
\item $j_i < j_{i+1}$, for all $i$ with $-k \leq i < n$, and $j_0 = j$;

\item $j_{i+1} - j_i \leq 2^{O(|\q|)}$ if $j_i > \max \Abox$ or $j_{i+1} < 0$;

\item $\tp_{\TO, \Abox}^{\sig(\TO)}(j_i) \subseteq \tp_i$ , for all $i$ with $-k \leq i < n$, and $\varkappa \in \tp_0$;

\item $\tp_{n} = \tp_{n'}$ and $\tp_{-k} = \tp_{-k'}$;

\item for all $i < n'$ and $\Rdiamond \varkappa' \in \sub(\varkappa)$, $\Rdiamond \varkappa' \in \tp_i$ implies $\varkappa' \in \tp_{i'}$ for some $i' \in (i, n]$;

\item for all $i \in [n', n]$ and $\Rdiamond \varkappa' \in \sub(\varkappa)$, $\Rdiamond \varkappa' \in \tp_i$ implies $\varkappa' \in \tp_{i'}$ for some $i' \in [n', n]$;

\item for all $i < n'$ and $\Rdiamond \varkappa' \in \sub(\varkappa)$, $\varkappa' \in \tp_i$ implies $\Rdiamond \varkappa' \in \tp_{i'}$ for all $i' < i$;

\item for all $i \in [n', n]$ and $\Rdiamond \varkappa' \in \sub(\varkappa)$, $\varkappa' \in \tp_i$ implies $\Rdiamond  \varkappa' \in \tp_{i'}$ for all $i' \in [n', n]$,
\end{itemize}
and similarly for $\Ldiamond \varkappa'$ formulas. 

$(\Rightarrow)$ Suppose $(\TO,\Abox)\models\varkappa(j)$, so $\varkappa \in \tp_{\TO, \Abox}(j)$. Let $\Phi$ be the set of $\Diamond \varkappa' \in \tp_{\TO, \Abox}(j)$ for which there exist (unique) $j_{\varkappa'}$ satisfying $\neg \Diamond \varkappa', \varkappa' \in \tp_{\TO, \Abox}(j_{\varkappa'})$. Let $\{ j_{\varkappa'} \mid \Diamond \varkappa' \in \Phi\} \cup \{ j \} = \{j_{-k'}, \dots, j_{-1}, j_0, j_1, \dots, j_{n'} \}$ such that $j_0 = j$  $j_{-k'} < j_{-k'+1} < \dots  < j_{n'-1} < j_{n'}$. We take the smallest numbers $j_{n'+1}$ and $j''$ exceeding  $\max \Abox$ for which $\tp_{\TO, \Abox}(j_{n'+1})= \tp_{\TO, \Abox}(j'')$ and $j'' > j_{n'+1}$. Let $\Psi$ be the set of all $\Rdiamond \varkappa' \in \tp_{\TO, \Abox}(j_{n'+1})$. For each $\Rdiamond \varkappa' \in \Psi$, we take the smallest $j_{\varkappa'} \in (j_{n'+1}, j'']$ with $\varkappa' \in \tp_{\TO, \Abox}(j_{\varkappa'})$. Let $\{ j_{\varkappa'} \mid \Rdiamond \varkappa' \in \Psi\} = \{j_{n'+2}, \dots, j_{n-1} \}$. Finally, we set $j_n = j''$ (for the appropriate $n$). The selection of $k$ and $j_{-k}, \dots, j_{-k'-1}$ is analogous and left to the reader. We take $\tp_i = \tp_{\TO, \Abox}(j_i)$, for $i \in [-k, n]$. Using the periodicity property of the canonical models~\cite<e.g.,>[Lemma 22]{DBLP:journals/ai/ArtaleKKRWZ21}, one  can check that the required conditions are satisfied.

$(\Leftarrow)$ Suppose there are $n$, $m$, $n'$, $m'$, $j_i$ and $\tp_i$ satisfying the conditions above. It is easy to check by induction on the construction of $\varkappa'$ that $\varkappa' \in \tp_i$ implies $\varkappa' \in \tp_{\TO, \Abox}(j_i)$ for all $\varkappa' \in \sub(\varkappa)$. As $\varkappa \in \tp_0$, it follows that $(\TO,\Abox)\models\varkappa(j)$.
%

It is now easy to provide the required \NP{} algorithm. Indeed, we first guess the required binary numbers $j_i$ and types (recall that $j_{i+1} - j_i \leq 2^{O(|\TO|+|\varkappa|)}$ if $j_i > \max \Abox$ or $j_{i+1} < 0$). The list of conditions above can be checked in polynomial time. In particular, $\tp_{\TO, \Abox}^{\sig(\TO)}(j_i) \subseteq \tp_i$ for $\Abox = \emptyset^{i_1}a_1 \dots \emptyset^{i_l}a_l \emptyset^{i_{l+1}}$ can be checked in polynomial time using arithmetic progressions~\cite<e.g.,>[Theorem 14]{DBLP:journals/ai/ArtaleKKRWZ21}.
\end{proof}


\vskip 0.2in

\begin{thebibliography}{}

\bibitem[\protect\BCAY{Abiteboul, Hull,\ \BBA\ Vianu}{Abiteboul
  et~al.}{1995}]{Abitebouletal95}
Abiteboul, S., Hull, R., \BBA\ Vianu, V. \BBOP1995\BBCP.
\newblock {\Bem Foundations of Databases}.
\newblock Addison-Wesley.

\bibitem[\protect\BCAY{Afrati\ \BBA\ Papadimitriou}{Afrati\ \BBA\
  Papadimitriou}{1993}]{DBLP:journals/jacm/AfratiP93}
Afrati, F.~N.\BBACOMMA\  \BBA\ Papadimitriou, C.~H. \BBOP1993\BBCP.
\newblock \BBOQ The parallel complexity of simple logic programs\BBCQ\
\newblock {\Bem J. {ACM}}, {\Bem 40\/}(4), 891--916.

\bibitem[\protect\BCAY{Alur\ \BBA\ Henzinger}{Alur\ \BBA\
  Henzinger}{1993}]{DBLP:journals/iandc/AlurH93}
Alur, R.\BBACOMMA\  \BBA\ Henzinger, T.~A. \BBOP1993\BBCP.
\newblock \BBOQ Real-time logics: Complexity and expressiveness\BBCQ\
\newblock {\Bem Inf. Comput.}, {\Bem 104\/}(1), 35--77.

\bibitem[\protect\BCAY{Arora\ \BBA\ Barak}{Arora\ \BBA\
  Barak}{2009}]{Arora&Barak09}
Arora, S.\BBACOMMA\  \BBA\ Barak, B. \BBOP2009\BBCP.
\newblock {\Bem Computational Complexity: A Modern Approach}.
\newblock Cambridge University Press, New York, NY, USA.

\bibitem[\protect\BCAY{Artale, Kontchakov, Ryzhikov,\ \BBA\
  Zakharyaschev}{Artale et~al.}{2013}]{AKRZ:LPAR13}
Artale, A., Kontchakov, R., Ryzhikov, V., \BBA\ Zakharyaschev, M.
  \BBOP2013\BBCP.
\newblock \BBOQ The complexity of clausal fragments of {LTL}\BBCQ\
\newblock In {\Bem Proc.\ of the 19th Int.\ Conf.\ on Logic for Programming,
  Artificial Intelligence and Reasoning, LPAR 2013}, \lowercase{\BVOL}\ 8312 of
  {\Bem Lecture Notes in Computer Science}, \BPGS\ 35--52. Springer.

\bibitem[\protect\BCAY{Artale, Calvanese, Kontchakov,\ \BBA\
  Zakharyaschev}{Artale et~al.}{2009}]{ACKZ09}
Artale, A., Calvanese, D., Kontchakov, R., \BBA\ Zakharyaschev, M.
  \BBOP2009\BBCP.
\newblock \BBOQ The {DL}-{L}ite family and relations\BBCQ\
\newblock {\Bem J. Artif. Intell. Res.}, {\Bem 36}, 1--69.

\bibitem[\protect\BCAY{Artale, Kontchakov, Kovtunova, Ryzhikov, Wolter,\ \BBA\
  Zakharyaschev}{Artale et~al.}{2015}]{DBLP:conf/ijcai/ArtaleKKRWZ15}
Artale, A., Kontchakov, R., Kovtunova, A., Ryzhikov, V., Wolter, F., \BBA\
  Zakharyaschev, M. \BBOP2015\BBCP.
\newblock \BBOQ First-order rewritability of temporal ontology-mediated
  queries\BBCQ\
\newblock In {\Bem Proc. of the 24th Int. Joint Conference on Artificial
  Intelligence, {IJCAI'15}}, \BPGS\ 2706--2712.

\bibitem[\protect\BCAY{Artale, Kontchakov, Kovtunova, Ryzhikov, Wolter,\ \BBA\
  Zakharyaschev}{Artale et~al.}{2017}]{DBLP:conf/time/ArtaleKKRWZ17}
Artale, A., Kontchakov, R., Kovtunova, A., Ryzhikov, V., Wolter, F., \BBA\
  Zakharyaschev, M. \BBOP2017\BBCP.
\newblock \BBOQ Ontology-mediated query answering over temporal data: {A}
  survey (invited talk)\BBCQ\
\newblock In Schewe, S., Schneider, T., \BBA\ Wijsen, J.\BEDS, {\Bem 24th
  International Symposium on Temporal Representation and Reasoning, {TIME}
  2017, October 16-18, 2017, Mons, Belgium}, \lowercase{\BVOL}~90 of {\Bem
  LIPIcs}, \BPGS\ 1:1--1:37. Schloss Dagstuhl - Leibniz-Zentrum f{\"{u}}r
  Informatik.

\bibitem[\protect\BCAY{Artale, Kontchakov, Kovtunova, Ryzhikov, Wolter,\ \BBA\
  Zakharyaschev}{Artale et~al.}{2021}]{DBLP:journals/ai/ArtaleKKRWZ21}
Artale, A., Kontchakov, R., Kovtunova, A., Ryzhikov, V., Wolter, F., \BBA\
  Zakharyaschev, M. \BBOP2021\BBCP.
\newblock \BBOQ First-order rewritability of ontology-mediated queries in
  linear temporal logic\BBCQ\
\newblock {\Bem Artif. Intell.}, {\Bem 299}, 103536.

\bibitem[\protect\BCAY{Artale, Kontchakov, Kovtunova, Ryzhikov, Wolter,\ \BBA\
  Zakharyaschev}{Artale et~al.}{2022}]{DBLP:journals/corr/abs-2111-06806}
Artale, A., Kontchakov, R., Kovtunova, A., Ryzhikov, V., Wolter, F., \BBA\
  Zakharyaschev, M. \BBOP2022\BBCP.
\newblock \BBOQ First-order rewritability and complexity of two-dimensional
  temporal ontology-mediated queries\BBCQ\
\newblock {\Bem J. Artif. Intell. Res.}, {\Bem 75}, 1223--1291.

\bibitem[\protect\BCAY{Baget, Lecl{\`e}re, Mugnier,\ \BBA\ Salvat}{Baget
  et~al.}{2011}]{DBLP:journals/ai/BagetLMS11}
Baget, J.-F., Lecl{\`e}re, M., Mugnier, M.-L., \BBA\ Salvat, E. \BBOP2011\BBCP.
\newblock \BBOQ On rules with existential variables: Walking the decidability
  line\BBCQ\
\newblock {\Bem Artif. Intell.}, {\Bem 175\/}(9--10), 1620--1654.

\bibitem[\protect\BCAY{Barrington}{Barrington}{1989}]{DBLP:journals/jcss/Barrington89}
Barrington, D. A.~M. \BBOP1989\BBCP.
\newblock \BBOQ Bounded-width polynomial-size branching programs recognize
  exactly those languages in {NC\({^1}\)}\BBCQ\
\newblock {\Bem J. Comput. Syst. Sci.}, {\Bem 38\/}(1), 150--164.

\bibitem[\protect\BCAY{Barrington, Compton, Straubing,\ \BBA\
  Th{\'{e}}rien}{Barrington et~al.}{1992}]{DBLP:journals/jcss/BarringtonCST92}
Barrington, D. A.~M., Compton, K.~J., Straubing, H., \BBA\ Th{\'{e}}rien, D.
  \BBOP1992\BBCP.
\newblock \BBOQ Regular languages in {NC}{\({^1}\)}\BBCQ\
\newblock {\Bem J. Comput. Syst. Sci.}, {\Bem 44\/}(3), 478--499.

\bibitem[\protect\BCAY{Barrington\ \BBA\ Th{\'e}rien}{Barrington\ \BBA\
  Th{\'e}rien}{1988}]{DBLP:journals/jacm/BarringtonT88}
Barrington, D. A.~M.\BBACOMMA\  \BBA\ Th{\'e}rien, D. \BBOP1988\BBCP.
\newblock \BBOQ Finite monoids and the fine structure of
  {NC}$^{\mbox{1}}$\BBCQ\
\newblock {\Bem J. ACM}, {\Bem 35\/}(4), 941--952.

\bibitem[\protect\BCAY{Beaudry, McKenzie,\ \BBA\ Th\'{e}rien}{Beaudry
  et~al.}{1992}]{Beaudryetal92}
Beaudry, M., McKenzie, P., \BBA\ Th\'{e}rien, D. \BBOP1992\BBCP.
\newblock \BBOQ The membership problem in aperiodic transformation
  monoids\BBCQ\
\newblock {\Bem J. ACM}, {\Bem 39\/}(3), 599–616.

\bibitem[\protect\BCAY{Benedikt, ten Cate, Colcombet,\ \BBA\ {Vanden
  Boom}}{Benedikt et~al.}{2015}]{DBLP:conf/lics/BenediktCCB15}
Benedikt, M., ten Cate, B., Colcombet, T., \BBA\ {Vanden Boom}, M.
  \BBOP2015\BBCP.
\newblock \BBOQ The complexity of boundedness for guarded logics\BBCQ\
\newblock In {\Bem 30th Annual {ACM/IEEE} Symposium on Logic in Computer
  Science, {LICS} 2015, Kyoto, Japan, July 6-10, 2015}, \BPGS\ 293--304. {IEEE}
  Computer Society.

\bibitem[\protect\BCAY{Bennett, Martin, O'Bryant,\ \BBA\ Rechnitzer}{Bennett
  et~al.}{2018}]{Illin2018}
Bennett, M., Martin, G., O'Bryant, K., \BBA\ Rechnitzer, A. \BBOP2018\BBCP.
\newblock \BBOQ Explicit bounds for primes in arithmetic progressions\BBCQ\
\newblock {\Bem Illinois Journal of Mathematics}, {\Bem 62\/}(1--4), 427--532.

\bibitem[\protect\BCAY{Bern{\'{a}}tsky}{Bern{\'{a}}tsky}{1997}]{DBLP:journals/actaC/Bernatsky97}
Bern{\'{a}}tsky, L. \BBOP1997\BBCP.
\newblock \BBOQ Regular expression star-freeness is {PSPACE}-complete\BBCQ\
\newblock {\Bem Acta Cybern.}, {\Bem 13\/}(1), 1--21.

\bibitem[\protect\BCAY{Bienvenu, ten Cate, Lutz,\ \BBA\ Wolter}{Bienvenu
  et~al.}{2014}]{DBLP:journals/tods/BienvenuCLW14}
Bienvenu, M., ten Cate, B., Lutz, C., \BBA\ Wolter, F. \BBOP2014\BBCP.
\newblock \BBOQ Ontology-based data access: {A} study through disjunctive
  datalog, {CSP}, and {MMSNP}\BBCQ\
\newblock {\Bem {ACM} Transactions on Database Systems}, {\Bem 39\/}(4),
  33:1--44.

\bibitem[\protect\BCAY{Bienvenu, Kikot, Kontchakov, Podolskii, Ryzhikov,\ \BBA\
  Zakharyaschev}{Bienvenu et~al.}{2017}]{DBLP:conf/pods/BienvenuKKPRZ17}
Bienvenu, M., Kikot, S., Kontchakov, R., Podolskii, V.~V., Ryzhikov, V., \BBA\
  Zakharyaschev, M. \BBOP2017\BBCP.
\newblock \BBOQ The complexity of ontology-based data access with {OWL} 2 {QL}
  and bounded treewidth queries\BBCQ\
\newblock In {\Bem Proc.\ of the 36th {ACM} {SIGMOD-SIGACT-SIGAI} Symposium on
  Principles of Database Systems, PODS 2017}, \BPGS\ 201--216. ACM.

\bibitem[\protect\BCAY{Bienvenu, Kikot, Kontchakov, Podolskii,\ \BBA\
  Zakharyaschev}{Bienvenu et~al.}{2018}]{DBLP:journals/jacm/BienvenuKKPZ18}
Bienvenu, M., Kikot, S., Kontchakov, R., Podolskii, V.~V., \BBA\ Zakharyaschev,
  M. \BBOP2018\BBCP.
\newblock \BBOQ Ontology-mediated queries: Combined complexity and succinctness
  of rewritings via circuit complexity\BBCQ\
\newblock {\Bem J. {ACM}}, {\Bem 65\/}(5), 28:1--28:51.

\bibitem[\protect\BCAY{Borgwardt, Forkel,\ \BBA\ Kovtunova}{Borgwardt
  et~al.}{2019}]{DBLP:conf/ruleml/BorgwardtFK19}
Borgwardt, S., Forkel, W., \BBA\ Kovtunova, A. \BBOP2019\BBCP.
\newblock \BBOQ Finding new diamonds: Temporal minimal-world query answering
  over sparse aboxes\BBCQ\
\newblock In Fodor, P., Montali, M., Calvanese, D., \BBA\ Roman, D.\BEDS, {\Bem
  Rules and Reasoning - Third International Joint Conference, RuleML+RR 2019,
  Bolzano, Italy, September 16-19, 2019, Proceedings}, \lowercase{\BVOL}\ 11784
  of {\Bem Lecture Notes in Computer Science}, \BPGS\ 3--18. Springer.

\bibitem[\protect\BCAY{Bourhis\ \BBA\ Lutz}{Bourhis\ \BBA\
  Lutz}{2016}]{DBLP:conf/kr/BourhisL16}
Bourhis, P.\BBACOMMA\  \BBA\ Lutz, C. \BBOP2016\BBCP.
\newblock \BBOQ Containment in monadic disjunctive datalog, {MMSNP}, and
  expressive description logics\BBCQ\
\newblock In Baral, C., Delgrande, J.~P., \BBA\ Wolter, F.\BEDS, {\Bem
  Principles of Knowledge Representation and Reasoning: Proceedings of the
  Fifteenth International Conference, {KR} 2016, Cape Town, South Africa, April
  25-29, 2016}, \BPGS\ 207--216. {AAAI} Press.

\bibitem[\protect\BCAY{Brandt, Calvanese, Kalayci, Kontchakov, M{\"{o}}rzinger,
  Ryzhikov, Xiao,\ \BBA\ Zakharyaschev}{Brandt
  et~al.}{2019}]{DBLP:conf/time/0001CKKMR0Z19}
Brandt, S., Calvanese, D., Kalayci, E.~G., Kontchakov, R., M{\"{o}}rzinger, B.,
  Ryzhikov, V., Xiao, G., \BBA\ Zakharyaschev, M. \BBOP2019\BBCP.
\newblock \BBOQ Two-dimensional rule language for querying sensor log data: {A}
  framework and use cases\BBCQ\
\newblock In Gamper, J., Pinchinat, S., \BBA\ Sciavicco, G.\BEDS, {\Bem 26th
  International Symposium on Temporal Representation and Reasoning, {TIME}
  2019, October 16-19, 2019, M{\'{a}}laga, Spain}, \lowercase{\BVOL}\ 147 of
  {\Bem LIPIcs}, \BPGS\ 7:1--7:15. Schloss Dagstuhl - Leibniz-Zentrum f{\"{u}}r
  Informatik.

\bibitem[\protect\BCAY{Brandt, Kalayci, Ryzhikov, Xiao,\ \BBA\
  Zakharyaschev}{Brandt et~al.}{2018}]{DBLP:journals/jair/BrandtKRXZ18}
Brandt, S., Kalayci, E.~G., Ryzhikov, V., Xiao, G., \BBA\ Zakharyaschev, M.
  \BBOP2018\BBCP.
\newblock \BBOQ Querying log data with metric temporal logic\BBCQ\
\newblock {\Bem J. Artif. Intell. Res.}, {\Bem 62}, 829--877.

\bibitem[\protect\BCAY{Cal\`{\i}, Gottlob,\ \BBA\ Pieris}{Cal\`{\i}
  et~al.}{2012}]{DBLP:journals/ai/CaliGP12}
Cal\`{\i}, A., Gottlob, G., \BBA\ Pieris, A. \BBOP2012\BBCP.
\newblock \BBOQ Towards more expressive ontology languages: The query answering
  problem\BBCQ\
\newblock {\Bem Artif. Intell.}, {\Bem 193}, 87--128.

\bibitem[\protect\BCAY{Calvanese, De~Giacomo, Lembo, Lenzerini,\ \BBA\
  Rosati}{Calvanese et~al.}{2007}]{CDLLR07}
Calvanese, D., De~Giacomo, G., Lembo, D., Lenzerini, M., \BBA\ Rosati, R.
  \BBOP2007\BBCP.
\newblock \BBOQ Tractable reasoning and efficient query answering in
  description logics: the {{\textit{DL-Lite}}} family\BBCQ\
\newblock {\Bem Journal of Automated Reasoning}, {\Bem 39\/}(3), 385--429.

\bibitem[\protect\BCAY{Carton\ \BBA\ Dartois}{Carton\ \BBA\
  Dartois}{2015}]{carton_et_al:LIPIcs:2015:5413}
Carton, O.\BBACOMMA\  \BBA\ Dartois, L. \BBOP2015\BBCP.
\newblock \BBOQ {Aperiodic Two-way Transducers and FO-Transductions}\BBCQ\
\newblock In Kreutzer, S.\BED, {\Bem 24th EACSL Annual Conference on Computer
  Science Logic (CSL 2015)}, \lowercase{\BVOL}~41 of {\Bem Leibniz
  International Proceedings in Informatics (LIPIcs)}, \BPGS\ 160--174,
  Dagstuhl, Germany. Schloss Dagstuhl--Leibniz-Zentrum fuer Informatik.

\bibitem[\protect\BCAY{Cho\ \BBA\ Huynh}{Cho\ \BBA\
  Huynh}{1991}]{DBLP:journals/TCS/ChoHyunh91}
Cho, S.\BBACOMMA\  \BBA\ Huynh, D.~T. \BBOP1991\BBCP.
\newblock \BBOQ Finite-automaton aperiodicity is {PSPACE}-complete\BBCQ\
\newblock {\Bem Theor. Comp. Sci.}, {\Bem 88\/}(1), 99--116.

\bibitem[\protect\BCAY{Civili\ \BBA\ Rosati}{Civili\ \BBA\
  Rosati}{2012}]{DBLP:conf/datalog/CiviliR12}
Civili, C.\BBACOMMA\  \BBA\ Rosati, R. \BBOP2012\BBCP.
\newblock \BBOQ A broad class of first-order rewritable tuple-generating
  dependencies\BBCQ\
\newblock In {\Bem Proc.\ of the 2nd Int.\ Datalog 2.0 Workshop},
  \lowercase{\BVOL}\ 7494 of {\Bem Lecture Notes in Computer Science}, \BPGS\
  68--80. Springer.

\bibitem[\protect\BCAY{Compton\ \BBA\ Laflamme}{Compton\ \BBA\
  Laflamme}{1990}]{DBLP:journals/iandc/ComptonL90}
Compton, K.~J.\BBACOMMA\  \BBA\ Laflamme, C. \BBOP1990\BBCP.
\newblock \BBOQ An algebra and a logic for {NC}{\({^1}\)}\BBCQ\
\newblock {\Bem Inf. Comput.}, {\Bem 87\/}(1/2), 240--262.

\bibitem[\protect\BCAY{Cosmadakis, Gaifman, Kanellakis,\ \BBA\
  Vardi}{Cosmadakis et~al.}{1988}]{DBLP:conf/stoc/CosmadakisGKV88}
Cosmadakis, S.~S., Gaifman, H., Kanellakis, P.~C., \BBA\ Vardi, M.~Y.
  \BBOP1988\BBCP.
\newblock \BBOQ Decidable optimization problems for database logic programs
  (preliminary report)\BBCQ\
\newblock In {\Bem STOC}, \BPGS\ 477--490.

\bibitem[\protect\BCAY{Demri, Goranko,\ \BBA\ Lange}{Demri
  et~al.}{2016}]{DBLP:books/cu/Demri2016}
Demri, S., Goranko, V., \BBA\ Lange, M. \BBOP2016\BBCP.
\newblock {\Bem Temporal Logics in Computer Science}.
\newblock Cambridge Tracts in Theoretical Computer Science. Cambridge
  University Press.

\bibitem[\protect\BCAY{Feier, Kuusisto,\ \BBA\ Lutz}{Feier
  et~al.}{2019}]{DBLP:journals/lmcs/FeierKL19}
Feier, C., Kuusisto, A., \BBA\ Lutz, C. \BBOP2019\BBCP.
\newblock \BBOQ Rewritability in monadic disjunctive datalog, {MMSNP}, and
  expressive description logics\BBCQ\
\newblock {\Bem Log. Methods Comput. Sci.}, {\Bem 15\/}(2).

\bibitem[\protect\BCAY{Fisher, Dixon,\ \BBA\ Peim}{Fisher
  et~al.}{2001}]{FisherDP01}
Fisher, M., Dixon, C., \BBA\ Peim, M. \BBOP2001\BBCP.
\newblock \BBOQ Clausal temporal resolution\BBCQ\
\newblock {\Bem {ACM} Trans.\ Comput. Logic}, {\Bem 2\/}(1), 12--56.

\bibitem[\protect\BCAY{Fleischer\ \BBA\ Kufleitner}{Fleischer\ \BBA\
  Kufleitner}{2018}]{fleischeretal18}
Fleischer, L.\BBACOMMA\  \BBA\ Kufleitner, M. \BBOP2018\BBCP.
\newblock \BBOQ The intersection problem for finite monoids\BBCQ\
\newblock In Niedermeier, R.\BBACOMMA\  \BBA\ Vall{\'e}e, B.\BEDS, {\Bem Proc.
  STACS 2018}, \lowercase{\BVOL}~96 of {\Bem {LIPIcs}}, \BPGS\ 30:1--30:14,
  Dagstuhl, Germany. Schloss Dagstuhl--Leibniz-Zentrum fuer Informatik.

\bibitem[\protect\BCAY{Furst, Saxe,\ \BBA\ Sipser}{Furst
  et~al.}{1984}]{DBLP:journals/mst/FurstSS84}
Furst, M.~L., Saxe, J.~B., \BBA\ Sipser, M. \BBOP1984\BBCP.
\newblock \BBOQ Parity, circuits, and the polynomial-time hierarchy\BBCQ\
\newblock {\Bem Mathematical Systems Theory}, {\Bem 17\/}(1), 13--27.

\bibitem[\protect\BCAY{Gabbay, Kurucz, Wolter,\ \BBA\ Zakharyaschev}{Gabbay
  et~al.}{2003}]{gkwz}
Gabbay, D., Kurucz, A., Wolter, F., \BBA\ Zakharyaschev, M. \BBOP2003\BBCP.
\newblock {\Bem Many-Dimensional Modal Logics: Theory and Applications},
  \lowercase{\BVOL}\ 148 of {\Bem Studies in Logic}.
\newblock Elsevier.

\bibitem[\protect\BCAY{Gerasimova, Kikot, Kurucz, Podolskii,\ \BBA\
  Zakharyaschev}{Gerasimova et~al.}{2020}]{DBLP:conf/kr/GerasimovaKKPZ20}
Gerasimova, O., Kikot, S., Kurucz, A., Podolskii, V.~V., \BBA\ Zakharyaschev,
  M. \BBOP2020\BBCP.
\newblock \BBOQ A data complexity and rewritability tetrachotomy of
  ontology-mediated queries with a covering axiom\BBCQ\
\newblock In Calvanese, D., Erdem, E., \BBA\ Thielscher, M.\BEDS, {\Bem
  Proceedings of the 17th International Conference on Principles of Knowledge
  Representation and Reasoning, {KR} 2020, Rhodes, Greece, September 12-18,
  2020}, \BPGS\ 403--413.

\bibitem[\protect\BCAY{Guti{\'{e}}rrez{-}Basulto\ \BBA\
  Jung}{Guti{\'{e}}rrez{-}Basulto\ \BBA\
  Jung}{2017}]{DBLP:conf/ijcai/Gutierrez-Basulto17}
Guti{\'{e}}rrez{-}Basulto, V.\BBACOMMA\  \BBA\ Jung, J.~C. \BBOP2017\BBCP.
\newblock \BBOQ Combining dl-lite{\_}\{bool\}{\^{}}n with branching time: {A}
  gentle marriage\BBCQ\
\newblock In Sierra, C.\BED, {\Bem Proceedings of the Twenty-Sixth
  International Joint Conference on Artificial Intelligence, {IJCAI} 2017,
  Melbourne, Australia, August 19-25, 2017}, \BPGS\ 1074--1080. ijcai.org.

\bibitem[\protect\BCAY{Hillebrand, Kanellakis, Mairson,\ \BBA\
  Vardi}{Hillebrand et~al.}{1995}]{DBLP:journals/jlp/HillebrandKMV95}
Hillebrand, G.~G., Kanellakis, P.~C., Mairson, H.~G., \BBA\ Vardi, M.~Y.
  \BBOP1995\BBCP.
\newblock \BBOQ Undecidable boundedness problems for datalog programs\BBCQ\
\newblock {\Bem J. Log. Program.}, {\Bem 25\/}(2), 163--190.

\bibitem[\protect\BCAY{Hodges}{Hodges}{1993}]{DBLP:books/daglib/0071316}
Hodges, W. \BBOP1993\BBCP.
\newblock {\Bem Model theory}, \lowercase{\BVOL}~42 of {\Bem Encyclopedia of
  mathematics and its applications}.
\newblock Cambridge University Press.

\bibitem[\protect\BCAY{Jukna}{Jukna}{2012}]{DBLP:books/daglib/0028687}
Jukna, S. \BBOP2012\BBCP.
\newblock {\Bem Boolean Function Complexity - Advances and Frontiers},
  \lowercase{\BVOL}~27 of {\Bem Algorithms and combinatorics}.
\newblock Springer.

\bibitem[\protect\BCAY{Kaminski, Nenov,\ \BBA\ {Cuenca~Grau}}{Kaminski
  et~al.}{2016}]{DBLP:journals/ai/KaminskiNG16}
Kaminski, M., Nenov, Y., \BBA\ {Cuenca~Grau}, B. \BBOP2016\BBCP.
\newblock \BBOQ Datalog rewritability of disjunctive datalog programs and
  non-{H}orn ontologies\BBCQ\
\newblock {\Bem Artif. Intell.}, {\Bem 236}, 90--118.

\bibitem[\protect\BCAY{Kamp}{Kamp}{1968}]{phd-kamp}
Kamp, H.~W. \BBOP1968\BBCP.
\newblock {\Bem Tense Logic and the Theory of Linear Order}.
\newblock {PhD} thesis, Computer Science Department, University of California
  at Los~Angeles, USA.

\bibitem[\protect\BCAY{Kaplan\ \BBA\ Levy}{Kaplan\ \BBA\
  Levy}{2010}]{kaplan_levy_2010}
Kaplan, G.\BBACOMMA\  \BBA\ Levy, D. \BBOP2010\BBCP.
\newblock \BBOQ Solvability of finite groups via conditions on products of
  2-elements and odd p-elements\BBCQ\
\newblock {\Bem Bulletin of the Australian Mathematical Society}, {\Bem
  82\/}(2), 265–273.

\bibitem[\protect\BCAY{Kikot, Kurucz, Podolskii,\ \BBA\ Zakharyaschev}{Kikot
  et~al.}{2021}]{DBLP:conf/pods/KikotKPZ21}
Kikot, S., Kurucz, A., Podolskii, V.~V., \BBA\ Zakharyaschev, M.
  \BBOP2021\BBCP.
\newblock \BBOQ Deciding boundedness of monadic sirups\BBCQ\
\newblock In Libkin, L., Pichler, R., \BBA\ Guagliardo, P.\BEDS, {\Bem PODS'21:
  Proceedings of the 40th {ACM} {SIGMOD-SIGACT-SIGAI} Symposium on Principles
  of Database Systems, Virtual Event, China, June 20-25, 2021}, \BPGS\
  370--387. {ACM}.

\bibitem[\protect\BCAY{King}{King}{2005}]{DBLP:conf/bcc/King05}
King, O.~H. \BBOP2005\BBCP.
\newblock \BBOQ The subgroup structure of finite classical groups in terms of
  geometric configurations\BBCQ\
\newblock In Webb, B.~S.\BED, {\Bem Surveys in Combinatorics, 2005 [invited
  lectures from the Twentieth British Combinatorial Conference, Durham, UK,
  July 2005]}, \lowercase{\BVOL}\ 327 of {\Bem London Mathematical Society
  Lecture Note Series}, \BPGS\ 29--56. Cambridge University Press.

\bibitem[\protect\BCAY{Knuth}{Knuth}{1998}]{DBLP:books/lib/Knuth98}
Knuth, D.~E. \BBOP1998\BBCP.
\newblock {\Bem The art of computer programming, Volume {II:} Seminumerical
  Algorithms, 3rd Edition}.
\newblock Addison-Wesley.

\bibitem[\protect\BCAY{Koymans}{Koymans}{1990}]{DBLP:journals/rts/Koymans90}
Koymans, R. \BBOP1990\BBCP.
\newblock \BBOQ Specifying real-time properties with metric temporal
  logic\BBCQ\
\newblock {\Bem Real-Time Systems}, {\Bem 2\/}(4), 255--299.

\bibitem[\protect\BCAY{Kozen}{Kozen}{1977}]{Kozen77}
Kozen, D. \BBOP1977\BBCP.
\newblock \BBOQ Lower bounds for natural proof systems\BBCQ\
\newblock In {\Bem 18th Annual Symposium on Foundations of Computer Science
  ({SFCS} 1977)}, \BPGS\ 254--266.

\bibitem[\protect\BCAY{Libkin}{Libkin}{2004}]{Libkin}
Libkin, L. \BBOP2004\BBCP.
\newblock {\Bem Elements Of Finite Model Theory}.
\newblock Springer.

\bibitem[\protect\BCAY{Lutz, Wolter,\ \BBA\ Zakharyaschev}{Lutz
  et~al.}{2008}]{LutzWZ08}
Lutz, C., Wolter, F., \BBA\ Zakharyaschev, M. \BBOP2008\BBCP.
\newblock \BBOQ Temporal description logics: {A} survey\BBCQ\
\newblock In {\Bem Proc. of the 15th Int.\ Symposium on Temporal Representation
  and Reasoning (TIME 2008)}, \BPGS\ 3--14.

\bibitem[\protect\BCAY{Lutz\ \BBA\ Sabellek}{Lutz\ \BBA\
  Sabellek}{2017}]{DBLP:conf/ijcai/LutzS17}
Lutz, C.\BBACOMMA\  \BBA\ Sabellek, L. \BBOP2017\BBCP.
\newblock \BBOQ Ontology-mediated querying with the description logic {EL:}
  trichotomy and linear datalog rewritability\BBCQ\
\newblock In Sierra, C.\BED, {\Bem Proceedings of the Twenty-Sixth
  International Joint Conference on Artificial Intelligence, {IJCAI} 2017,
  Melbourne, Australia, August 19-25, 2017}, \BPGS\ 1181--1187. ijcai.org.

\bibitem[\protect\BCAY{Lutz\ \BBA\ Sabellek}{Lutz\ \BBA\
  Sabellek}{2019}]{DBLP:journals/corr/abs-1904-12533}
Lutz, C.\BBACOMMA\  \BBA\ Sabellek, L. \BBOP2019\BBCP.
\newblock \BBOQ A complete classification of the complexity and rewritability
  of ontology-mediated queries based on the description logic {EL}\BBCQ\
\newblock {\Bem CoRR}, {\Bem abs/1904.12533}.

\bibitem[\protect\BCAY{Marcinkowski}{Marcinkowski}{1996}]{DBLP:conf/lics/Marcinkowski96}
Marcinkowski, J. \BBOP1996\BBCP.
\newblock \BBOQ {DATALOG} sirups uniform boundedness is undecidable\BBCQ\
\newblock In {\Bem Proceedings, 11th Annual {IEEE} Symposium on Logic in
  Computer Science, New Brunswick, New Jersey, USA, July 27-30, 1996}, \BPGS\
  13--24. {IEEE} Computer Society.

\bibitem[\protect\BCAY{Marcinkowski}{Marcinkowski}{1999}]{DBLP:journals/siamcomp/Marcinkowski99}
Marcinkowski, J. \BBOP1999\BBCP.
\newblock \BBOQ Achilles, turtle, and undecidable boundedness problems for
  small {DATALOG} programs\BBCQ\
\newblock {\Bem {SIAM} J. Comput.}, {\Bem 29\/}(1), 231--257.

\bibitem[\protect\BCAY{McNaughton\ \BBA\ Papert}{McNaughton\ \BBA\
  Papert}{1971}]{McNaughton&Papert71}
McNaughton, R.\BBACOMMA\  \BBA\ Papert, S. \BBOP1971\BBCP.
\newblock {\Bem Counter-free automata}.
\newblock The MIT Press.

\bibitem[\protect\BCAY{Ouaknine\ \BBA\ Worrell}{Ouaknine\ \BBA\
  Worrell}{2008}]{DBLP:conf/formats/OuaknineW08}
Ouaknine, J.\BBACOMMA\  \BBA\ Worrell, J. \BBOP2008\BBCP.
\newblock \BBOQ Some recent results in metric temporal logic\BBCQ\
\newblock In {\Bem Formal Modeling and Analysis of Timed Systems, 6th
  International Conference, {FORMATS} 2008, Saint Malo, France, September
  15-17, 2008. Proceedings}, \BPGS\ 1--13.

\bibitem[\protect\BCAY{Poggi, Lembo, Calvanese, De~Giacomo, Lenzerini,\ \BBA\
  Rosati}{Poggi et~al.}{2008}]{PLCD*08}
Poggi, A., Lembo, D., Calvanese, D., De~Giacomo, G., Lenzerini, M., \BBA\
  Rosati, R. \BBOP2008\BBCP.
\newblock \BBOQ Linking data to ontologies\BBCQ\
\newblock {\Bem J. Data Semant.}, {\Bem X}, 133--173.

\bibitem[\protect\BCAY{Rabinovich}{Rabinovich}{2014}]{DBLP:journals/corr/Rabinovich14}
Rabinovich, A. \BBOP2014\BBCP.
\newblock \BBOQ A proof of {K}amp's theorem\BBCQ\
\newblock {\Bem Logical Methods in Computer Science}, {\Bem 10\/}(1).

\bibitem[\protect\BCAY{Rodriguez{-}Muro, Kontchakov,\ \BBA\
  Zakharyaschev}{Rodriguez{-}Muro
  et~al.}{2013}]{DBLP:conf/semweb/Rodriguez-MuroKZ13}
Rodriguez{-}Muro, M., Kontchakov, R., \BBA\ Zakharyaschev, M. \BBOP2013\BBCP.
\newblock \BBOQ Ontology-based data access: Ontop of databases\BBCQ\
\newblock In Alani, H., Kagal, L., Fokoue, A., Groth, P.~T., Biemann, C.,
  Parreira, J.~X., Aroyo, L., Noy, N.~F., Welty, C., \BBA\ Janowicz, K.\BEDS,
  {\Bem The Semantic Web - {ISWC} 2013 - 12th International Semantic Web
  Conference, Sydney, NSW, Australia, October 21-25, 2013, Proceedings, Part
  {I}}, \lowercase{\BVOL}\ 8218 of {\Bem Lecture Notes in Computer Science},
  \BPGS\ 558--573. Springer.

\bibitem[\protect\BCAY{Rotman}{Rotman}{1999}]{rotman1999introduction}
Rotman, J.~J. \BBOP1999\BBCP.
\newblock {\Bem An introduction to the theory of groups}.
\newblock Springer-Verlag, New York; Berlin; Heidelberg [etc.].

\bibitem[\protect\BCAY{Ryzhikov, Walega,\ \BBA\ Zakharyaschev}{Ryzhikov
  et~al.}{2019}]{DBLP:conf/ijcai/RyzhikovWZ19}
Ryzhikov, V., Walega, P.~A., \BBA\ Zakharyaschev, M. \BBOP2019\BBCP.
\newblock \BBOQ Data complexity and rewritability of ontology-mediated queries
  in metric temporal logic under the event-based semantics\BBCQ\
\newblock In Kraus, S.\BED, {\Bem Proceedings of the Twenty-Eighth
  International Joint Conference on Artificial Intelligence, {IJCAI} 2019,
  Macao, China, August 10-16, 2019}, \BPGS\ 1851--1857. ijcai.org.

\bibitem[\protect\BCAY{Sch{\"{u}}tzenberger}{Sch{\"{u}}tzenberger}{1965}]{DBLP:journals/iandc/Schutzenberger65a}
Sch{\"{u}}tzenberger, M.~P. \BBOP1965\BBCP.
\newblock \BBOQ On finite monoids having only trivial subgroups\BBCQ\
\newblock {\Bem Inf. Control.}, {\Bem 8\/}(2), 190--194.

\bibitem[\protect\BCAY{{Shepherdson}}{{Shepherdson}}{1959}]{5392614}
{Shepherdson}, J.~C. \BBOP1959\BBCP.
\newblock \BBOQ The reduction of two-way automata to one-way automata\BBCQ\
\newblock {\Bem IBM J. of Research and Development}, {\Bem 3\/}(2), 198--200.

\bibitem[\protect\BCAY{Stern}{Stern}{1985}]{DBLP:journals/iandc/Stern85}
Stern, J. \BBOP1985\BBCP.
\newblock \BBOQ Complexity of some problems from the theory of automata\BBCQ\
\newblock {\Bem Inf. Control.}, {\Bem 66\/}(3), 163--176.

\bibitem[\protect\BCAY{Stockmeyer\ \BBA\ Meyer}{Stockmeyer\ \BBA\
  Meyer}{1973}]{DBLP:conf/stoc/StockmeyerM73}
Stockmeyer, L.~J.\BBACOMMA\  \BBA\ Meyer, A.~R. \BBOP1973\BBCP.
\newblock \BBOQ Word problems requiring exponential time: Preliminary
  report\BBCQ\
\newblock In Aho, A.~V., Borodin, A., Constable, R.~L., Floyd, R.~W., Harrison,
  M.~A., Karp, R.~M., \BBA\ Strong, H.~R.\BEDS, {\Bem Proceedings of the 5th
  Annual {ACM} Symposium on Theory of Computing, April 30 - May 2, 1973,
  Austin, Texas, {USA}}, \BPGS\ 1--9. {ACM}.

\bibitem[\protect\BCAY{Straubing}{Straubing}{1994}]{Straubing94}
Straubing, H. \BBOP1994\BBCP.
\newblock {\Bem Finite Automata, Formal Logic, and Circuit Complexity}.
\newblock Birkhauser Verlag.

\bibitem[\protect\BCAY{Tahrat, Braun, Artale, Gario,\ \BBA\ Ozaki}{Tahrat
  et~al.}{2020}]{DBLP:conf/dlog/TahratBAGO20}
Tahrat, S., Braun, G.~A., Artale, A., Gario, M., \BBA\ Ozaki, A.
  \BBOP2020\BBCP.
\newblock \BBOQ Automated reasoning in temporal {DL}-{L}ite (extended
  abstract)\BBCQ\
\newblock In Borgwardt, S.\BBACOMMA\  \BBA\ Meyer, T.\BEDS, {\Bem Proceedings
  of the 33rd International Workshop on Description Logics {(DL} 2020)
  co-located with the 17th International Conference on Principles of Knowledge
  Representation and Reasoning {(KR} 2020), Online Event [Rhodes, Greece],
  September 12th to 14th, 2020}, \lowercase{\BVOL}\ 2663 of {\Bem {CEUR}
  Workshop Proceedings}. CEUR-WS.org.

\bibitem[\protect\BCAY{{Tena Cucala}, Walega, {Cuenca Grau},\ \BBA\
  Kostylev}{{Tena Cucala} et~al.}{2021}]{DBLP:conf/aaai/CucalaWGK21}
{Tena Cucala}, D.~J., Walega, P.~A., {Cuenca Grau}, B., \BBA\ Kostylev, E.~V.
  \BBOP2021\BBCP.
\newblock \BBOQ Stratified negation in datalog with metric temporal
  operators\BBCQ\
\newblock In {\Bem Thirty-Fifth {AAAI} Conference on Artificial Intelligence,
  {AAAI} 2021, Thirty-Third Conference on Innovative Applications of Artificial
  Intelligence, {IAAI} 2021, The Eleventh Symposium on Educational Advances in
  Artificial Intelligence, {EAAI} 2021, Virtual Event, February 2-9, 2021},
  \BPGS\ 6488--6495. {AAAI} Press.

\bibitem[\protect\BCAY{Thompson}{Thompson}{1968}]{thompson1968}
Thompson, J.~G. \BBOP1968\BBCP.
\newblock \BBOQ Nonsolvable finite groups all of whose local subgroups are
  solvable\BBCQ\
\newblock {\Bem Bull. Amer. Math. Soc.}, {\Bem 74\/}(3), 383--437.

\bibitem[\protect\BCAY{Ullman\ \BBA\ Gelder}{Ullman\ \BBA\
  Gelder}{1988}]{DBLP:journals/algorithmica/UllmanG88}
Ullman, J.~D.\BBACOMMA\  \BBA\ Gelder, A.~V. \BBOP1988\BBCP.
\newblock \BBOQ Parallel complexity of logical query programs\BBCQ\
\newblock {\Bem Algorithmica}, {\Bem 3}, 5--42.

\bibitem[\protect\BCAY{van~der Meyden}{van~der
  Meyden}{2000}]{DBLP:journals/ijfcs/Meyden00}
van~der Meyden, R. \BBOP2000\BBCP.
\newblock \BBOQ Predicate boundedness of linear monadic datalog is in
  {PSPACE}\BBCQ\
\newblock {\Bem Int. J. Found. Comput. Sci.}, {\Bem 11\/}(4), 591--612.

\bibitem[\protect\BCAY{Vardi}{Vardi}{1988}]{DBLP:conf/pods/Vardi88}
Vardi, M.~Y. \BBOP1988\BBCP.
\newblock \BBOQ Decidability and undecidability results for boundedness of
  linear recursive queries\BBCQ\
\newblock In Edmondson{-}Yurkanan, C.\BBACOMMA\  \BBA\ Yannakakis, M.\BEDS,
  {\Bem Proceedings of the Seventh {ACM} {SIGACT-SIGMOD-SIGART} Symposium on
  Principles of Database Systems, March 21-23, 1988, Austin, Texas, {USA}},
  \BPGS\ 341--351. {ACM}.

\bibitem[\protect\BCAY{Vardi}{Vardi}{1989}]{10.1016/0020-0190(89)90205-6}
Vardi, M.~Y. \BBOP1989\BBCP.
\newblock \BBOQ A note on the reduction of two-way automata to one-way
  atuomata\BBCQ\
\newblock {\Bem Inf. Process. Lett.}, {\Bem 30\/}(5), 261–264.

\bibitem[\protect\BCAY{Vardi}{Vardi}{2007}]{DBLP:books/el/07/Vardi07}
Vardi, M.~Y. \BBOP2007\BBCP.
\newblock \BBOQ Automata-theoretic techniques for temporal reasoning\BBCQ\
\newblock In Blackburn, P., van Benthem, J. F. A.~K., \BBA\ Wolter, F.\BEDS,
  {\Bem Handbook of Modal Logic}, \lowercase{\BVOL}~3 of {\Bem Studies in logic
  and practical reasoning}, \BPGS\ 971--989. North-Holland.

\bibitem[\protect\BCAY{Vardi\ \BBA\ Wolper}{Vardi\ \BBA\
  Wolper}{1986}]{VardiW86}
Vardi, M.~Y.\BBACOMMA\  \BBA\ Wolper, P. \BBOP1986\BBCP.
\newblock \BBOQ An automata-theoretic approach to automatic program
  verification (preliminary report)\BBCQ\
\newblock In {\Bem Proc.\ of the Symposium on Logic in Computer Science
  (LICS'86)}, \BPGS\ 332--344.

\bibitem[\protect\BCAY{Walega, {Cuenca Grau}, Kaminski,\ \BBA\ Kostylev}{Walega
  et~al.}{2020a}]{DBLP:conf/kr/WalegaGKK20}
Walega, P.~A., {Cuenca Grau}, B., Kaminski, M., \BBA\ Kostylev, E.~V.
  \BBOP2020a\BBCP.
\newblock \BBOQ Datalogmtl over the integer timeline\BBCQ\
\newblock In Calvanese, D., Erdem, E., \BBA\ Thielscher, M.\BEDS, {\Bem
  Proceedings of the 17th International Conference on Principles of Knowledge
  Representation and Reasoning, {KR} 2020, Rhodes, Greece, September 12-18,
  2020}, \BPGS\ 768--777.

\bibitem[\protect\BCAY{Walega, {Cuenca Grau}, Kaminski,\ \BBA\ Kostylev}{Walega
  et~al.}{2020b}]{DBLP:conf/ijcai/WalegaGKK20}
Walega, P.~A., {Cuenca Grau}, B., Kaminski, M., \BBA\ Kostylev, E.~V.
  \BBOP2020b\BBCP.
\newblock \BBOQ Tractable fragments of datalog with metric temporal
  operators\BBCQ\
\newblock In Bessiere, C.\BED, {\Bem Proceedings of the Twenty-Ninth
  International Joint Conference on Artificial Intelligence, {IJCAI} 2020},
  \BPGS\ 1919--1925. ijcai.org.

\bibitem[\protect\BCAY{Wang, Hu, Walega,\ \BBA\ {Cuenca Grau}}{Wang
  et~al.}{2022}]{DBLP:journals/corr/abs-2201-04596}
Wang, D., Hu, P., Walega, P.~A., \BBA\ {Cuenca Grau}, B. \BBOP2022\BBCP.
\newblock \BBOQ Meteor: Practical reasoning in datalog with metric temporal
  operators\BBCQ\
\newblock {\Bem CoRR}, {\Bem abs/2201.04596}.

\bibitem[\protect\BCAY{Xiao, Calvanese, Kontchakov, Lembo, Poggi, Rosati,\
  \BBA\ Zakharyaschev}{Xiao et~al.}{2018}]{DBLP:conf/ijcai/XiaoCKLPRZ18}
Xiao, G., Calvanese, D., Kontchakov, R., Lembo, D., Poggi, A., Rosati, R.,
  \BBA\ Zakharyaschev, M. \BBOP2018\BBCP.
\newblock \BBOQ Ontology-based data access: {A} survey\BBCQ\
\newblock In Lang, J.\BED, {\Bem Proceedings of the Twenty-Seventh
  International Joint Conference on Artificial Intelligence, {IJCAI} 2018, July
  13-19, 2018, Stockholm, Sweden.}, \BPGS\ 5511--5519. ijcai.org.

\bibitem[\protect\BCAY{Xiao, Ding, Cogrel,\ \BBA\ Calvanese}{Xiao
  et~al.}{2019}]{DBLP:journals/dint/XiaoDCC19}
Xiao, G., Ding, L., Cogrel, B., \BBA\ Calvanese, D. \BBOP2019\BBCP.
\newblock \BBOQ Virtual knowledge graphs: An overview of systems and use
  cases\BBCQ\
\newblock {\Bem Data Intell.}, {\Bem 1\/}(3), 201--223.

\bibitem[\protect\BCAY{Xiao, Lanti, Kontchakov, Komla{-}Ebri, Kalayci, Ding,
  Corman, Cogrel, Calvanese,\ \BBA\ Botoeva}{Xiao
  et~al.}{2020}]{DBLP:conf/semweb/XiaoLKKKDCCCB20}
Xiao, G., Lanti, D., Kontchakov, R., Komla{-}Ebri, S., Kalayci, E.~G., Ding,
  L., Corman, J., Cogrel, B., Calvanese, D., \BBA\ Botoeva, E. \BBOP2020\BBCP.
\newblock \BBOQ The virtual knowledge graph system {O}ntop\BBCQ\
\newblock In Pan, J.~Z., Tamma, V. A.~M., d'Amato, C., Janowicz, K., Fu, B.,
  Polleres, A., Seneviratne, O., \BBA\ Kagal, L.\BEDS, {\Bem The Semantic Web -
  {ISWC} 2020 - 19th International Semantic Web Conference, Athens, Greece,
  November 2-6, 2020, Proceedings, Part {II}}, \lowercase{\BVOL}\ 12507 of
  {\Bem Lecture Notes in Computer Science}, \BPGS\ 259--277. Springer.

\end{thebibliography}

\end{document}